%% file: main.tex
\documentclass[sigconf,screen]{acmart}

\copyrightyear{2022}
\acmYear{2022}
\setcopyright{acmcopyright}\acmConference[ASPLOS '22]{27th ACM International Conference on Architectural Support for Programming Languages and Operating Systems}{February 28-March 4, 2022}{Lausanne, Switzerland}
\acmBooktitle{27th ACM International Conference on Architectural Support for Programming Languages and Operating Systems (ASPLOS '22), February 28-March 4, 2022, Lausanne, Switzerland}
\acmPrice{15.00}
\acmDOI{10.1145/3503222.3507734}
\acmISBN{978-1-4503-9205-1/22/02}

\author{Umang Mathur}
\orcid{0000-0002-7610-0660}
\affiliation{
\institution{National University of Singapore}
\country{Singapore}
}
\email{umathur@comp.nus.edu.sg}

\author{Andreas Pavlogiannis}
\orcid{0000-0002-8943-0722}
\affiliation{
\institution{Aarhus University}
\country{Denmark}
}
\email{pavlogiannis@cs.au.dk}

\author{Hünkar Can Tunç}
\orcid{0000-0001-9125-8506}
\affiliation{
\institution{Aarhus University}
\country{Denmark}
}
\email{tunc@cs.au.dk}

\author{Mahesh Viswanathan}
\orcid{}
\affiliation{
\institution{University of Illinois at Urbana-Champaign}
\country{USA}
}
\email{vmahesh@illinois.edu}

\newcommand{\asplossubmissionnumber}{476}
\usepackage[normalem]{ulem}

\input{packages}
\input{macros}

\renewcommand{\smallskip}{}

\includeversion{arxiv} 
\excludeversion{asplos} 

\begin{document}

\title{A Tree Clock Data Structure for Causal Orderings in Concurrent Executions}

\input{abstract}

\maketitle

\thispagestyle{empty}

\input{intro}
\input{prelims}
\input{treeclocks}

\input{applications}
\input{experiments}
\input{related_work}

\clearpage

\bibliographystyle{plainurl}
\bibliography{bibliography}

\clearpage

\appendix
\input{arxiv/arxiv_appendix}

\end{document}

%% file: packages.tex
\usepackage[utf8]{inputenc} 
\usepackage{amsmath, amsthm, amsfonts}
\usepackage{thmtools}
\usepackage[english]{babel}
\usepackage{pdfpages}
\usepackage{graphicx}
\usepackage{subcaption}
\usepackage{caption}
\usepackage{longtable}
\usepackage{booktabs}
\usepackage{tabu}
\usepackage[referable]{threeparttablex}
\usepackage{varwidth}
\usepackage{multirow}
\usepackage{multicol}
\usepackage{wrapfig}
\usepackage{lipsum}
\usepackage{stmaryrd}
\usepackage{float}
\usepackage{listings}

\usepackage{pifont}

\usepackage[ruled, linesnumbered,noend]{algorithm2e}

\usepackage{parskip}

\usepackage{hyperref}
\usepackage{thm-restate}
\usepackage[capitalise]{cleveref}

\crefname{figure}{Figure}{Figure}
\crefformat{footnote}{#2\footnotemark[#1]#3}
\usepackage{tikz}
\usepackage{tikz-qtree}
\usepackage{comment}
\usepackage{tikz}
\usetikzlibrary{fit,arrows,automata,shapes,decorations,decorations.markings,calc, matrix,decorations.pathmorphing, patterns,backgrounds}

\usepackage{bm}
\usepackage{pgfplots}
\usepackage{enumitem}

\usepackage{pbox}

\usepackage{appendix}
\usepackage{calc}
\usepackage{fp}
\usepackage{multirow}
\usepackage{pbox}

\usepackage{etoolbox}

\usepackage{arydshln}

\usepackage{microtype} 
\usepackage{xtab}
\usepackage{ltablex}
\usepackage{booktabs}
\usepackage{versions}

%% file: macros.tex
\newtheorem{lemma}{Lemma}

\theoremstyle{plain}
\newtheorem{remark}{Remark}

\def\treeclocksexamplecaling{0.93}

\newcommand{\Paragraph}[1]{\noindent{\bf #1}}
\newcommand{\SubParagraph}[1]{\smallskip\noindent{\em #1}}

\DeclareMathAlphabet{\mathpzc}{OT1}{pzc}{m}{it}

\mathchardef\mhyphen="2D % Define a "math hyphen"
\newcommand{\Nats}{\mathbb{N}}

\newcommand{\setpred}[2]{\{#1 \,|\, #2\}}

\newcommand{\tr}{\sigma}
\newcommand{\Trace}{\tr} %% Andreas
\newcommand{\lk}{\ell}
\newcommand{\Event}{e}
\newcommand{\ev}[1]{\langle #1\rangle}
\newcommand{\conf}{\asymp}
\newcommand{\Confl}[2]{#1 \conf #2}
\newcommand{\opfont}[1]{\texttt{#1}}
\newcommand{\acq}{\opfont{acq}}
\newcommand{\rel}{\opfont{rel}}
\newcommand{\rd}{\opfont{r}}
\newcommand{\wt}{\opfont{w}}

\newcommand{\sync}{\opfont{sync}}

\newcommand{\Read}{\rd}
\newcommand{\Write}{\wt}
\newcommand{\Acquire}{\acq}
\newcommand{\Release}{\rel}
\newcommand{\Sync}{\sync}
\newcommand{\LastWrite}{\operatorname{lw}}

\newcommand{\True}{\operatorname{True}}

\newcommand{\Thread}{t}
\newcommand{\ThreadOf}[1]{\tid(#1)}

\newcommand{\Variable}{\operatorname{Variable}}
\newcommand*{\Match}[2][]{
\ifx\\#1\\%
  \mathsf{match}(#2)
\else
  \mathsf{match}_{#1}(#2)
\fi
}
\newcommand{\acr}[1]{\textsf{#1}}
\newcommand{\acrtr}{\operatorname{\acr{tr}}}
\newcommand{\TO}{\operatorname{\acr{TO}}}
\newcommand{\SHB}{\operatorname{\acr{SHB}}}
\newcommand{\Maz}{\operatorname{\acr{MAZ}}}

\newcommand{\DC}{\operatorname{\acr{DC}}}
\newcommand{\HB}{\operatorname{\acr{HB}}}
\newcommand{\CP}{\operatorname{\acr{CP}}}
\newcommand{\WCP}{\operatorname{\acr{WCP}}}
\newcommand{\SDP}{\operatorname{\acr{SDP}}}
\newcommand{\WDP}{\operatorname{\acr{WDP}}}

\newcommand{\ord}[2]{\leq^{#1}_{\mathsf{#2}}}

\newcommand{\strictord}[2]{<^{#1}_{\mathsf{#2}}}
\newcommand{\unord}[2]{\parallel^{#1}_{\mathsf{#2}}}

\newcommand{\trord}[1]{\ord{#1}{\acrtr}}
\newcommand{\stricttrord}[1]{\strictord{#1}{\acrtr}}
\newcommand{\tho}[1]{\ord{#1}{\TO}}

\newcommand{\hb}[1]{\ord{#1}{\HB}}
\newcommand{\unordhb}[1]{\unord{#1}{\HB}}
\newcommand{\stricthb}[1]{\strictord{#1}{\HB}}
\newcommand{\shb}[1]{\ord{#1}{\SHB}}
\newcommand{\maz}[1]{\ord{#1}{\Maz}}

\newcommand{\Threads}[1]{\operatorname{Thrds}_{#1}}

\newcommand{\LocalTime}[2]{\operatorname{lTime}^{#1}(#2)}
\newcommand{\VectorTime}{\operatorname{V}}
\newcommand{\POTime}[2]{C^{#1}_{#2}}
\newcommand{\VectorClock}{\operatorname{VC}}
\newcommand{\TreeClock}{\operatorname{TC}}
\newcommand{\CSmaller}{\sqsubseteq}
\newcommand{\CJoin}{\sqcup}
\newcommand{\CAssign}[2]{[#1\to +#2]}
\newcommand{\GTree}{\operatorname{T}}
\newcommand{\Root}{\operatorname{root}}
\newcommand{\ThreadMap}{\operatorname{ThrMap}}
\newcommand{\TNodes}{\mathcal{V}}
\newcommand{\TEdges}{\mathcal{E}}
\newcommand{\ChildrenQ}{\operatorname{Chld}}

\newcommand{\Parent}{\operatorname{Prnt}}
\newcommand{\Stack}{\mathcal{S}}

\newcommand{\tid}{\operatorname{\mathsf{tid}}}
\newcommand{\clock}{\operatorname{clk}}
\newcommand{\pclock}{\operatorname{aclk}}

\newcommand{\DataStructure}{\operatorname{DS}}

\newcommand{\VTWork}{\operatorname{VTWork}}
\newcommand{\VCWork}{\operatorname{VCWork}}
\newcommand{\TCWork}{\operatorname{TCWork}}
\newcommand{\Time}{\mathcal{T}}

\SetAlFnt{\normalsize}
\SetInd{0.5em}{0.5em}

\SetKw{Break}{break}
\SetKwProg{MyFunction}{function}{}{}
\SetKwProg{MyRoutine}{routine}{}{}
\SetKwProg{MyProcedure}{procedure}{}{}
\SetKwFunction{FunctionCTInitialize}{{Init}}
\SetKwFunction{FunctionCTLessThan}{LessThan}
\SetKwFunction{FunctionCTJoin}{Join}
\SetKwFunction{FunctionVCJoin}{Join}
\SetKwFunction{FunctionCTSubRootJoin}{SubRootJoin}
\SetKwFunction{FunctionVCCopy}{Copy}
\SetKwFunction{FunctionCTMonotoneCopy}{MonotoneCopy}
\SetKwFunction{FunctionCTCopyCheckMonotone}{CopyCheckMonotone}
\SetKwFunction{FunctionCTVTime}{VectorTime}
\SetKwFunction{FunctionCTIncrement}{Increment}
\SetKwFunction{FunctionVCIncrement}{Increment}
\SetKwFunction{FunctionCTGet}{Get}
\SetKwFunction{RoutineTestDeque}{TestDeque}
\SetKwFunction{RoutineUpdatedNodesForJoin}{getUpdatedNodesJoin}
\SetKwFunction{RoutineUpdatedNodesForCopy}{getUpdatedNodesCopy}
\SetKwFunction{RoutineDetachNodes}{detachNodes}
\SetKwFunction{RoutineAttachNodesJoin}{attachNodesJoin}
\SetKwFunction{RoutineAttachNodesCopy}{attachNodesCopy}
\SetKwFunction{RoutineAttachNodes}{attachNodes}
\SetKwFunction{RoutinePushChild}{pushChild}
\SetKwFunction{RoutineInsertChild}{insertChild}
\SetKwFunction{ProcInitialize}{initialize}
\SetKwFunction{ProcAcquire}{acquire}
\SetKwFunction{ProcRelease}{release}
\SetKwFunction{ProcWrite}{write}
\SetKwFunction{ProcRead}{read}

\newcommand{\CTC}{\mathbb{C}}

\newcommand{\CTL}{\mathbb{L}}

\newcommand{\CTLW}{\mathbb{LW}}
\newcommand{\CTR}{\mathbb{R}}

\newcommand{\LWT}{\operatorname{LW}}

\newcommand{\LastReads}{\operatorname{LRDs}}

\newcommand{\Operation}{\operatorname{op}}
\newcommand{\NumEvents}{n}
\newcommand{\NumThreads}{k}

\newcommand{\fasttrack}{\textsc{FastTrack}\xspace}

\usetikzlibrary{arrows.meta,calc,decorations.markings,math,arrows.meta}

\preto\tabular{\setcounter{magicrownumbers}{0}}
\newcounter{magicrownumbers}

\pgfdeclarelayer{bg}    % declare background layer
\pgfsetlayers{bg,main}  % set the order of the layers (main is the standard layer)

\def \darkred {black!20!red}
\def \darkgreen {black!40!green}
\SetKw{Continue}{continue}

\SetCommentSty{mycommfont}

\newlist{compactenum}{enumerate}{3} % 3 is max-depth
\setlist[compactenum]{label=\arabic*., nosep,leftmargin=*}
\crefname{compactenumi}{Item}{Items}

\makeatletter
\newcommand*{\centerfloat}{%
  \parindent \z@
  \leftskip \z@ \@plus 1fil \@minus \textwidth
  \rightskip\leftskip
  \parfillskip \z@skip}
\makeatother

\newcommand{\camera}[1]{\textcolor{black}{#1}}
\newcommand{\hcomment}[1]{\textcolor{black}{#1}}

\input{exec_macros}

%% file: exec_macros.tex
\newcommand{\execution}[2]{
\scalebox{1}{
  \begin{tikzpicture}%
    \foreach \x in {1,...,#1}
    \node[right] at (1.5*\x+0.2,0.25) {$\Thread_{\x}$};
    \draw (1.2,0) -- (#1*1.5+1.3,0);%
    \pgfmathsetmacro{\y}{1};%
    #2%
    \draw (1.2,0) -- (1.2,-0.4*\y);%
    \draw (#1*1.5+1.3,0) -- (#1*1.5+1.3,-0.4*\y);%
    \foreach \x in {2,...,#1}
    \draw (1.2,-0.4*\y) -- (#1*1.5+1.3,-0.4*\y);%
  \end{tikzpicture}
}
}

\newcommand{\figev}[2]{
\pgfmathsetmacro{\y}{\y+1};
\pgfmathsetmacro{\y}{\y-1};
\node [left] at (1.25,-0.4*\y)  {\pgfmathprintnumber{\y}};%
\node at (#1*1.5 + 0.45,-0.4*\y) { #2 };%
\pgfmathsetmacro{\y}{\y+1};
}

\newcommand{\orderedgewithlabel}[8]{
  \draw (#1*1.5 + 0.45 + #3, -0.4*#2) edge[-{Latex[length=2mm, width=2mm]},  ->,>=stealth] node[anchor=center, sloped, #8] { #7} (#4*1.5 + 0.45 + #6, -0.4*#5);
}

\newcommand{\orderdashededgewithlabel}[8]{
  \draw (#1*1.5 + 0.45 + #3, -0.4*#2) edge[-{Latex[length=2mm, width=2mm]},  ->,>=stealth, dashed] node[anchor=center, sloped, #8] { #7} (#4*1.5 + 0.45 + #6, -0.4*#5);
}

\newcommand{\treethr}[1]{\Thread_{#1}}

\newcommand{\defaultScaling}{0.9}
\newcommand{\defaultThreadWidth}{1.2}
\newcommand{\defaultEventHeight}{0.55}
\newcommand{\defaultTIDXOffset}{0.2}
\newcommand{\defaultLeftBoundaryXOffset}{1.2}
\newcommand{\defaultTIDYOffset}{0.25}
\newcommand{\defaultOpXOffset}{0.25}

\newcommand{\executionfull}[9]{
\scalebox{#3}{
  \begin{tikzpicture}
    \foreach \x in {1,...,#1}
    \node[right] at (#4*\x+#6, #8) {$t_{\x}$}; % Print tids
    \draw (#7,0) -- (#1*#4+#7,0); % Draw top boundary
    \pgfmathsetmacro{\y}{1};
    #2 % Draw the execution
    \draw (#7,0) -- (#7,-#5*\y); % Draw left boundary
    \draw (#1*#4+#7,0) -- (#1*#4+#7,-#5*\y); % Draw right boundary
    \draw (#7,-#5*\y) -- (#1*#4+#7,-#5*\y); % Draw  bottom boundary
    \ifthenelse{#9 = 1}{
      \foreach \x in {2,...,#1}
      \draw[dashed] (#4*\x+#7-#4,0) -- (#4*\x+#7-#4,-#5*\y); % Draw inter-thread boundaries
    }{}
  \end{tikzpicture}
}
}

\newcommand{\figevfull}[9]{
\ifthenelse{#7 = 1}{
  \ifthenelse{#8 = -1}{
    \node [left] at (#5,(-#4*\y)  {\pgfmathprintnumber{\y}};%
  }{
    \node [left] at (#5,(-#4*\y)  {#9};%
  }
}{}
\node at (#1*#3 + #6,(-#4*\y) {$ #2 $};%
\pgfmathsetmacro{\y}{\y+1};
}

\newcommand{\executioncustom}[2]{
\executionfull
  {#1}
  {#2}
  {\defaultScaling}
  {\defaultThreadWidth}
  {\defaultEventHeight}
  {\defaultTIDXOffset}
  {\defaultLeftBoundaryXOffset}
  {\defaultTIDYOffset}
  {1}
}

\newcommand{\figevcustom}[2]{
\figevfull
  {#1}
  {#2}
  {\defaultThreadWidth}
  {\defaultEventHeight}
  {\defaultLeftBoundaryXOffset}
  {2.5*\defaultOpXOffset}
  {1}
  {-1}
  {}
}

%% file: abstract.tex
%!TEX root=main.tex
%
%% 2012 ACM Computing Classification System (CSS) concepts
%% Generate at 'http://dl.acm.org/ccs/ccs.cfm'.

\ccsdesc[500]{Software and its engineering~Software verification and validation}
\ccsdesc[300]{Theory of computation~Theory and algorithms for application domains}
\ccsdesc[300]{Theory of computation~Program analysis}

\begin{abstract}
\vspace{0.2in}
Dynamic techniques are a scalable and effective way to analyze concurrent programs.
Instead of analyzing all behaviors of a program, these techniques detect errors by focusing on a single program execution.
Often a crucial step in these techniques is to define a causal ordering between events in the execution, which is then computed using \emph{vector clocks}, a simple data structure that stores logical times of threads.
The two basic operations of vector clocks, namely join and copy, require $\Theta(\NumThreads)$ time, where $\NumThreads$ is the number of threads. Thus they are a computational bottleneck when $\NumThreads$ is large.

In this work, we introduce \emph{tree clocks}, a new data structure that replaces vector clocks for computing causal orderings in program executions.
Joining and copying tree clocks takes time that is roughly proportional to the number of entries being modified, and hence the two operations do not suffer the a-priori $\Theta(\NumThreads)$ cost per application.
We show that when used to compute the classic happens-before ($\HB$) partial order, 
tree clocks are \emph{optimal}, in the sense that no other data structure can lead to smaller asymptotic running time.
Moreover, we demonstrate that tree clocks can be used to compute other partial orders, such as schedulable-happens-before ($\SHB$) and the standard Mazurkiewicz ($\Maz$) partial order, and thus are a versatile data structure.
Our experiments show that just by replacing vector clocks with tree clocks, the computation becomes from \camera{$2.02 \times$} faster ($\Maz$) to \camera{$2.66 \times$} ($\SHB$) and \camera{$2.97 \times$} ($\HB$) on average per benchmark.
These results illustrate that tree clocks have the potential to become a standard data structure with wide applications in concurrent analyses.
%\ucomment{Add speedup of HB in the abstract?}
%\ucomment{Use $\times$ instead of X for speedup in the entire paper.}
\end{abstract}

\keywords{concurrency, happens-before, vector clocks, dynamic analyses}

%% file: intro.tex
%!TEX root=main.tex

\section{Introduction}\label{sec:intro}

The analysis of concurrent programs is one of the major challenges in formal methods,
due to the non-determinism of inter-thread communication.
The large space of communication interleavings poses a significant challenge to the programmer, as intended invariants can be broken by unexpected communication patterns.
The subtlety of these patterns also makes verification a demanding task, as exposing a bug requires searching an exponentially large space~\cite{Musuvathi08}.
Consequently, significant efforts are made towards understanding and detecting concurrency bugs efficiently
\cite{lpsz08,Shi10,Tu2019,SoftwareErrors2009,boehmbenign2011,Farchi03}.

\SubParagraph{Dynamic analyses and partial orders.}
One popular approach to the scalability problem of concurrent program verification is dynamic analysis~\cite{Mattern89,Pozniansky03,Flanagan09,Mathur2020b}.
Such techniques have the more modest goal of discovering faults by analyzing program executions instead of whole programs.
Although this approach cannot prove the absence of bugs, it is far more scalable than static analysis and typically makes sound reports of errors.
These advantages have rendered dynamic analyses a very effective and widely used approach to error detection in concurrent programs.

The first step in virtually all techniques that analyze concurrent executions is to establish 
a causal ordering between the events of the execution.
Although the notion of causality varies with the application, 
its transitive nature makes it naturally expressible as a  
\emph{partial order} between these events.
One prominent example is the Mazurkiewicz partial order ($\Maz$), 
which often serves as the canonical way to represent concurrent traces~\cite{Mazurkiewicz87,Bertoni1989} (aka Shasha-Snir traces~\cite{Shasha1988}).
Another vastly common partial order is 
Lamport's \emph{happens-before} ($\HB$)~\cite{Lamport78}, initially proposed
in the context of distributed systems~\cite{schwarz1994detecting}.
In the context of testing multi-threaded programs,
partial orders play a crucial role in dynamic race detection techniques,
and have been thoroughly exploited to
explore trade-offs between soundness, completeness, 
and running time of the underlying analysis.
Prominent examples include the widespread use of $\HB$~\cite{Itzkovitz1999,Flanagan09,Pozniansky03,threadsanitizer,Elmas07},
schedulably-happens-before ($\SHB$)~\cite{Mathur18},
causally-precedes ($\CP$)~\cite{Smaragdakis12},
weak-causally-precedes ($\WCP$)~\cite{Kini17},
doesn't-commute ($\DC$)~\cite{Roemer18}, and
 strong/weak-dependently-pre\-cedes  ($\SDP$/$\WDP$)~\cite{Genc19}, 
$\textsf{M2}$~\cite{Pavlogiannis2020} and SyncP~\cite{Mathur21}.
Beyond race detection, partial orders are often
employed to detect and reproduce other concurrency bugs such as
atomicity violations~\cite{Flanagan2008,Biswas14,Mathur2020},
deadlocks~\cite{Samak2014,Sulzmann2018},
and other concurrency vulnerabilities~\cite{Yu2021}.

\input{figures/motivating}

\SubParagraph{Vector clocks in dynamic analyses.}
Often, the computational task of determining the partial 
ordering between events of an execution is achieved using a simple data 
structure called \emph{vector clock}.
Informally, a vector clock $\CTC$ is an integer array indexed 
by the processes/threads in the execution, and succinctly encodes
the knowledge of a process about the whole system.
For vector clock $\CTC_{\Thread_1}$ associated with thread $\Thread_1$, 
if $\CTC_{\Thread_1}(\Thread_2) = i$ then it means that the 
latest event of $\Thread_1$ is ordered after the first 
$i$ events of thread $\Thread_2$ in the partial order. 
Vector clocks, thus seamlessly capture a partial order, 
with the point-wise ordering of the vector timestamps of two events 
capturing the ordering between the events with respect to the partial order of interest. 
For this reason, vector clocks are instrumental in 
computing the $\HB$ parial order efficiently~\cite{Mattern89,fidge1988timestamps,Fidge91},
and are ubiquitous in the efficient implementation of 
analyses based on partial orders even beyond $\HB$~\cite{Flanagan09,Mathur18,Kini17,Roemer18,Mathur2020,Samak2014,Sulzmann2018,Kulkarni2021}.

The fundamental operation on vector clocks is the pointwise \emph{join} 
$\CTC_{\Thread_1}\gets \CTC_{\Thread_1} \CJoin \CTC_{\Thread_2}$.
This occurs whenever there is a causal ordering from thread $\Thread_2$ to $\Thread_1$.
Operationally, a join is performed by updating $\CTC_{\Thread_1}(\Thread)\gets \max( \CTC_{\Thread_1}(\Thread), \CTC_{\Thread_2}(\Thread))$ for every thread $\Thread$,
and captures the transitivity of causal orderings:~as $\Thread_1$ learns about $\Thread_2$, it also learns about other threads $\Thread$ that $\Thread_2$ knows about.
Note that if $\Thread_1$ is aware of a later event of $\Thread$, this operation is vacuous.
With $\NumThreads$ threads, a vector clock join takes $\Theta(\NumThreads)$ time, 
and can quickly become a bottleneck even in systems with moderate $\NumThreads$.
This motivates the following question:~is it possible to speed up join 
operations by proactively avoiding vacuous updates?
The challenge in such a task comes from the efficiency of the join operation itself---since 
it only requires linear time in the size of the vector, any improvement must operate in sub-linear time, i.e., not even touch certain entries of the vector clock. 
We illustrate this idea on a concrete example, and present the key insight in this work.

\Paragraph{Motivating example.}
Consider the example in \cref{fig:motivating}. 
It shows a partial trace from a concurrent system with 
6 threads, along with the vector timestmamps at each event.
When event $\Event_2$ is ordered before event $\Event_3$ 
due to synchronization, the vector clock $\CTC_{\Thread_2}$ of 
$\Thread_2$ is joined with that of $\CTC_{\Thread_1}$, i.e., 
the $\Thread_j$-th entry of $\CTC_{\Thread_1}$ is updated 
to the maximum of $\CTC_{\Thread_1}(\Thread_j)$ and $\CTC_{\Thread_2}(\Thread_j)$\footnote{\camera{As with many presentations of dynamic analyses using vector clocks~\cite{Itzkovitz1999}, we assume that the \emph{local} entry of a thread's clock increments by $1$ after each event it performs. Hence, in \cref{fig:motivating}, the $\Thread_1$-th entry of $\CTC_{\Thread_1}$ increases from $27$ to $28$ after $\Event_1$ is performed.
% Note that performing an operation results in incrementing the timestamp of the performing thread by one in its own clock. Thus, the timestamp of ${\Thread_1}$ in $\CTC_{\Thread_1}$ is incremented by one after the join operation.
}}.
Now assume that thread $\Thread_2$ has learned of the current 
times of threads $\Thread_3$, $\Thread_4$, $\Thread_5$ 
and $\Thread_6$ via thread $\Thread_3$.
Since the $\Thread_3$-th component of the vector timestamp of event 
$\Event_1$ is larger than the corresponding component of event $\Event_2$, 
$\Thread_1$ cannot possibly learn any \emph{new} information about threads
 $\Thread_4$, $\Thread_5$, and $\Thread_6$ through the join performed at event $\Event_3$. 
Hence the naive pointwise updates will be redundant for the indices $j=\{3,4,5,6 \}$.
Unfortunately, the flat structure of vector clocks
is not amenable to such reasoning and cannot avoid these redundant operations.

To alleviate this problem, we introduce a new hierarchical tree-like data structure for maintaining vector times called a \emph{tree clock}.
The nodes of the tree encode local clocks, just like entries in a vector clock.
In addition, the structure of the tree naturally captures which clocks have been learned transitively via intermediate threads.
\cref{fig:motivating} (right) depicts a (simplified) tree clock encoding the vector times of $\CTC_{\Thread_2}$.
The subtree rooted at thread $\Thread_3$ encodes the fact that $\Thread_2$ has learned about the current times of $\Thread_4$, $\Thread_5$ and $\Thread_6$ \emph{transitively}, via $\Thread_3$.
To perform the join operation $\CTC_{\Thread_1}\gets \CTC_{\Thread_1} \CJoin \CTC_{\Thread_2}$, we start from the root of  $\CTC_{{\Thread}_2}$, and traverse the tree as follows. Given a current node $u$, we proceed to the children of $u$ \emph{if and only if}
$u$ represents the time of a thread that is not known to $\Thread_1$.
Hence, in the example, the join operation will now access only the light-gray area of the tree, and thus compute the join without accessing the whole tree, resulting in a \emph{sublinear running time} of the join operation.

The above principle, which we call \emph{direct monotonicity}
is one of two key ideas exploited by tree clocks; the other being \emph{indirect monotonicity}.
The key technical challenge in developing the tree clock data structure lies in
(i)~using direct and indirect monotonicity to perform efficient updates, and
(ii)~perform these updates such that direct and indirect monotonicity 
are preserved for future operations.
\cref{subsec:intuition} illustrates the intuition behind these two principles in depth.

\Paragraph{Contributions.}
Our contributions are as follows.
\begin{compactenum}
\item We introduce \emph{tree clock}, a new data structure for maintaining logical times in concurrent executions.
In contrast to the flat structure of the traditional vector clocks, the dynamic hierarchical structure of tree clocks naturally captures ad-hoc communication patterns between processes.
In turn, this allows for join and copy operations that run in \emph{sublinear time}.
As a data structure, tree clocks offer high versatility as they can be used in computing many different ordering relations.
\item We prove that tree clocks are an \emph{optimal data structure} for computing $\HB$, in the sense that, \emph{for every input trace}, the total computation time cannot be improved (asymptotically) by replacing tree clocks with any other data structure.
On the other hand, vector clocks do not enjoy this property.
%\item We prove that, for computing the $\HB$ partial order, tree clocks offer an optimality guarantee we call \emph{vector-time} (or \emph{vt-}) \emph{optimality}.
%Intuitively, vt-optimality guarantees that tree clocks are an \emph{optimal data structure} for $\HB$, in the sense that, \emph{for every input trace}, the total computation time cannot be improved (asymptotically) by replacing tree clocks with any other data structure.
%On the other hand, vector clocks do not enjoy this property.
\item We illustrate the versatility of tree clocks by presenting tree clock-based 
algorithms for the $\Maz$  and  $\SHB$ partial orders.
%\item \Andreas{Tree clocks for distributed?}
\item We perform a large-scale experimental evaluation of the tree clock data structure for computing the $\Maz$, $\SHB$ and $\HB$ partial orders, and compare its performance against the standard vector clock data structure.
Our results show that just by replacing vector clocks with tree clocks, the computation becomes up to $2.97 \times$ faster on average.
Given our experimental results, we believe that replacing vector clocks by tree clocks in partial order-based algorithms can lead to significant improvements on many applications.

\end{compactenum}

%\Paragraph{Organization.}
%% \ucomment{Move this somewhere else.}
%The rest of the paper is organized as follows.
%\begin{compactenum}
%\item In \cref{sec:prelims}, we discuss some preliminaries about execution traces, the $\HB$ partial order, vector timestamps and vector clocks.
%\item In \cref{sec:tree_clocks} we present the tree clock data structure and illustrate it in a few examples.
%\item In \cref{sec:race_detection} we present $\FastHB$, a new $\HB$ race detector.
%Compared to existing race detectors $\FastHB$
%(i)~replaces vector clocks with tree clocks for computing $\HB$
%\ucomment{, and (ii)~uses write- and read-epoch optimizations,}
%to offer optimality guarantees.
%\item In \cref{sec:race_prediction} we show how tree clocks can replace vector clocks in a number of race predictors,
%such as $\SHB$, $\WCP$ and $\DC$.
%\item In \cref{sec:experiments} we perform an experimental evaluation of the above techniques and measure their speedup over baseline methods.
%\end{compactenum}

%% file: figures/motivating.tex
%!TEX root=../main.tex
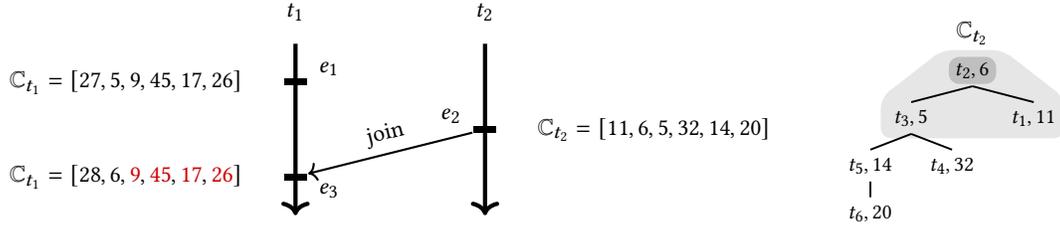
\begin{figure*}[t]
\centering
\begin{tikzpicture}[thick,
pre/.style={<-,shorten >= 2pt, shorten <=2pt, very thick},
post/.style={->,shorten >= 2pt, shorten <=2pt,  very thick},
seqtrace/.style={->, line width=2},
und/.style={very thick, draw=gray},
event/.style={rectangle, minimum height=0.8mm, minimum width=3mm, fill=black!100,  line width=1pt, inner sep=0},
virt/.style={circle,draw=black!50,fill=black!20, opacity=0},
sibling distance=10pt,
level distance=20pt,
]
\tikzset{edge from parent/.append style={ thick}}
\pgfdeclarelayer{background}
\pgfdeclarelayer{foreground}
\pgfsetlayers{background,main,foreground}

\newcommand{\xdisposition}{10}
\newcommand{\ydisposition}{-0.6}
\newcommand{\xtstep}{0.5}
\newcommand{\ytstep}{0.2}
\newcommand{\xstep}{2.8}
\newcommand{\ystep}{-0.7}
\newcommand{\grayDark}{gray!50}
\newcommand{\grayLight}{gray!20}
\newcommand{\halfyOuter}{0.30}
\newcommand{\halfxOuter}{0.45}
\newcommand{\halfyInner}{0.20}
\newcommand{\halfxInner}{0.35}

\begin{scope}[scale=0.9]
\node [] (t1start) at (0*\xstep, 0*\ystep) {};
\node [] (t1end) at (0*\xstep, 4*\ystep) {};

\node [] (t2start) at (1*\xstep, 0*\ystep) {};
\node [] (t2end) at (1*\xstep, 4*\ystep) {};

\draw[->, ultra thick] (t1start) to (t1end);
\draw[->, ultra thick] (t2start) to (t2end);

\node[node distance=3mm, above of=t1start] {\normalsize $\Thread_1$};
\node[node distance=3mm, above of=t2start] {\normalsize $\Thread_2$};

\node[event] (e1) at (0*\xstep, 1*\ystep) {};
\node[event] (e3) at (0*\xstep, 3*\ystep) {};
\node[event] (e2) at (1*\xstep, 2*\ystep) {};

\node[]  at (0*\xstep + \xtstep, 1*\ystep + \ytstep) {$\Event_1$};
\node[]  at (0*\xstep + \xtstep, 3*\ystep - \ytstep) {$\Event_3$};
\node[]  at (1*\xstep - \xtstep, 2*\ystep +  \ytstep) {$\Event_2$};

\draw[->, thick] (e2) to node[above, sloped]{join} (e3);

\node[]  at (0*\xstep -5*\xtstep, 1*\ystep ) {$\CTC_{\Thread_1}=[27,5,9,45,17,26]$};
\node[]  at (0*\xstep -5*\xtstep, 3*\ystep ) {$\CTC_{\Thread_1}=[28,6,\textcolor{\darkred}{9},\textcolor{\darkred}{45},\textcolor{\darkred}{17},\textcolor{\darkred}{26}]$};
\node[]  at (1*\xstep + 5*\xtstep, 2*\ystep) {$\CTC_{\Thread_2}=[11,6,5,32,14,20]$};

\begin{scope}[shift={(\xdisposition,\ydisposition)}]

\Tree [
 .\node[] (a2) {$\treethr{2},6$};
 [ .\node[] (a3) {$\treethr{3},5$}; [ .\node[] (a5) {$\treethr{5},14$}; [.\node[] (a6) {$\treethr{6},20$};]  ] [ .\node[] (a4) {$\treethr{4},32$};]  ]
 [ .\node[] (a1) {$\treethr{1},11$};]
]
\node[node distance=5mm, above of=a2] {\normalsize $\CTC_{\Thread_2}$};
\end{scope}

\begin{pgfonlayer}{background}
\draw[smooth, draw=none, rounded corners, fill=\grayLight] ($ (a2) + (0, \halfyOuter) $) to ($ (a2) + (-\halfxOuter, \halfyOuter) $) to ($ (a3) + (-\halfxOuter, \halfyOuter) $) to ($ (a3) + (-\halfxOuter, -\halfyOuter) $) to ($ (a1) + (-\halfxOuter, -\halfyOuter) $) to ($ (a1) + (\halfxOuter, -\halfyOuter) $) to ($ (a1) + (\halfxOuter, \halfyOuter) $)  to ($ (a2) + (\halfxOuter, \halfyOuter) $)  to ($ (a2) + (0, \halfyOuter) $);

\draw[smooth, draw=none, rounded corners, fill=\grayDark] ($ (a2) + (0, \halfyInner) $) to ($ (a2) + (-\halfxInner, \halfyInner) $) to ($ (a2) + (-\halfxInner, -\halfyInner) $) to ($ (a2) + (+\halfxInner, -\halfyInner) $) to ($ (a2) + (\halfxInner, \halfyInner) $)  to ($ (a2) + (0, \halfyInner) $);
\end{pgfonlayer}
\end{scope}
\end{tikzpicture}
\caption{
(Left) Illustration of the effect of a join operation $\CTC_{\Thread_1}\gets \CTC_{\Thread_1} \CJoin \CTC_{\Thread_2}$ on the clocks of the two threads.
The $j$-th entry in timestamps correspond to thread $t_j$.
Red entries remain unchanged, as $\Thread_1$ already knows of a later time. 
(Right) A tree representation of the clocks $\CTC_{\Thread_2}$ that encodes transitivity.
Dark gray marks the threads whose clock has processed in $\CTC_{\Thread_2}$ compared to $\CTC_{\Thread_1}$ (i.e., just $\Thread_2$).
Light gray marks the nodes that we need to examine when performing the join operation.
}
\label{fig:motivating}
\end{figure*}

%% file: prelims.tex
%!TEX root=main.tex

\section{Preliminaries}\label{sec:prelims}

In this section we develop relevant notation and present standard concepts regarding concurrent executions, partial orders and vector clocks.

\input{model}

\input{vectorclocks_prelim}

\input{happens_before_races}

%% file: model.tex
%!TEX root=main.tex

\subsection{Concurrent Model and Traces}\label{subsec:model}

We start with our main notation on traces.
The exposition is standard and follows related work (e.g.,~\cite{Flanagan09,Smaragdakis12,Kini17}).

\Paragraph{Events and traces.}
We consider execution traces of concurrent programs represented as a sequence of events performed by different threads.
Each event is a tuple $\Event=\ev{i, \Thread, \Operation}$, 
where $i$ is the unique event identifier of $\Event$, $\Thread$ 
is the identifier of the thread that performs $\Event$, 
and $\Operation$ is the operation performed by $\Event$, 
which can be one of the following types
\footnote{Fork and join events are ignored for ease of presentation. Handling such events is straightforward.}.
\begin{compactenum}
\item $\Operation=\Read(x)$, denoting that $\Event$ reads global variable $x$.
\item $\Operation=\Write(x)$, denoting that $\Event$ writes to global variable $x$.
\item $\Operation=\Acquire(\lk)$, denoting that $\Event$ acquires the lock $\lk$.
\item $\Operation=\Release(\lk)$, denoting that $\Event$ releases the lock $\lk$.
\end{compactenum}
We write $\tid(\Event)$ and $\Operation(\Event)$ to denote the thread identifier and 
the operation of $\Event$, respectively.
For a read/write event $\Event$, we denote by $\Variable(\Event)$
the (unique) variable that $\Event$ accesses.
% We write $\Decor(\Event)$ for the (unique) variable or lock accessed by $\Event$.
We often ignore the identifier $i$ and represent $\Event$ as $\ev{\Thread, \Operation}$.
In addition, we are often not interested in the thread of $\Event$, 
in which case we simply denote $\Event$ by its operation, e.g., we refer to event $\Read(x)$.
When the variable of $\Event$ is not relevant, it is also omitted
(e.g., we may refer to a read event $\Read$).

A (concrete) \emph{trace} is a sequence of events $\Trace=\Event_1, \dots, \Event_n$.
The trace $\Trace$ naturally defines a total order $\trord{\Trace}$ (pronounced \emph{trace order})
over the set of events appearing in $\Trace$,
i.e., we have $\Event \trord{\Trace} \Event'$ iff either $\Event = \Event'$ or
 $\Event$ appears before $\Event'$ in $\Trace$;
 when $\Event \neq \Event'$, then we say $\Event \stricttrord{\Trace} \Event'$.
% \ucomment{Delay the following sentence.
% We can identify each event $\Event$ of $\Trace$ by a pair $(\tid, \clock)$,
% where $\tid$ is the thread of $\Event$ and $\clock$ is its local time, i.e., one plus the number of events $\Event'$ of the same thread that precede $\Event$ in $\Trace$.
% }
We require that $\Trace$ respects the semantics of locks.
That is, for every lock $\lk$ and every two acquire events 
$\Acquire_1(\lk)$, $\Acquire_2(\lk)$ on the lock $\lk$ such that 
$\Acquire_1(\lk) \stricttrord{\Trace} \Acquire_2(\lk)$,
there exists a lock release event $\Release_1(\lk)$ in $\Trace$
with $\tid(\Acquire_1(\lk))=\tid(\Release_1(\lk))$
and $\Acquire_1(\lk) \stricttrord{\Trace} \Release_1(\lk) \stricttrord{\Trace}\Acquire_2(\lk)$.
Finally, we denote by $\Threads{\Trace}$ the set of
thread identifiers appearing in $\Trace$.
% and by $\Locks{\Trace}$ and $\Variables{\Trace}$ the set of locks and variables, respectively, that appear in the events of $\Trace$.

\Paragraph{Thread order.}
Given a trace $\Trace$, the \emph{thread order} $\tho{\Trace}$ is the smallest partial order
such that $\Event_1 \tho{\Trace} \Event_2$ iff $\tid(\Event_1)=\tid(\Event_2)$ and $\Event_1 \trord{\Trace}\Event_2$.
For an event $\Event$ in a trace $\Trace$,
the local time $\LocalTime{\Trace}{\Event}$ of $\Event$ is 
the number of events
that appear before $\Event$ in the trace $\Trace$ that are also
performed by $\ThreadOf{\Event}$, i.e.,
$
\LocalTime{\Trace}{\Event} = |\setpred{\Event'}{\Event' \tho{\Trace} \Event}|
$.
We remark that the pair $(\ThreadOf{\Event}, \LocalTime{\Trace}{\Event})$ uniquely
identifies the event $\Event$ in the trace $\Trace$.

\Paragraph{Conflicting events.}
Two events of $\Event_1$, $\Event_2$ of $\Trace$ are called 
\emph{conflicting}, denoted by $\Confl{\Event_1}{\Event_2}$, if 
(i)~$\Variable(\Event_1)=\Variable(\Event_2)$,
(ii)~$\tid(\Event_1)\neq \tid(\Event_2)$, and
(iii)~at least one of $\Event_1$, $\Event_2$ is a write event.
The standard approach in concurrent analyses is to detect conflicting events that are causally independent, according to some pre-defined notion of causality, and can thus be executed concurrently.

%\Paragraph{Conflicting events and data races.}
%% The trace $\Trace$ has a \emph{data race} if there exists conflicting events 
%% $\Event_1 \stricttrord{\Trace} \Event_2$ that are successive in $\Trace$, i.e., 
%% $\not \exists \Event_3$ such that $\Event_1 \stricttrord{\Trace} \Event_3 \stricttrord{\Trace}\Event_2$.
%% \ucomment{I made it informal.}
%\Andreas{Do we need this? I would rather omit it here -- most of the presentation deals with detecting concurrent events wrt a partial order, and that's the message we should have.
%We can talk about data races in experiments, which concern a particular type of concurrent events.}
%Two events of $\Event_1$, $\Event_2$ of $\Trace$ are called 
%\emph{conflicting}, denoted by $\Confl{\Event_1}{\Event_2}$, if 
%(i)~$\Variable(\Event_1)=\Variable(\Event_2)$,
%(ii)~$\tid(\Event_1)\neq \tid(\Event_2)$, and
%(iii)~at least one of $\Event_1$, $\Event_2$ is a write event.
%Trace $\Trace$ has a \emph{data race} if there there
%are conflicting events $\Event_1, \Event_2$ in $\Trace$ 
%that are not ordered by synchronization in $\Trace$.
%This is formalized using the $\HB$ partial order
%which we define shortly.

%% file: vectorclocks_prelim.tex
%!TEX root=main.tex

\subsection{Partial Orders, Vector Times and Vector Clocks}
\label{subsec:vc_prelim}
A partial order on a set $S$ is a reflexive, transitive
and anti-symmetric binary relation on the elements of $S$.
Partial orders are the standard mathematical object for analyzing concurrent executions.
%Partial orders popularly used in the analysis
%of concurrent systems and race detection are the same.
%The $\HB$ partial order, first used in the context of distributed systems~\cite{Lamport78},
%is the most popular one on this list.
%Recently, other partial orders such as $\CP$~\cite{Smaragdakis12},
%$\SHB$~\cite{Mathur18}, $\WCP$~\cite{Kini17}, $\DC$~\cite{Roemer18},
%$\SDP$, $\WDP$~\cite{Genc19} have been proposed to
%enhance race detection in a \emph{predictive} setting.
%In the case of race detection and prediction, partial orders provide the promise of 
%efficiency~\cite{Kini17,Smaragdakis12} as against more expensive
%approaches for detecting data races~\cite{Said11,Huang14}.
The main idea behind such techniques is to define a 
partial order $\ord{\Trace}{P}$ on the set of events of the trace $\Trace$ 
being analyzed.
The intuition is that $\ord{\Trace}{P}$ 
captures \emph{causality} --- the relative order of two events of $\Trace$ 
must be maintained if they are ordered by $\ord{\Trace}{P}$.
More importantly, when two events $\Event_1$ and $\Event_2$ are unordered by
$\ord{\Trace}{P}$ (denoted $\Event_1 \unord{\Trace}{P} \Event_2$), 
then they can be deemed \emph{concurrent}.
This principle forms the backbone of all partial-order based concurrent analyses.
%race detection techniques ---
%look for two conflicting events in the trace that are unordered by the partial order of choice.

A na\"{i}ve approach for constructing such a partial order is to explicitly represent it as an acyclic directed graph over the events of $\Trace$, and then perform a graph search whenever needed to determine whether two events are ordered.
Vector clocks, on the other hand, provide a more efficient
method to represent partial orders and therefore are the key data structure in most partial order-based algorithms. The use of vector clocks enables designing streaming algorithms, which are also suitable for monitoring the system.
These algorithms associate 
\emph{vector timestamps}~\cite{Mattern89,Fidge91,fidge1988timestamps} with events
so that the point-wise ordering between timestamps reflects the underlying partial order.
Let us formalize these notions now.

\Paragraph{Vector Timestamps.}
Let us fix the set of threads $\Threads{}$ in the trace.
A \emph{vector timestamp} (or simply vector time) 
is a mapping $\VectorTime\colon \Threads{} \to \Nats$.
It supports the following operations.

{
\setlength\tabcolsep{3pt}
\begin{tabular}{lclr}
$\VectorTime_1 \CSmaller\VectorTime_2$ 
& iff & $\forall t\colon \VectorTime_1(\Thread)\leq \VectorTime_2(\Thread)$
& (Comparison)\\
$\VectorTime_1\CJoin \VectorTime_2$ & $=$ & $\lambda \Thread\colon \max(\VectorTime_1(\Thread), \VectorTime_2(\Thread))$ & (Join)\\
$\VectorTime\CAssign{\Thread'}{i}$ & $=$ & 
$
\lambda \Thread \colon
\begin{cases}
\VectorTime(\Thread) + i, & \text{if } \Thread=\Thread'\\
\VectorTime(\Thread), & \text{otherwise }
\end{cases}
$
& (Increment)
\end{tabular}
}

We write $\VectorTime_1 =\VectorTime_2$ to denote that $\VectorTime_1 \CSmaller\VectorTime_2$ and $\VectorTime_2 \CSmaller\VectorTime_1$.
Let us see how vector timestamps provide an efficient 
implicit representation of partial orders.

% Vector times provide an efficient implicit representation of $\HB$, exploiting the fact that $\leq^{\Trace}_{\TO} \subseteq  \leq^{\Trace}_{\HB}$.
% In particular, each event $\Event$ is assigned a vector time $\VectorTime_{\Event}$, such that $\VectorTime_{\Event}(\Thread)$ denotes the local time of the last event $\Event'$ of thread $\Thread$ such that $\Event'<^{\Trace}_{\HB}\Event$.
% A data structure implements the vector time interface by storing a vector time in its state and supporting vector-time operations.
% Naturally, the predominant data structure for this task is a simple integer array, commonly known as \emph{vector clock}.

\Paragraph{Timestamping for a partial order.}
Consider a partial order $\ord{\Trace}{P}$ defined
on the set of events of $\Trace$ such that $\tho{\Trace} \subseteq \ord{\Trace}{P}$.
In this case, we can define the $\mathsf{P}$-timestamp of
an event $\Event$ as the following vector timestamp:
\begin{equation*}
% \POTime{\Trace}{P}{\Event} = \lambda u:
\POTime{\ord{\Trace}{P}}{\Event} = \lambda u: \max \setpred{\,\LocalTime{\Trace}{f}}{f \ord{\Trace}{P} e,\ \ThreadOf{f} = u}
\end{equation*}
\camera{In words, $\POTime{\ord{\Trace}{P}}{\Event}$ contains the timestamps of the events that appear the latest in their respective threads such that they are ordered before $e$ in the partial order $\ord{\Trace}{P}$}. We remark that $\POTime{\ord{\Trace}{P}}{\Event}(\ThreadOf{\Event}) = \LocalTime{\Trace}{\Event}$.
The following observation then shows that the timestamps
defined above precisely capture the order $\ord{\Trace}{P}$.
\begin{lemma}
\label{lem:PO-timestamps}
Let $\ord{\Trace}{P}$ be a partial order defined on the set of events
of trace $\Trace$ such that $\tho{\Trace} \subseteq \ord{\Trace}{P}$. 
Then for any two events $\Event_1, \Event_2$ of $\Trace$, we have,
$\POTime{\ord{\Trace}{P}}{\Event_1} \CSmaller \POTime{\ord{\Trace}{P}}{\Event_2} \iff \Event_1 \ord{\Trace}{P} \Event_2$.
% \[
% \POTime{\ord{\Trace}{P}}{\Event_1} \CSmaller \POTime{\ord{\Trace}{P}}{\Event_2} \iff \Event_1 \ord{\Trace}{P} \Event_2
% \]
\end{lemma}

\camera{
In words, \cref{lem:PO-timestamps} implies that, in order to check whether two events are ordered according to $\ord{\Trace}{P}$, it suffices to compare their vector timestamps.
}

\Paragraph{The vector clock data structure.}
When establishing a causal order over the events of a trace,
the timestamps of an event is computed using timestamps of other events in the trace.
Instead of explicitly storing timestamps of each event, it is
often sufficient to store only the timestamps of a few events, as the algorithms is running.
Typically a data-structure called \emph{vector clocks} is used to store vector times.
Vector clocks are implemented as a simple integer array indexed by thread identifiers,
and they support all the operations on vector timestamps.
A useful feature of this data-structure is the
ability to perform in-place operations.
In particular, there are methods such as
$\FunctionVCJoin(\cdot)$, $\FunctionVCCopy(\cdot)$ or 
$\FunctionVCIncrement(\cdot,\cdot)$
that store the result of the corresponding vector time operation
in the original instance of the data-structure.
For example, for a vector clock $\mathbb{C}$ and a vector time $V$, 
a function call $\mathbb{C}.\FunctionVCJoin(V)$ 
%(resp. $\mathbb{C}.\FunctionVCCopy(V)$)
stores the value  $\mathbb{C} \CJoin V$ back in $\mathbb{C}$.
Each of these operations iterates over all the thread identifiers
(indices of the array representation) and compares the
corresponding components in $\mathbb{C}$ and $V$.
%and then stores the maximum of the two in the corresponding index of $\mathbb{C}$.
%Assuming arithmetic operations take constant time,
The running time of the join operation
for the vector clock data structure is thus 
$\Theta(\NumThreads)$, where $\NumThreads$ is the number
of threads.
Similarly, copy and comparison operations take $\Theta(\NumThreads)$
time.
% and an increment operation takes $O(1)$ time with vector clocks.

% \ucomment{Define in place operations on data structures}

%% file: happens_before_races.tex
%!TEX root=main.tex

\subsection{The Happens-Before Partial Order}\label{subsec:happens_before_races}

Lamport's Happens-Before ($\HB$)~\cite{Lamport78} is one of the most frequently used partial orders for the analysis of concurrent executions, with wide applications in domains such as dynamic race detection.
Here we use $\HB$ to illustrate the disadvantages of vector clocks and form the basis for the tree clock data structure.
In later sections we show how tree clocks also apply to other partial orders, such as Schedulably-Happens-Before and the Mazurkiewicz partial order.

% For the sake of completeness and in order to develop some intuition behind tree clocks, here we present the basis on race detection.
% Although there exist many techniques in race \emph{prediction}~\cite{Smaragdakis12,Kini17,Roemer18,Pavlogiannis2020},
% this section focuses on race \emph{detection} based on happens-before races.
% We will show how tree clocks apply on race prediction techniques later in \cref{sec:race_prediction}.
%Although happens-before is a general-purpose partial order,
%it has seen wide adoption in dynamic data-race detection.
%and forms the basis of
%popular race detectors such as \tsan~\cite{threadsanitizer}.
%Here we use this context to illustrate happens-before.

\Paragraph{Happens-before.}
Given a trace $\Trace$, the \emph{happens-before} ($\HB$) partial order $\hb{\Trace}$ of $\Trace$ 
is the smallest partial order over the events of $\Trace$ that satisfies the following conditions.
\begin{compactenum}
\item $\tho{\Trace} \subseteq  \hb{\Trace}$.
\item For every release event  
$\Release(\lk)$ and  acquire event $\Acquire(\lk)$ on the same lock $\lk$ with 
$\Release(\lk) \stricttrord{\Trace} \Acquire(\lk)$, we have 
$\Release(\lk)\hb{\Trace} \Acquire(\lk)$.
\end{compactenum}
For two events $\Event_1, \Event_2$ in trace $\Trace$, we use
$\Event_1 \unordhb{\Trace} \Event_2$ to denote that neither
$\Event_1 \hb{\Trace} \Event_2$, nor $\Event_2 \hb{\Trace} \Event_1$.
We say $\Event_1 \stricthb{\Trace} \Event_2$
when $\Event_1 \neq \Event_2$ and $\Event_1 \hb{\Trace} \Event_2$.
% 
% \input{figures/fig_hb_race}
% 
% \Paragraph{Happens-before races.}
Given a trace $\Trace$, two events $\Event_1$, $\Event_2$ of $\Trace$ are said to be in a \emph{happens-before (data) race} if 
(i)~$\Confl{\Event_1}{\Event_2}$ and
(ii)~$\Event_1\unordhb{\Trace} \Event_2$.
% The motivation behind this definition is the observation that the existence of an $\HB$ race in $\Trace$
% guarantees the existence of a data race in the program that generated $\Trace$.
% See \cref{fig:hb_race} for an illustration.

% The general algorithm for race detection is in computing the $\HB$ partial order.
% For reasons of efficiency, the algorithm is based on the notion vector times and vector clocks~\cite{Lamport78,Mattern89}, which we introduce below.

\input{algorithms/algo_hb_basic}

\Paragraph{The happens-before algorithm.}
In light of \cref{lem:PO-timestamps}, race detection based on
$\HB$ constructs the $\hb{\Trace}$ partial order in terms of vector timestamps
and detects races using these. 
The core algorithm for constructing $\hb{}$ is shown in \cref{algo:hb_basic}.
% \ucomment{Write it in terms of in place operations}
% It is a simple streaming algorithm, in the sense that it performs a single top-down scan of $\Trace$.
The algorithm maintains a vector clock $\CTC_{\Thread}$ for every thread 
$\Thread\in \Threads{}$, and a similar one $\CTC_{\lk}$ for every lock $\lk$.
When processing an event $\Event=\ev{\Thread, \Operation}$, it performs an update
$\CTC_{\Thread}.\FunctionVCIncrement(\Thread, 1)$,
which is implicit and not shown in \cref{algo:hb_basic}.
Moreover, if $\Operation=\Acquire(\lk)$ or $\Operation=\Release(\lk)$, 
the algorithm executes the corresponding procedure.
The $\HB$-timestamp of $\Event$ is then simply the value stored in
$\CTC_{\ThreadOf{e}}$ right after $\Event$ has been processed.

%\Paragraph{Race detection logic.}
%\cref{algo:hb_basic} only computes $\HB$ partial order without the race detection logic.
%This is typically done by maintaining two integer arrays $\VCR_x$, $\VCW_x$, for every $x\in \Variable$.
%For each thread $\Thread$,  $\VCR_x(\Thread)$ (resp., $\VCW_x(\Thread)$) stores the local index of the last read (resp., write) event of $\Thread$. 
%When the algorithm process a read event $\Read(x)$ from thread $\Thread$,
%it first computes the vector time of $\Read$ in $\CTC_{\Thread}$, 
%and reports a race on $\Read$ iff $\VCW_x \not \CSmaller \CTC_{\Thread}$.
%The logic for a write event $\Write(x)$ is similar, testing for both $\VCW_x \not \CSmaller \CTC_{\Thread}$ and $\VCR_x \not \CSmaller \CTC_{\Thread}$.
%Some alternatives such as the epoch optimization of~\cite{Flanagan09} will be discussed later.

%\input{figures/fig_hb_suboptimal}
\Paragraph{Running time using vector clocks.}
If a trace $\Trace$ has $\NumEvents$ events and $\NumThreads$ threads,
computing the $\HB$ partial order with 
\cref{algo:hb_basic} and using vector clocks takes $O(\NumEvents\cdot \NumThreads)$ time.
%In fact, the bound is $\Theta(\NumEvents\cdot \NumThreads)$, 
%as there exist inputs for which the algorithm hits its upper-bound complexity.
The quadratic bound occurs because every vector clock join and copy operation iterates over all $\NumThreads$ threads.

%% file: algorithms/algo_hb_basic.tex
%!TEX root=../main.tex

%\small
\begin{algorithm}[t]
\vspace*{-\multicolsep}
\begin{multicols}{2}

\MyProcedure{\ProcAcquire{$\Thread$, $\lk$}}{
% $\CTC_{\Thread}\gets \CTC_{\Thread}\CJoin \CTC_{\lk}$
$\CTC_{\Thread}.\FunctionVCJoin(\CTC_{\lk})$

}
% \BlankLine
\BlankLine
\MyProcedure{\ProcRelease{$\Thread$, $\lk$}}{
$\CTC_{\lk}.\FunctionVCCopy(\CTC_{\Thread})$
}
\end{multicols}
\vspace*{-0.5\multicolsep}
\normalsize
\caption{
Computing the $\HB$ partial order.
}
\label{algo:hb_basic}
\end{algorithm}
\normalsize

%% file: treeclocks.tex
\section{The Tree Clock Data Structure}\label{sec:tree_clocks}

In this section we introduce tree clocks, a new data structure for representing logical times in concurrent and distributed systems.
We first illustrate the intuition behind tree clocks, and then develop the data structure in detail.

\input{intuition}

\input{treeclocks_details}

%% file: intuition.tex
%!TEX root=main.tex

\subsection{Intuition}\label{subsec:intuition}

\input{figures/fig_intuition1}

Like vector clocks, tree clocks represent vector timestamps that record a thread's knowledge of events in other threads. 
Thus, for each thread $\Thread$, a tree clock records the last known local time of $\Thread$. 
However, unlike a vector clock which is flat, a tree clock maintains this information hierarchically --- nodes store local times of a thread, while the tree structure records how this information has been obtained transitively through intermediate threads. 
%Before introducing the data structure, we outline two key ideas that motivate its benefits. 
%These ideas are critical in enabling clock updates in sub-linear time. 
In the following examples we use the operation $\Sync(\lk)$ to 
denote the sequence $\Acquire(\lk), \Release(\lk)$.

\Paragraph{1.~Direct monotonicity.}
Recall that a vector clock-based algorithm like \cref{algo:hb_basic} maintains a vector clock $\CTC_\Thread$ which intuitively captures thread $\Thread$'s knowledge about all threads. However, it does not maintain \emph{how} this information was acquired. Knowledge of how such information  was acquired can be exploited in join operations, as we show through an example. Consider a computation of the $\HB$ partial order for the trace $\Trace$ shown in \cref{subfig:intuition1}. At event $\Event_7$, thread $\Thread_4$ transitively learns information about events in the trace through thread $\Thread_3$ because $\Event_{6} \stricthb{\Trace} \Event_{7}$ (dashed edge in \cref{subfig:intuition1}). This is accomplished by joining with clock $\CTC_{\Thread_3}$ of thread $\Thread_3$. Such a join using vector clocks will take 4 steps because we need to take the pointwise maximum of two vectors of length $4$. 

Suppose in addition to these timestamps, we maintain how these timestamps were updated in each clock. This would allow one to make the following observations.
\begin{compactenum}
\item Thread $\Thread_3$ knows of event $\Event_1$ of $\Thread_1$ transitively, through event $\Event_2$ of thread $\Thread_2$.
\item Thread $\Thread_4$ (before the join at $\Event_7$) knows of event $\Event_1$ through $\Event_4$ of thread $\Thread_2$.
\end{compactenum}
Before the join, since $\Thread_4$ has a more recent view of $\Thread_2$ when compared to $\Thread_3$, it is aware of all the information that thread $\Thread_3$ knows about the world via thread $\Thread_2$. 
Thus, when performing the join, we need not examine the component corresponding to thread $\Thread_1$ in the two clocks. 
Tree clocks, by maintaining such additional information, can avoid examining some components of a vector timestamp and yield sublinear updates.

\Paragraph{2.~Indirect monotonicity.}
We now illustrate that if in addition to information about ``how a view of a thread was updated'', we also maintained ``when the view of a thread was updated'', the cost of join operations can be further reduced.
Consider the trace $\Trace$ of \cref{subfig:intuition2}. 
At each of the events of thread $\Thread_4$, it learns about events in the trace transitively through thread $\Thread_3$ by performing two join operations. 
At the first join (event $\Event_5$), thread $\Thread_4$ learns about events $\Event_1$, $\Event_2$, $\Event_3$ transitively through event $\Event_4$. 
At event $\Event_7$, thread $\Thread_4$ finds out about new events in thread $\Thread_3$ (namely, $\Event_6$). 
However, it does not need to update its knowledge about threads $\Thread_1$ and $\Thread_2$ --- thread $\Thread_3$'s information about threads $\Thread_1$ and $\Thread_2$ were acquired by the time of event $\Event_4$ about which thread $\Thread_4$ is aware. 
Thus, if information about when knowledge was acquired is also kept, this form of ``indirect monotonicity'' can be exploited to avoid examining all components of a vector timestamp.

%Direct and indirect monotonicity reason about transitive orderings.
The flat structure of vector clocks misses the transitivity of information sharing, and thus arguments based on monotonicity are lost, resulting in vacuous operations.
On the other hand, tree clocks maintain transitivity in their hierarchical structure.
This enables reasoning about direct and indirect monotonicity, and thus avoid redundant operations.

%% file: figures/fig_intuition1.tex
%!TEX root=../main.tex

\begin{figure*}[t]
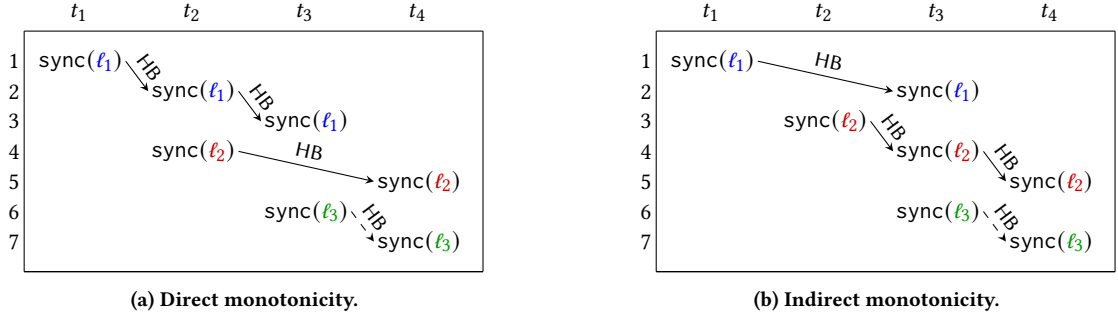

\begin{subfigure}[b]{0.45\textwidth}
\centering
\execution{4}{
%\figev{1}{$\Acquire(\textcolor{blue}{\ell_1})$}
%\figev{1}{$\Release(\textcolor{blue}{\ell_1})$}
\figev{1}{$\Sync(\textcolor{blue}{\ell_1})$}
%\figev{2}{$\Acquire(\textcolor{blue}{\ell_1})$}
%\figev{2}{$\Release(\textcolor{blue}{\ell_1})$}\\
\figev{2}{$\Sync(\textcolor{blue}{\ell_1})$}
%\figev{3}{$\Acquire(\textcolor{blue}{\ell_1})$}
%\figev{3}{$\Release(\textcolor{blue}{\ell_1})$}
\figev{3}{$\Sync(\textcolor{blue}{\ell_1})$}
%\figev{2}{$\Acquire(\textcolor{\darkred}{\ell_2})$}
%\figev{2}{$\Release(\textcolor{\darkred}{\ell_2})$}
\figev{2}{$\Sync(\textcolor{\darkred}{\ell_2})$}
%\figev{4}{$\Acquire(\textcolor{\darkred}{\ell_2})$}
%\figev{4}{$\Release(\textcolor{\darkred}{\ell_2})$}
\figev{4}{$\Sync(\textcolor{\darkred}{\ell_2})$}
%\figev{3}{$\Acquire(\textcolor{\darkgreen}{\ell_3})$}
%\figev{3}{$\Release(\textcolor{\darkgreen}{\ell_3})$}
\figev{3}{$\Sync(\textcolor{\darkgreen}{\ell_3})$}
%\figev{4}{$\Acquire(\textcolor{\darkgreen}{\ell_3})$}
%\figev{4}{$\Release(\textcolor{\darkgreen}{\ell_3})$}
\figev{4}{$\Sync(\textcolor{\darkgreen}{\ell_3})$}
\orderedgewithlabel{1}{1}{0.6}{2}{2}{-0.6}{\small $\HB$}{above}
\orderedgewithlabel{2}{2}{0.6}{3}{3}{-0.6}{\small $\HB$}{above}
\orderedgewithlabel{2}{4}{0.6}{4}{5}{-0.6}{\small $\HB$}{above}
\orderdashededgewithlabel{3}{6}{0.6}{4}{7}{-0.6}{\small $\HB$}{above}
}
\caption{
Direct monotonicity.
}
\label{subfig:intuition1}
\end{subfigure}
\quad
\begin{subfigure}[b]{0.45\textwidth}
\centering
\execution{4}{
%\figev{1}{$\Acquire(\textcolor{blue}{\ell_1})$}
%\figev{1}{$\Release(\textcolor{blue}{\ell_1})$}
\figev{1}{$\Sync(\textcolor{blue}{\ell_1})$}
%\figev{3}{$\Acquire(\textcolor{blue}{\ell_1})$}
%\figev{3}{$\Release(\textcolor{blue}{\ell_1})$}
\figev{3}{$\Sync(\textcolor{blue}{\ell_1})$}
%\figev{2}{$\Acquire(\textcolor{\darkred}{\ell_2})$}
%\figev{2}{$\Release(\textcolor{\darkred}{\ell_2})$}
\figev{2}{$\Sync(\textcolor{\darkred}{\ell_2})$}
%\figev{3}{$\Acquire(\textcolor{\darkred}{\ell_2})$}
%\figev{3}{$\Release(\textcolor{\darkred}{\ell_2})$}
\figev{3}{$\Sync(\textcolor{\darkred}{\ell_2})$}
%\figev{4}{$\Acquire(\textcolor{\darkred}{\ell_2})$}
%\figev{4}{$\Release(\textcolor{\darkred}{\ell_2})$}
\figev{4}{$\Sync(\textcolor{\darkred}{\ell_2})$}
%\figev{3}{$\Acquire(\textcolor{\darkgreen}{\ell_3})$}
%\figev{3}{$\Release(\textcolor{\darkgreen}{\ell_3})$}
\figev{3}{$\Sync(\textcolor{\darkgreen}{\ell_3})$}
%\figev{4}{$\Acquire(\textcolor{\darkgreen}{\ell_3})$}
%\figev{4}{$\Release(\textcolor{\darkgreen}{\ell_3})$}
\figev{4}{$\Sync(\textcolor{\darkgreen}{\ell_3})$}
\orderedgewithlabel{1}{1}{0.6}{3}{2}{-0.6}{\small $\HB$}{above}
\orderedgewithlabel{2}{3}{0.6}{3}{4}{-0.6}{\small $\HB$}{above}
\orderedgewithlabel{3}{4}{0.6}{4}{5}{-0.6}{\small $\HB$}{above}
\orderdashededgewithlabel{3}{6}{0.6}{4}{7}{-0.6}{\small $\HB$}{above}
}
\caption{
Indirect monotonicity.
}
\label{subfig:intuition2}
\end{subfigure}
\caption{
Illustration of the two insights behind tree clocks.
An event $\Sync(\ell)$ represents two events $\Acquire(\ell),\Release(\ell)$.
}
\label{fig:intuitions}
\end{figure*}

%% file: treeclocks_details.tex
%!TEX root=main.tex

\subsection{Tree Clocks}\label{subsec:clocktrees_details}

We now present the tree clock data structure in detail.

\input{figures/fig_treeclocks}

\Paragraph{Tree clocks.}
A tree clock $\TreeClock$ consists of the following.
\begin{compactenum}
\item $\GTree=(\TNodes, \TEdges)$ is a \emph{rooted tree} of nodes of the form $(\tid,\clock, \pclock)\in \Threads{} \times  \Nats^2$.
Every node $u$ stores its children in an ordered list 
%(e.g., a stack with random access) 
$\ChildrenQ(u)$ of descending $\pclock$ order.
We also store a pointer $\Parent(u)$ of $u$ to its parent in $\GTree$.
\item $\ThreadMap\colon \Threads{} \to \TNodes$ is a \emph{thread map},
with the property that if $\ThreadMap(\Thread)=(\tid, \clock, \pclock)$, then $\Thread=\tid$.
\end{compactenum}
We denote by $\GTree.\Root$ the root of $\GTree$, and
for a tree clock $\TreeClock$ we refer by $\TreeClock.\GTree$ and $\TreeClock.\ThreadMap$ to the rooted tree and thread map of $\TreeClock$, respectively.
For a node $u=(\tid,\clock,\pclock)$ of $\GTree$, we let $u.\tid=\tid$, $u.\clock=\clock$ and $u.\pclock=\pclock$,
and say that $u$ \emph{points to} the unique event $e$ \camera{with
$\ThreadOf{e} = \tid$ and $\LocalTime{}{e} = \clock$.}
 % points to, i.e., it is the event of thread $\tid$ with local time $\clock$.
Intuitively, if $v=\Parent(u)$, then $u$ represents the following information.
\begin{compactenum}
\item $\TreeClock$ has the \emph{local time} $u.\clock$ for thread $u.\tid$.
\item $u.\pclock$ is the \emph{attachment time} of $v.\tid$, which is the local time of $v$ when $v$ learned about $u.\clock$ of $u.\tid$
(this will be the time that $v$ had when $u$ was attached to $v$).
\end{compactenum}
Naturally, if $u=\GTree.\Root$ then $u.\pclock=\bot$.
See \cref{fig:treeclocks}.
%for an illustration.

\input{algorithms/algo_clock_tree}

\Paragraph{Tree clock operations.}
Just like vector clocks, tree clocks provide functions for initialization, update and comparison.
There are two main operations worth noting.
The first is $\FunctionCTJoin$ --- 
$\TreeClock_1.\FunctionCTJoin(\TreeClock_2)$ joins the tree clock $\TreeClock_2$ to $\TreeClock_1$.
In contrast to vector clocks, this operation takes advantage of the direct and indirect monotonicity outlined in \cref{subsec:intuition} to perform the join in sublinear time in the size of $\TreeClock_1$ and $\TreeClock_2$ (when possible).
The second is $\FunctionCTMonotoneCopy$.
We use $\TreeClock_1.\FunctionCTMonotoneCopy(\TreeClock_2)$ to copy $\TreeClock_2$ to $\TreeClock_1$ when we know that $\TreeClock_1\CSmaller \TreeClock_2$.
The idea is that when this holds, the copy operation has the same semantics as the join, and hence the principles that make $\FunctionCTJoin$ run in sublinear time also apply to $\FunctionCTMonotoneCopy$.

\input{figures/fig_ctjoin}
\input{figures/fig_ctmonotonecopy}

\cref{algo:clock_tree} gives a pseudocode description of this functionality.
The functions on the left column present operations that can be performed on tree clocks, while the right column lists helper routines for the more involved functions $\FunctionCTJoin$ and $\FunctionCTMonotoneCopy$.
In the following we give an intuitive description of each function.

\SubParagraph{1. $\FunctionCTInitialize(\Thread)$.}
This function initializes a tree clock $\TreeClock_{\Thread}$ that belongs to thread $\Thread$, by creating a node $u=(\Thread, 0, \bot)$.
Node $u$ will always be the root of $\TreeClock_{\Thread}$.
This initialization function is only used for tree clocks that represent the clocks of threads.
Auxiliary tree clocks for storing vector times of release events do not execute this initialization.

\SubParagraph{2. $\FunctionCTGet(\Thread)$.}
This function simply returns the time of thread $\Thread$ stored in $\TreeClock$, while it returns $0$ if $\Thread$ is not present in $\TreeClock$.
%If $\TreeClock$ is not aware of any event of thread $\Thread$, the value returned is $0$.

\SubParagraph{3. $\FunctionCTIncrement(i)$.}
This function increments the time of the root node of $\TreeClock$.
It is only used on tree clocks that have been initialized using $\FunctionCTInitialize$,
i.e., the tree clock belongs to a thread that is always stored in the root of the tree.

\SubParagraph{4. $\FunctionCTLessThan(\TreeClock')$.}
This function compares the vector time of $\TreeClock$ to the vector time of $\TreeClock'$,
i.e., it returns $\True$ iff $\TreeClock\CSmaller\TreeClock'$.

\SubParagraph{5. $\FunctionCTJoin(\TreeClock')$.}
This function implements the join operation with $\TreeClock'$,
i.e., updating $\TreeClock\gets \TreeClock\CJoin \TreeClock'$.
At a high level, the function performs the following steps.
\begin{compactenum}
\item Routine $\RoutineUpdatedNodesForJoin$ performs a pre-order traversal of $\TreeClock'$, and gathers in a stack $\Stack$ the nodes of $\TreeClock'$ that have progressed in $\TreeClock'$ compared to $\TreeClock$. \camera{The traversal may stop early due to direct or indirect monotonicity, hence, this routine generally takes sub-linear time.}
\item Routine $\RoutineDetachNodes$ detaches from $\TreeClock$ the nodes whose $\tid$ appears in $\Stack$, as these will be repositioned in the tree.
\item Routine $\RoutineAttachNodes$ updates the nodes of $\TreeClock$ that were detached in the previous step, and repositions them in the tree.
This step effectively creates a subtree of nodes of $\TreeClock$ that is identical to the subtree of $\TreeClock'$ that contains the progressed nodes computed by $\RoutineUpdatedNodesForJoin$.
\item Finally, the last 4 lines of $\FunctionCTJoin$ attach the subtree constructed in the previous step under the root $z$ of $\TreeClock$, at the front of the $\ChildrenQ(z)$ list.
\end{compactenum}
\cref{fig:ctjoin} provides an illustration.

\SubParagraph{6. $\FunctionCTMonotoneCopy(\TreeClock')$.}
This function implements the copy operation $\TreeClock\gets \TreeClock'$ assuming that $\TreeClock \CSmaller\TreeClock'$.
The function is very similar to $\FunctionCTJoin$.
The key difference is that this time, the root of $\TreeClock$ is always considered to have progressed in $\TreeClock'$, even if the respective times are equal.
This is required for changing the root of $\TreeClock$ from the current node to one with $\tid$ equal to the root of $\TreeClock'$.
\cref{fig:ctmonotonecopy} provides an illustration.

The crucial parts of $\FunctionCTJoin$ and $\FunctionCTMonotoneCopy$ that exploit the hierarchical structure of tree clocks are in $\RoutineUpdatedNodesForJoin$ and $\RoutineUpdatedNodesForCopy$.
In each case, we proceed from a parent $u'$ to its children $v'$ only if $u'$ has progressed wrt its time in $\TreeClock$ (recall \cref{subfig:intuition1}), capturing \emph{direct monotonicity}.
Moreover, we proceed from a child $v'$ of $u'$ to the next child $v''$ (in order of appearance in $\ChildrenQ(u')$) only if $\TreeClock$ is not yet aware of the attachment time of $v'$ on $u'$ (recall \cref{subfig:intuition2}), capturing \emph{indirect monotonicity}.

\begin{remark}[Constant time epoch accesses]\label{rem:epochs}
The function $\TreeClock.\FunctionCTGet(\Thread)$ returns the time of thread $\Thread$ stored in $\TreeClock$ in $O(1)$ time, just like vector clocks.
This allows all epoch-related optimizations~\cite{Flanagan09,Roemer20}  from vector clocks to apply to tree clocks.
\end{remark}

%% file: figures/fig_treeclocks.tex
%!TEX root=../main.tex
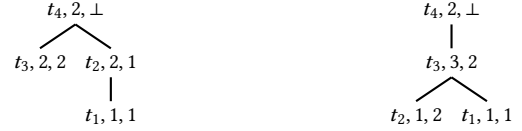
\begin{figure}
\small
\newcommand{\xdisposition}{5}
\newcommand{\ydisposition}{0}
\newcommand{\xtstep}{0.9}
\newcommand{\ytstep}{0.5}
\newcommand{\xstep}{2.8}
\newcommand{\ystep}{1.2}
\newcommand{\xtscale}{0.8}

\newcommand{\halfyOuter}{0.4}
\newcommand{\halfxOuter}{0.5}
\newcommand{\halfyInner}{0.3}
\newcommand{\halfxInner}{0.4}
\newcommand{\grayDark}{gray!50}
\newcommand{\grayLight}{gray!20}

\centering
\scalebox{1}{
\begin{tikzpicture}[thick,sibling distance=1em,
every tree node/.style={thick, rectangle, rounded corners, minimum height=0.8mm, minimum width=5mm, inner sep=2pt},
sibling distance=3pt,
level distance=20pt,
]
\tikzset{edge from parent/.append style={ thick}}
\pgfdeclarelayer{background}
\pgfdeclarelayer{foreground}
\pgfsetlayers{background,main,foreground}

\begin{scope}[shift={(0,0)}]

\Tree [
 . \node[] (a1) {$\treethr{4},2,\bot$};
	[. \node[] (a2) {$\treethr{3},2,2$};]
	[. \node[] (a3) {$\treethr{2},2,1$}; 
		[ .\node[] (a7) {$\treethr{1},1,1$}; ] 
	]
]

%\node[above of=a1] {\normalsize $\TreeClock_1$};

\end{scope}

\begin{scope}[shift={(\xdisposition,0)}]

\Tree [
 . \node[] (b1) {$\treethr{4},2,\bot$};
	[ .\node[] (b2) {$\treethr{3},3,2$}; 
		[ .\node[] (b3) {$\treethr{2},1,2$};] 
		[ .\node[] (b4) {$\treethr{1},1,1$};] 
	]
]

\end{scope}

\end{tikzpicture}
}
\caption{
The tree clock of $\treethr{4}$ after processing the event $\Event_{7}$ in the traces of \cref{subfig:intuition1} (left) and \cref{subfig:intuition2} (right).
}
\label{fig:treeclocks}
\end{figure}

%% file: algorithms/algo_clock_tree.tex
%!TEX root = ../main.tex

%\small
\begin{algorithm*}
% \BlankLine
% \vspace{-1.5em}
\SetInd{0.3em}{0.5em}
\addtolength{\columnsep}{-10pt}
\vspace*{-\multicolsep}
\begin{multicols}{2}
%\BlankLine
%\tcp*[f]{}
\tcp{Initialize a tree clock for thread $\Thread$}
\MyFunction{\FunctionCTInitialize{$\Thread$}}{
%Let $\GTree\gets $ an empty tree\\
%Let $\ThreadMap \gets$ a map initialized to $\bot$\\
Let $u\gets(\Thread, 0, \bot)$\\
Make $u$  the root of $\GTree$\\
Let $\ThreadMap(\Thread)\gets u$\\
}
%\BlankLine
\BlankLine
\tcp{Get the clock for thread $\Thread$}
\MyFunction{\FunctionCTGet{$\Thread$}}{
\uIf{$\TreeClock.\ThreadMap(\Thread)\neq \bot$}{
Let $u\gets \ThreadMap(\Thread)$\\
\Return{$u.\clock$}
}
\Return{$0$}
}
%\BlankLine
\BlankLine
\tcp{Increment the clock of the root thread}
\MyFunction{\FunctionCTIncrement{$i$}}{
Let $z\gets \GTree.\Root$\\
$z.\clock\gets z.\clock +  i$
}
%\BlankLine
\BlankLine
\tcp{$\True$ iff $\CSmaller\TreeClock'$}
\MyFunction{\FunctionCTLessThan{$\TreeClock'$}}{
Let $z\gets \GTree.\Root$\\
\Return{$z.\clock \leq  \TreeClock'.\FunctionCTGet(z.\tid)$}
}
%\BlankLine
\BlankLine
\tcp{Update with $\CJoin\TreeClock'$}
\MyFunction{\FunctionCTJoin{$\TreeClock'$}}{
Let $z'\gets \TreeClock'.\GTree.\Root$\\
\uIf{$z'.\clock \leq \FunctionCTGet(z'.\tid)$}{\label{line:functionctjoinifroot}
\Return %\tcp*[f]{We have $\CSmaller\TreeClock'$, join is silent}
}
Let $\Stack\gets$ an empty stack\\
\RoutineUpdatedNodesForJoin($\Stack$, $z'$)\\
\RoutineDetachNodes($\Stack$)\label{line:functionctjoin_call_detach}\\
\RoutineAttachNodes($\Stack$)\\
\tcp{Place the updated subtree under the root of $\GTree$}
Let $w\gets \ThreadMap(z'.\tid)$\\
Let $z\gets \GTree.\Root$\\
Assign $w.\pclock\gets z.\clock$\\
\RoutinePushChild($w$, $z$)\\
}
%\BlankLine
\BlankLine
\tcp{Monotone copy, assumes that $\mathsf{this}\CSmaller\TreeClock'$}
\MyFunction{\FunctionCTMonotoneCopy{$\TreeClock'$}}{
Let $z'\gets\TreeClock'.\GTree.\Root$\\
Let $z\gets\GTree.\Root$\\ 
Let $\Stack\gets$ an empty stack\\
\RoutineUpdatedNodesForCopy($\Stack$, $z'$, $z$)\\
\RoutineDetachNodes($\Stack$)\label{line:functionctmonotonecopy_call_detach}\\
\RoutineAttachNodes($\Stack$)\\
\tcp{New root has the same $\tid$ as the root of $\TreeClock'.\GTree$}
Assign $\GTree.\Root\gets \ThreadMap (z'.\tid)$\\
}
%\BlankLine
%\BlankLine
\pagebreak
\tcp{
Populate $\Stack$ with a pre-order traversal of the subtree rooted at $u'$
with nodes whose clock has progressed
}
\MyRoutine{\RoutineUpdatedNodesForJoin{$\Stack$, $u'$}}{
\ForEach{$v'$ in $\ChildrenQ(u')$}{\label{line:routineupdatednodesforjoin_loop}
\lIf{$\FunctionCTGet(v'.\tid) < v'.\clock$}{\label{line:routineupdatednodesforjoin_if}
\RoutineUpdatedNodesForJoin($\Stack$, $v'$)
}
\lElseIf{$v'.\pclock \leq \FunctionCTGet(u'.\tid)$}{\label{line:routineupdatednodesforjoinifbreak}
\Break
}
}
Push $u'$ in $\Stack$\\
}
%\BlankLine
\BlankLine
\tcp{Detach from $\GTree$ the nodes with $\tid$ that appears in $\Stack$}
\MyRoutine{\RoutineDetachNodes{$\Stack$}}{
\ForEach{$v'$ in $\Stack$}{
\uIf{$\ThreadMap(v'.\tid)\neq \bot$}{
Let $v\gets \ThreadMap(v'.\tid)$\\
\uIf{$v\neq \GTree.\Root$}{
Let $x \gets \Parent(v)$\\
Remove $v$  from $\ChildrenQ(x)$
}
}
}
}
%\BlankLine
\BlankLine
\tcp{
Re-attach the nodes of $\GTree$ with $\tid$ that appears in $\Stack$ to obtain the shape corresponding to $\TreeClock'.\GTree$
}
\MyRoutine{\RoutineAttachNodes{$\Stack$}}{
\While{$\Stack$ is not empty}{\label{line:routineattachnodesmainloop}
Let $u'\gets$ pop  $\Stack$\\
\eIf{$\ThreadMap(u'.\tid)\neq \bot$}{
Let $u\gets \ThreadMap(u'.\tid)$
}{
Let $u\gets (u'.\tid, 0, \bot)$\\
Let $\ThreadMap(u.\tid)\gets u$\\
}
Assign $u.\clock \gets u'.\clock$\\
Let $y'\gets \Parent(u')$\\
\uIf{$y'\neq \bot$}{
Assign $u.\pclock \gets u'.\pclock$\\
Let $y\gets\ThreadMap(y'.\tid)$\\
\RoutinePushChild($u$, $y$)\\
}
}
}
%\BlankLine
\BlankLine
\tcp{Similar to $\RoutineUpdatedNodesForJoin$}
\MyRoutine{\RoutineUpdatedNodesForCopy{$\Stack$, $u'$, $z$}}{
\ForEach{$v'$ in $\ChildrenQ(u')$}{\label{line:routineupdatednodesforcopy_loop}
\eIf{$\FunctionCTGet(v'.\tid) < v'.\clock$}{\label{line:routineupdatednodesforcopy_if}
\RoutineUpdatedNodesForCopy($\Stack$, $v'$, $z$)\\
}{
\lIf{$z\neq \bot$ and $v'.\tid=z.\tid$}{
Push $v'$ in $\Stack$\label{line:routineupdatenodesforcopy_root}
}
\lIf{$v'.\pclock \leq \FunctionCTGet(u'.\tid)$}{\label{line:routineupdatednodesforcopyifbreak}
\Break
}
}
}
Push $u'$ in $\Stack$\\
}
%\BlankLine
\BlankLine
\tcp{Push $u$  in the front of head of $\ChildrenQ(v)$}
\MyRoutine{\RoutinePushChild{$u$, $v$}}{
%Assign $u.\pclock \gets v.\clock$\\
Assign $\Parent(u)\gets v$\\
Push $u$ to the front of $\ChildrenQ(v)$\\
}
\end{multicols}

\vspace{-0.7em}

\normalsize
\caption{
The tree clock data structure.
}
\label{algo:clock_tree}
\end{algorithm*}
\normalsize

%% file: figures/fig_ctjoin.tex
%!TEX root = ../main.tex

\begin{figure*}[t]
\small
\newcommand{\xdisposition}{4.95}
\newcommand{\ydisposition}{0}
\newcommand{\xtstep}{0.9}
\newcommand{\ytstep}{0.5}
\newcommand{\xstep}{2.8}
\newcommand{\ystep}{1.2}
\newcommand{\xtscale}{0.8}

\newcommand{\halfyOuter}{0.4}
\newcommand{\halfxOuter}{0.55}
\newcommand{\halfyInner}{0.3}
\newcommand{\halfxInner}{0.45}
\newcommand{\grayDark}{gray!50}
\newcommand{\grayLight}{gray!20}
\newcommand{\scalefactor}{1}

\centering
\scalebox{\treeclocksexamplecaling}{
\begin{tikzpicture}[thick,sibling distance=1em,
every tree node/.style={thick, rectangle, rounded corners, minimum height=0.8mm, minimum width=5mm, inner sep=2pt},
sibling distance=1pt,
level distance=20pt,
]
\tikzset{edge from parent/.append style={ thick}}
\pgfdeclarelayer{background}
\pgfdeclarelayer{foreground}
\pgfsetlayers{background,main,foreground}

\begin{scope}[shift={(0,0)}]
\begin{scope}[scale=\scalefactor]
\Tree [
 . \node[] (a1) {$\treethr{1},16,\bot$};
	[ .\node[] (a2) {$\treethr{2},20,9$};
		[.\node[] (a4) {$\treethr{4},23,18$}; 
			[ .\node[] (a8) {$\treethr{8},2,19$}; ] 
		]
		[ .\node[] (a5) {$\treethr{5},4,14$}; 
			[ .\node[] (a9) {$\treethr{9},10,4$}; 
				[ .\node[] (a11) {$\treethr{11},8,7$}; ]
				[ .\node[] (a12) {$\treethr{12},2,4$}; ] 
			] 
		]
		[ .\node[] (a6) {$\treethr{6},15,8$}; 
			[. \node[] (a10) {$\treethr{10},2,15$}; ] 
		]
		[ .\node[] (a7) {$\treethr{7},11,2$}; ]
	]
	[. \node[] (a3) {$\treethr{3},17,7$}; ]
]
\node[above of=a1] {\normalsize $\TreeClock_1$};
\end{scope}

\begin{pgfonlayer}{background}
\draw[smooth, draw=none, rounded corners, fill=\grayLight] ($ (a1) + (0, \halfyOuter) $) to ($ (a1) + (-\halfxOuter, \halfyOuter) $) to ($ (a4) + (-\halfxOuter, \halfyOuter) $) to ($ (a4) + (-\halfxOuter, -\halfyOuter) $) to ($ (a5) + (\halfxOuter, -\halfyOuter) $) to ($ (a5) + (\halfxOuter, \halfyOuter) $) to ($ (a2) + (\halfxOuter, -\halfyOuter) $) to ($ (a3) + (-\halfxOuter, -\halfyOuter) $) to ($ (a3) + (\halfxOuter, -\halfyOuter) $) to ($ (a3) + (\halfxOuter, \halfyOuter) $) to ($ (a1) + (\halfxOuter, \halfyOuter) $) to ($ (a1) + (0, \halfyOuter) $);

\draw[smooth, draw=none, rounded corners, fill=\grayDark] ($ (a1) + (0, \halfyInner) $) to ($ (a1) + (-\halfxInner, \halfyInner) $) to ($ (a2) + (-\halfxInner, \halfyInner) $) to ($ (a2) + (-\halfxInner, -\halfyInner) $) to ($ (a3) + (\halfxInner, -\halfyInner) $) to ($ (a3) + (\halfxInner, \halfyInner) $) to ($ (a1) + (\halfxInner, \halfyInner) $) to ($ (a1) + (0, \halfyInner) $);

\end{pgfonlayer}

\end{scope}

\begin{scope}[shift={(\xdisposition,0)}]
\begin{scope}[scale=\scalefactor]
\Tree [
. \node[] (b6) {$\treethr{12},25,\bot$};
	[ . \node[] (b5){$\treethr{5},8,20$};
		[ . \node[] (b8) {$\treethr{8},10,8$}; 
			[ .\node[] (b9) {$\treethr{9},16,5$}; 
				[. \node[] (b10){$\treethr{10},6,12$}; ] 
			] 
		]
		[ .\node[fill=\grayDark] (b1) {$\treethr{1},4,4$}; 
			[ .\node[fill=\grayDark] (b3) {$\treethr{3},10,4$}; ] 
		]
	]
	[ .\node[] (b7) {$\treethr{7},24,16$};
		[ .\node[] (b12) {$\treethr{4},31,20$}; ]
		[ .\node[] (b11) {$\treethr{11},15,7$}; 
			[ .\node[fill=\grayDark] (b2) {$\treethr{2},14,9$}; 
				[ .\node[] (b4) {$\treethr{6},15,8$}; ] 
			] 
		]
	]
]

\node[above of=b6] {\normalsize $\TreeClock_2$};
\end{scope}

\end{scope}

\begin{scope}[shift={(2.09*\xdisposition,0)}]
\begin{scope}[scale=\scalefactor]
\Tree [
.\node[] (c6) {$\treethr{12},25,\bot$};
	[ .\node[] (c1) {$\treethr{1},16,25$}; 
		[ .\node[] (c2) {$\treethr{2},20,9$}; 
			[ .\node[] (c4) {$\treethr{6},15,8$}; ] 
		] 
		[ .\node[] (c3) {$\treethr{3},17,7$}; ] 
	]
	[ .\node[] (c5) {$\treethr{5},8,20$}; 
		[ .\node[] (c8) {$\treethr{8},10,8$}; 
			[ .\node[] (c9) {$\treethr{9},16,5$}; 
				[ .\node[] (c10) {$\treethr{10},6,12$}; ]  
			]  
		] 
	]
	[ .\node[] (c7) {$\treethr{7},24,16$}; 
		[ .\node[] (c12) {$\treethr{4},31,20$}; ] 
		[ .\node[] (c11) {$\treethr{11},15,7$}; ] 
	]
]

\node[above of=c6] {\normalsize $\TreeClock'_2$};
\end{scope}

\begin{pgfonlayer}{background}
\draw[smooth, draw=none, rounded corners, fill=\grayDark] ($ (c1) + (0, \halfyInner) $) to ($ (c1) + (-\halfxInner, \halfyInner) $) to ($ (c2) + (-\halfxInner, \halfyInner) $) to ($ (c2) + (-\halfxInner, -\halfyInner) $) to ($ (c3) + (\halfxInner, -\halfyInner) $) to ($ (c3) + (\halfxInner, \halfyInner) $) to ($ (c1) + (\halfxInner, \halfyInner) $) to ($ (c1) + (0, \halfyInner) $);
\end{pgfonlayer}

\end{scope}

\end{tikzpicture}
}
\caption{
Illustration of $\TreeClock_2.\protect\FunctionCTJoin(\TreeClock_1)$.
Light gray marks the nodes of $\TreeClock_1$ whose time is compared to the time of the respective thread in $\TreeClock_2$ (i.e., the total iterations in \cref{line:routineupdatednodesforjoin_loop}).
Dark gray marks the nodes that are updating/being updated (i.e., the size of $\Stack$).
$\TreeClock'_{2}$ is the result of the join, where dark gray marks the sub-tree updated by $\protect\FunctionCTJoin$.
}
\label{fig:ctjoin}
\end{figure*}

%% file: figures/fig_ctmonotonecopy.tex
%!TEX root = ../main.tex

\begin{figure*}[t]
\small
\newcommand{\xdisposition}{5.6}
\newcommand{\ydisposition}{0}
\newcommand{\xtstep}{0.9}
\newcommand{\ytstep}{0.5}
\newcommand{\xstep}{2.8}
\newcommand{\ystep}{1.2}
\newcommand{\xtscale}{0.8}

\newcommand{\halfyOuter}{0.4}
\newcommand{\halfxOuter}{0.55}
\newcommand{\halfyInner}{0.3}
\newcommand{\halfxInner}{0.45}
\newcommand{\grayDark}{gray!50}
\newcommand{\grayLight}{gray!20}
\newcommand{\scalefactor}{1}

\centering
\scalebox{\treeclocksexamplecaling}{
\begin{tikzpicture}[thick,sibling distance=1em,
every tree node/.style={thick, rectangle, rounded corners, minimum height=0.8mm, minimum width=5mm, inner sep=2pt},
sibling distance=3pt,
level distance=20pt,
]
\tikzset{edge from parent/.append style={ thick}}
\pgfdeclarelayer{background}
\pgfdeclarelayer{foreground}
\pgfsetlayers{background,main,foreground}

\begin{scope}[shift={(0,0)}]

\begin{scope}[scale=\scalefactor]
\Tree [
.\node[] (a1) {$\treethr{1},28, \bot$};
	[ .\node[] (a2) {$\treethr{2},13,9$}; 
		[ .\node[] (a5) {$\treethr{5},8,11$};] 
	]
	[ .\node[] (a3) {$\treethr{3},14,7$};
		[ .\node[] (a6) {$\treethr{6},8,9$};  
			[ .\node[] (a10) {$\treethr{10},2,2$}; 
				[ .\node[] (a9) {$\treethr{9},9,6$};]
			] 
		] 
		[ .\node[] (a7) {$\treethr{7},8,4$}; 
			[ .\node[] (a11) {$\treethr{11},7,5$}; 
				[ .\node[] (a12) {$\treethr{12},15,2$}; ] 
			] 
		]
		[ .\node[] (a8) {$\treethr{8},8,2$}; ]
	]
	[ .\node[] (a4) {$\treethr{4},12,5$}; ]
]

\node[above of=a1] {\normalsize $\TreeClock_1$};
\end{scope}

\begin{pgfonlayer}{background}
\draw[smooth, draw=none, rounded corners, fill=\grayLight] ($ (a1) + (0, \halfyOuter) $) to ($ (a1) + (-\halfxOuter, \halfyOuter) $) to ($ (a2) + (-\halfxOuter, \halfyOuter) $) to ($ (a2) + (-\halfxOuter, -\halfyOuter) $) to ($ (a5) + (-\halfxOuter, -\halfyOuter) $) to ($ (a5) + (+\halfxOuter, -\halfyOuter) $) to ($ (a2) + (+\halfxOuter, -\halfyOuter) $) to ($ (a3) + (-\halfxOuter, -\halfyOuter) $) to ($ (a3) + (\halfxOuter, -\halfyOuter) $) to ($ (a4) + (-\halfxOuter, -\halfyOuter) $) to ($ (a4) + (\halfxOuter, -\halfyOuter) $) to ($ (a4) + (\halfxOuter, \halfyOuter) $) to ($ (a1) + (\halfxOuter, \halfyOuter) $) to ($ (a1) + (0, \halfyOuter) $);

\draw[smooth, draw=none, rounded corners, fill=\grayDark] ($ (a1) + (0, \halfyInner) $) to ($ (a1) + (-\halfxInner, \halfyInner) $) to ($ (a2) + (-\halfxInner, \halfyInner) $) to ($ (a2) + (-\halfxInner, -\halfyInner) $) to ($ (a5) + (-\halfxInner, -\halfyInner) $) to ($ (a5) + (+\halfxInner, -\halfyInner) $) to ($ (a2) + (+\halfxInner, -\halfyInner) $) to ($ (a3) + (-\halfxInner, -\halfyInner) $) to ($ (a3) + (\halfxInner+0.1, -\halfyInner) $) to ($ (a3) + (\halfxInner+0.1, \halfyInner) $) to ($ (a1) + (\halfxInner+0.1, -\halfyInner) $) to ($ (a1) + (\halfxInner+0.1, \halfyInner) $) to ($ (a1) + (-\halfxInner, \halfyInner) $);

\end{pgfonlayer}

\end{scope}

\begin{scope}[shift={(1.08*\xdisposition,0)}]

\begin{scope}[scale=\scalefactor]
\Tree [
.\node[fill=\grayDark] (b3) {$\treethr{3},14,\bot$};
	[ .\node[] (b6) {$\treethr{6},8,9$}; 
		[ .\node[fill=\grayDark] (b1) {$\treethr{1},4,2$}; 
			[ .\node[] (b4) {$\treethr{4},12,5$}; ] 
		]  
		[ .\node[] (b10) {$\treethr{10},2,2$};  
			[ .\node[] (b9) {$\treethr{9},9,6$}; ] 
		] 
	] 
	[ .\node[] (b7) {$\treethr{7},8,4$}; 
		[ .\node[] (b11) {$\treethr{11},7,5$}; 
			[ .\node[] (b12) {$\treethr{12},15,2$}; ] 
		] 
	]
	[ .\node[] (b8) {$\treethr{8},8,2$}; 
		[ .\node[fill=\grayDark] (b5) {$\treethr{5},4,4$}; 
			[ .\node[fill=\grayDark] (b2) {$\treethr{2},6,2$}; ]  
		] 
	]
]

\node[above of=b3] {\normalsize $\TreeClock_2$};
\end{scope}

\end{scope}

\begin{scope}[shift={(2.09*\xdisposition,0)}]

\begin{scope}[scale=\scalefactor]
\Tree [
.\node[] (c1) {$\treethr{1},28, \bot$};
	[ .\node[] (c2) {$\treethr{2},13,9$}; 
		[ .\node[] (c5) {$\treethr{5},8,11$};] 
	]
	[ .\node[] (c3) {$\treethr{3},14,7$};
		[ .\node[] (c6) {$\treethr{6},8,9$};  
			[ .\node[] (c10) {$\treethr{10},2,2$};  
				[ .\node[] (c9) {$\treethr{9},9,6$}; ] 
			] 
		] 
		[ .\node[] (c7) {$\treethr{7},8,4$}; 
			[ .\node[] (c11) {$\treethr{11},7,5$}; 
				[ .\node[] (c12) {$\treethr{12},15,2$}; ] 
			] 
		]
		[ .\node[] (c8) {$\treethr{8},8,2$}; ]
	]
	[ .\node[] (c4) {$\treethr{4},12,5$}; 
	]
]

\node[above of=c1] {\normalsize $\TreeClock'_2$};
\end{scope}

\begin{pgfonlayer}{background}
\draw[smooth, draw=none, rounded corners, fill=\grayDark] ($ (c1) + (0, \halfyInner) $) to ($ (c1) + (-\halfxInner, \halfyInner) $) to ($ (c2) + (-\halfxInner, \halfyInner) $) to ($ (c2) + (-\halfxInner, -\halfyInner) $) to ($ (c5) + (-\halfxInner, -\halfyInner) $) to ($ (c5) + (+\halfxInner, -\halfyInner) $) to ($ (c2) + (+\halfxInner, -\halfyInner) $) to ($ (c3) + (-\halfxInner, -\halfyInner) $) to ($ (c3) + (\halfxInner+0.1, -\halfyInner) $) to ($ (c3) + (\halfxInner+0.1, \halfyInner) $) to ($ (c1) + (\halfxInner + 0.1, -\halfyInner) $) to ($ (c1) + (\halfxInner + 0.1, \halfyInner) $) to ($ (c1) + (-\halfxInner, \halfyInner) $);
\end{pgfonlayer}

\end{scope}

\end{tikzpicture}
}
\caption{
Illustration of $\TreeClock_2.\protect\FunctionCTMonotoneCopy(\TreeClock_1)$.
Light gray marks the nodes of $\TreeClock_1$ whose time is compared to the time of the respective thread in $\TreeClock_2$ (i.e., the total iterations in \cref{line:routineupdatednodesforcopy_loop}).
Dark gray marks the nodes that are updating/being updated (i.e., the size of $\Stack$).
$\TreeClock'_{2}$ is the result of the copy, where dark gray marks the sub-tree updated by $\protect\FunctionCTMonotoneCopy$.
Node $(\treethr{3},14,\bot)$ (i.e., the root) of $\TreeClock_{2}$ is updated although $\treethr{3}$ has not progressed in $\TreeClock_{1}$, as it is placed under the new root $(\treethr{1},28,\bot)$ in $\TreeClock'_2$. %(\cref{line:routineupdatenodesforcopy_root}).
}
\label{fig:ctmonotonecopy}
\end{figure*}

%% file: applications.tex
\input{hb}
\input{other}
%\input{distributed}

%% file: hb.tex
%!TEX root=main.tex
\section{Tree Clocks for Happens-Before}\label{sec:race_detection}

%In this section we present $\FastHB$, a new data race detector based on the happens-before partial order.
%We proceed in two steps.
%First, we show how tree clocks can be used to compute the $\HB$ partial order, and prove their optimality guarantee (\cref{subsec:tc_hb}).
%Then we incorporate read and write epochs to the above algorithm to arrive at $\FastHB$,  and show that $\FastHB$ is optimal for race detection.

%\subsection{Tree Clocks for Happens-Before}\label{subsec:tc_hb}

Let us see how tree clocks are employed for computing the $\HB$ partial order.
We start with the following observation.

\smallskip
\begin{restatable}[Monotonicity of copies]{lemma}{lemcopymonotonicity}\label{lem:copy_monotonicity}
Right before \cref{algo:hb_basic} processes a lock-release event $\langle \Thread, \Release(\ell)\rangle$,
we have $\CTC_{\ell} \CSmaller \CTC_{\Thread}$.
\end{restatable}

\Paragraph{Tree clocks for $\HB$.}
\cref{algo:hb} shows the algorithm for computing $\HB$ using the tree clock data structure for implementing vector times.
When processing a lock-acquire event, the vector-clock join operation has been replaced by a tree-clock join.
Moreover, in light of \cref{lem:copy_monotonicity}, when processing a lock-release event, the vector-clock copy operation has been replaced by a tree-clock monotone copy.
\begin{arxiv}
We refer to \cref{sec:app_hb_example} for an example run of \cref{algo:hb} on a trace $\Trace$, showing how tree clocks grow during the execution.
\end{arxiv}
\input{algorithms/algo_hb}

\Paragraph{Correctness.}
We now state the correctness of \cref{algo:hb}, i.e., we show that the algorithm indeed computes the $\HB$ partial order.
We start with two monotonicity invariants of tree clocks.

\smallskip
\begin{restatable}{lemma}{lemtcmonotonicity}\label{lem:tc_monotonicity}
Consider any tree clock $\CTC$ and node $u$ of $\CTC.\GTree$.
For any tree clock $\CTC'$, the following assertions hold.
\begin{compactenum}
\item\label{item:monotonicity1}
\emph{Direct monotonicity:} If $u.\clock \leq \CTC'.\FunctionCTGet(u.\tid)$ then for every descendant $w$ of $u$ we have that $w.\clock \leq \CTC'.\FunctionCTGet(w.\tid)$.
\item\label{item:monotonicity2}
\emph{Indirect monotonicity:} If $v.\pclock \leq \CTC'.\FunctionCTGet(u.\tid)$ where $v$ is a child of $u$ then for every descendant $w$ of $v$ we have that $w.\clock\leq \CTC'.\FunctionCTGet(w.\tid)$.
\end{compactenum}
\end{restatable}

The following lemma follows from the above invariants
and establishes that \cref{algo:hb} with tree clocks computes the correct timestamps on all events, i.e., the correctness of tree clocks for $\HB$.

\smallskip
\begin{restatable}{lemma}{lemhbcor}\label{lem:hb_cor}
When \cref{algo:hb} processes an event $\Event$,
the vector time stored in the tree clock $\CTC_{\tid(\Event)}$ is $\POTime{\ord{\Trace}{\HB}}{\Event}$.
\end{restatable}

%\smallskip
%\begin{restatable}{lemma}{lemhbcor}\label{lem:hb_cor}
%The following assertions hold.
%\begin{compactenum}
%\item After \cref{algo:hb} processes an event $\langle \Thread, %\Acquire(\ell) \rangle$
%we have $\CTC_{\Thread} = \CTC_{\Thread} \CJoin \CTL_{\ell}$.
%\item After \cref{algo:hb} processes an event $\langle \Thread, %\Release(\ell) \rangle$
%we have $\CTC_{\ell} = \CTC_{\Thread}$.
%\end{compactenum}
%\end{restatable}

\Paragraph{Data structure optimality.}
Just like vector clocks, computing $\HB$ with tree clocks takes $\Theta(\NumEvents\cdot \NumThreads)$ time in the worst case,
and it is known that this quadratic bound is likely to be tight for common applications such as dynamic race prediction~\cite{Kulkarni2021}.
However, we have seen that tree clocks can take sublinear time on join and copy operations, whereas vector clocks always require time linear in the size of the vector (i.e., $\Theta(\NumThreads)$).
A natural question arises: is there a more efficient data structure than tree clocks?
More generally, what is the most efficient data structure for the $\HB$ algorithm to represent vector times?
To answer this question, we define \emph{vector-time work}, which gives a lower bound on the number of data structure operations that $\HB$ has to perform regardless of the actual data structure used to store vector times.
Then, we show that tree clocks match this lower bound, hence achieving optimality for $\HB$.

\Paragraph{Vector-time work.}
Consider the general $\HB$ algorithm (\cref{algo:hb_basic}) and let 
$\mathfrak{D}=\{ \CTC_{1}, \CTC_{2},\dots, \CTC_{m} \}$ be the set of the 
vector-time data structures used.
Consider the execution of the algorithm on a trace $\Trace$.
Given some $1\leq i \leq |\Trace|$, we let $C_{j}^i$ denote the vector time 
represented by $\CTC_{j}$ after the algorithm has processed the $i$-th event of $\Trace$.
We define the \emph{vector-time work} (or \emph{vt-work}, for short) on $\Trace$ as
\[
\VTWork(\Trace) = \sum_{1\leq i\leq |\Trace|} \sum_{j} | \Thread\in \Threads{} \colon C_{j}^{i-1}(\Thread)\neq  C_{j}^{i}(\Thread)|.
\] 
In words, for every processed event, we add the number of vector-time entries that change as a result of processing the event, and
$\VTWork(\Trace)$ counts the total number of entry updates in the overall course of the algorithm.
Note that vt-work is independent of the data structure used to represent each $\CTC_{j}$,
and  satisfies the inequality 
\[
\NumEvents \leq \VTWork(\Trace) \leq \NumEvents\cdot \NumThreads.
\]
as with every event of $\Trace$ the algorithm updates one of $\CTC_{j}$.

\Paragraph{Vector-time optimality.}
Given an input trace $\Trace$, we denote by $\Time_{\DataStructure}(\Trace)$ the time taken by the $\HB$ algorithm (\cref{algo:hb_basic}) to process $\Trace$ using the data structure $\DataStructure$ to store vector times.
Intuitively, $\VTWork(\Trace)$ captures the number of times that instances of $\DataStructure$ change state.
For data structures that represent vector times explicitly, 
$\VTWork(\Trace)$ presents a natural lower bound for $\Time_{\DataStructure}(\Trace)$.
Hence, we say that the data structure $\DataStructure$ is \emph{vt-optimal} if $\Time_{\DataStructure}(\Trace)=O(\VTWork(\Trace))$.
It is not hard to see that vector clocks  are not vt-optimal,
i.e., taking $\DataStructure=\VectorClock$ to be the vector clock data structure, one can construct simple traces $\Trace$ where 
$\VTWork(\Trace)=O(\NumEvents)$
but $\Time_{\DataStructure}(\Trace)=\Omega(\NumEvents\cdot \NumThreads)$,
and thus the running time is $\NumThreads$ times more than the vt-work that must be performed on $\Trace$.
In contrast, the following theorem states that tree clocks are vt-optimal.
%, i.e., when $\DataStructure=\TreeClock$ is the tree clock data structure, we always have $\Time_{\TreeClock}(\Trace)=O(\VTWork(\Trace))$.

\smallskip
\begin{restatable}[Tree-clock Optimality]{theorem}{thmvtoptimality}\label{thm:vtoptimality}
For any input trace $\Trace$,  we have $\Time_{\TreeClock}(\Trace)=O(\VTWork(\Trace))$.
\end{restatable}

The key observation behind \cref{thm:vtoptimality} is that, when $\HB$ uses tree clocks, the total number of tree-clock entries that are accessed over all join
and monotone copy operations (i.e., the sum of the sizes of the light-gray areas in \cref{fig:ctjoin} and \cref{fig:ctmonotonecopy}) is $\leq 3\cdot \VTWork(\Trace)$.

\smallskip
\begin{remark}
\cref{thm:vtoptimality} establishes \emph{strong optimality} for tree clocks, in the sense that they are vt-optimal \emph{on every input}.
This is in contrast to usual notions of optimality that is guaranteed on \emph{only some} inputs.
\end{remark}

%% file: algorithms/algo_hb.tex
%!TEX root = ../main.tex

%\small
\begin{algorithm}[h]
\vspace*{-\multicolsep}
\begin{multicols}{2}
\MyProcedure{\ProcAcquire{$\Thread$, $\ell$}}{
$\CTC_{\Thread}.\FunctionCTJoin(\CTC_{\ell})$
}
% \BlankLine
\BlankLine
\MyProcedure{\ProcRelease{$\Thread$, $\ell$}}{
$\CTC_{\ell}.\FunctionCTMonotoneCopy(\CTC_{\Thread})$
}
\end{multicols}
\vspace*{-0.5\multicolsep}
\normalsize
\caption{
$\HB$ with tree clocks.
}
\label{algo:hb}
\end{algorithm}
\normalsize

%% file: other.tex
\section{Tree Clocks in Other Partial Orders}\label{subsec:other}

\input{shb}
\input{maz}

%% file: shb.tex
%!TEX root=main.tex
\subsection{Schedulable-Happens-Before}\label{subsec:shb}

%\Paragraph{The $\SHB$ partial order.}
$\SHB$ is a strengthening of $\HB$, introduced recently~\cite{Mathur18} in the context of race detection.
%, and has the property that for every two events $\Event_1$, $\Event_2$ of a trace $\Trace$, if $\Event_1\parallel_{\SHB}^{\Trace} \Event_2$, 
%then $\Trace$ can be soundly reordered to a trace $\Trace^*$ that ends with $\Event_1, \Event_2$.
%The partial order $\SHB$ is defined as follows.
Given a trace $\Trace$ and a read event $\Read$ let $\LastWrite_{\Trace}(\Read)$ be the last write event of $\Trace$ before $\Read$ with $\Variable(\Write)=\Variable(\Read)$.
$\SHB$ is the smallest partial order that satisfies the following.
\begin{compactenum}
\item\label{item:shb1} $ \hb{\Trace} \subseteq  \shb{\Trace}$.
\item\label{item:shb2} for every read event $\Read$, we have $\LastWrite_{\Trace}(\Read)\shb{\Trace} \Read$.
\end{compactenum}

\Paragraph{Algorithm for $\SHB$.}
Similarly to $\HB$, the $\SHB$ partial order is computed by a single pass of the input trace $\Trace$ using vector-times~\cite{Mathur18}.
The $\SHB$ algorithm processes synchronization events (i.e., $\Acquire(\ell)$ and $\Release(\ell)$) similarly to $\HB$.
In addition, for each variable $x$, the algorithm maintains a data structure $\LWT_x$ that stores the vector time of the latest write event on $x$.
When a write event $\Write(x)$ is encountered, the vector time $\CTC_{\tid(\Write)}$ is copied to $\LWT_{x}$.
In turn, when a read event $\Read(x)$ is encountered the algorithm joins $\LWT_{x}$ to $\CTC_{\tid(\Read)}$.

\Paragraph{$\SHB$ with tree clocks.}
Tree clocks can directly be used as the data structure to store vector times in the $\SHB$ algorithm.
We refer to \cref{algo:shb} for the pseudocode.
The important new component is the function $\FunctionCTCopyCheckMonotone$ in \cref{line:copycheckmonotone} that copies the vector time of $\CTC_{\Thread}$ to $\LWT_x$.
In contrast to $\FunctionCTMonotoneCopy$, 
this copy  is not guaranteed to be monotone, i.e., we might have $\LWT_x\not \CSmaller\CTC_{\Thread}$.
Note, however, that using tree clocks, this test requires only constant time.
Internally, $\FunctionCTCopyCheckMonotone$ performs $\FunctionCTMonotoneCopy$ if $\LWT_x \CSmaller\CTC_{\Thread}$ (running in sublinear time), otherwise it performs a deep copy for the whole tree clock (running in linear time).
In practice, we expect that most of the times $\FunctionCTCopyCheckMonotone$ results in  $\FunctionCTMonotoneCopy$ and thus is very efficient.
The key insight is that if $\FunctionCTMonotoneCopy$ is not used, then $\LWT_x\not \CSmaller\CTC_{\Thread}$
and thus we have a race $(\LastWrite_{\Trace}(\Read),\Read)$.
Hence, the number of times a deep copy is performed is bounded by the number of write-read races in $\Trace$ between a read and its last write.
%Our evaluation shows that these cases are rare.

\input{algorithms/algo_shb}

%% file: algorithms/algo_shb.tex
%!TEX root=../main.tex

%\small
\begin{algorithm}[t]
% \BlankLine
% \BlankLine
%\tcp*[f]{}
\SetInd{0.3em}{0.3em}
\addtolength{\columnsep}{-40pt}
\vspace*{-\multicolsep}
\begin{multicols}{2}
\MyProcedure{\ProcAcquire{$\Thread$, $\ell$}}{
$\CTC_{\Thread}.\FunctionCTJoin(\CTL_{\ell})$
}
% \BlankLine
\BlankLine
\MyProcedure{\ProcRead{$\Thread$, $x$}}{
$\CTC_{\Thread}.\FunctionCTJoin(\CTLW_{x})$
}
% \BlankLine
\BlankLine
\MyProcedure{\ProcRelease{$\Thread$, $\ell$}}{
$\CTL_{\ell}.\FunctionCTMonotoneCopy(\CTC_{\Thread})$\\
%$\CTC_{\Thread}.\FunctionCTIncrement()$
}

% \BlankLine
\BlankLine
\MyProcedure{\ProcWrite{$\Thread$, $x$}}{
$\CTLW_{x}.\FunctionCTCopyCheckMonotone(\CTC_{\Thread})$\label{line:copycheckmonotone}%\tcp*{Deep or Monotone}
%$\CTC_{\Thread}.\FunctionCTIncrement()$\\
}
\end{multicols}
\vspace*{-0.5\multicolsep}
\normalsize
\caption{
$\SHB$ with tree clocks.
}
\label{algo:shb}
\end{algorithm}
\normalsize

%% file: maz.tex
\subsection{The Mazurkiewicz Partial Order}\label{subsec:maz}
The Mazurkiewicz partial order ($\Maz$)~\cite{Mazurkiewicz87} has been the canonical  way to represent concurrent executions algebraically using an independence relation that defines which events can be reordered.
This algebraic treatment allows to naturally lift language-inclusion problems from the verification of sequential programs to concurrent programs~\cite{Bertoni1989}.
As such, it has been the most studied partial order in the context of concurrency, with deep applications in dynamic analyses~\cite{Netzer1990,Flanagan2008,Mathur2020}, ensuring consistency~\cite{Shasha1988} and stateless model checking~\cite{Flanagan2005}.
In shared memory concurrency, the standard independence relation deems two events as dependent if they conflict, and independent otherwise~\cite{Godefroid1996}.
In particular, $\Maz$ is the smallest partial order that satisfies the following conditions.
\begin{compactenum}
\item\label{item:maz1} $ \hb{\Trace} \subseteq  \maz{\Trace}$.
\item\label{item:maz2} for every two events $\Event_1,\Event_2$ such that $\Event_1\trord{\Trace}\Event_2$ and $\Confl{\Event_1}{\Event_2}$, we have $\Event_1\maz{\Trace} \Event_2$.
\end{compactenum}

\Paragraph{$\Maz$ with tree clocks.}
The algorithm for computing $\Maz$ is similar to that for $\SHB$.
The main difference is that $\Maz$ includes read-to-write orderings, and thus we need to store additional vector times $\CTR_{\Thread,x}$ of the last event $\Read(x)$ of thread $\Thread$.
In addition, we use the set $\LastReads_x$ to store the threads that have executed a $\Read(x)$ event after the latest $\Write(x)$ event so far.
This allows us to only spend computation time in the first read-to-write ordering, as orderings between the read event and later write events follow transitively via intermediate write-to-write orderings.
Overall, this approach yields the efficient time complexity $O(\NumEvents\cdot \NumThreads)$ for $\Maz$, similarly to $\HB$ and $\SHB$.
We refer to \cref{algo:maz} for the pseudocode.

\input{algorithms/algo_maz}

%% file: algorithms/algo_maz.tex
%!TEX root=../main.tex

%\small
\begin{algorithm}[t]
\SetInd{0.3em}{0.5em}
\vspace*{-\multicolsep}
\begin{multicols}{2}

\MyProcedure{\ProcAcquire{$\Thread$, $\ell$}}{
$\CTC_{\Thread}.\FunctionCTJoin(\CTL_{\ell})$
}
% \BlankLine
\BlankLine
\MyProcedure{\ProcRead{$\Thread$, $x$}}{
$\CTC_{\Thread}.\FunctionCTJoin(\CTLW_{x})$\\
$\CTR_{\Thread,x}.\FunctionCTMonotoneCopy(\CTC_{\Thread})$\\
$\LastReads_x\gets \LastReads_x\cup \{\Thread\}$\\
}
% \BlankLine
\BlankLine
\MyProcedure{\ProcRelease{$\Thread$, $\ell$}}{
$\CTL_{\ell}.\FunctionCTMonotoneCopy(\CTC_{\Thread})$\\
}
% \BlankLine
\BlankLine
\MyProcedure{\ProcWrite{$\Thread$, $x$}}{
$\CTC_{\Thread}.\FunctionCTJoin(\CTLW_{x})$\\
\ForEach{$\Thread'\in \LastReads_x$}{
$\CTC_{\Thread}.\FunctionCTJoin(\CTR_{\Thread',x})$\\
}
$\CTLW_{x}.\FunctionCTMonotoneCopy(\CTC_{\Thread})$\\
$\LastReads_x\gets \emptyset$\\
}
\end{multicols}
\vspace*{-0.5\multicolsep}
\normalsize
\caption{
$\Maz$ with tree clocks.
}
\label{algo:maz}
\end{algorithm}
\normalsize

%%\small
%\begin{algorithm}[t]
%% \BlankLine
%% \BlankLine
%%\tcp*[f]{}
%\SetInd{0.3em}{0.5em}
%\MyProcedure{\ProcAcquire{$\Thread$, $\ell$}}{
%$\CTC_{\Thread}.\FunctionCTJoin(\CTL_{\ell})$
%}
%% \BlankLine
%\BlankLine
%\MyProcedure{\ProcRelease{$\Thread$, $\ell$}}{
%$\CTL_{\ell}.\FunctionCTMonotoneCopy(\CTC_{\Thread})$\\
%}
%% \BlankLine
%\BlankLine
%\MyProcedure{\ProcRead{$\Thread$, $x$}}{
%$\CTC_{\Thread}.\FunctionCTJoin(\CTLW_{x})$\\
%$\CTR_{\Thread,x}.\FunctionCTMonotoneCopy(\CTC_{\Thread})$\\
%$\LastReads_x\gets \LastReads_x\cup \{\Thread\}$\\
%}
%% \BlankLine
%\BlankLine
%\MyProcedure{\ProcWrite{$\Thread$, $x$}}{
%$\CTC_{\Thread}.\FunctionCTJoin(\CTLW_{x})$\\
%\ForEach{$\Thread'\in \LastReads_x$}{
%$\CTC_{\Thread}.\FunctionCTJoin(\CTR_{\Thread',x})$\\
%}
%$\CTLW_{x}.\FunctionCTMonotoneCopy(\CTC_{\Thread})$\\
%$\LastReads_x\gets \emptyset$\\
%}
%\normalsize
%\caption{
%$\Maz$ with tree clocks.
%}
%\label{algo:maz}
%\end{algorithm}
%\normalsize

%% file: experiments.tex
%!TEX root=main.tex

\section{Experiments}\label{sec:experiments}

In this section we report on an implementation and experimental evaluation of the tree clock data structure.
The primary goal of these experiments is to evaluate the practical advantage of tree clocks over the vector clocks for keeping track of logical times in a concurrent program executions.

\Paragraph{Implementation.}
Our implementation is in Java and closely follows \cref{algo:clock_tree}.
\camera{
The tree clock data structure is represented as two arrays of length $\NumThreads$ (number of threads), the first one encoding the shape of the tree and the second one encoding the integer timestamps as in a standard vector clock.
}
For efficiency reasons, recursive routines have been made iterative.

%The thread map $\ThreadMap$ is implemented as a list with random access,
%while the $\ChildrenQ(u)$ array storing the children of node $u$ is implemented as a doubly linked list.
%We also implemented \djitp style optimizations~\cite{Pozniansky03} for both tree and vector clocks.

\Paragraph{Benchmarks.}
Our benchmark set consists of standard benchmarks found in benchmark \camera{suites} and recent literature.
In particular, we used the Java benchmarks from the IBM Contest suite~\cite{Farchi03}, Java Grande suite~\cite{Smith01},
DaCapo~\cite{Blackburn06}, and
SIR~\cite{doESE05}.
In addition, we used OpenMP benchmark programs, whose execution lenghts and number of threads can be tuned, from DataRaceOnAccelerator~\cite{schmitz2019dataraceonaccelerator}, DataRaceBench \cite{liao2017dataracebench}, OmpSCR~\cite{dorta2005openmp} and the NAS parallel benchmarks~\cite{nasbenchmark}, 
as well as large OpenMP applications contained in the following benchmark suites: CORAL~\cite{coral2, coral}, ECP proxy applications~\cite{ecp}, and Mantevo project~\cite{mantevo}.
Each benchmark was instrumented and executed in order to log a single concurrent trace, using the tools RV-Predict~\cite{rvpredict} (for Java programs) and ThreadSanitizer~\cite{threadsanitizer} (for OpenMP programs).
Overall, this process yielded a large set of \camera{$153$} benchmark traces that were used in our evaluation.
\cref{tab:trace_stats} presents aggregate information about
%the number of threads, locks, variables and events, and the relative number of access (read or write) events and synchronization (acquire or release) events in 
the benchmark traces generated. 
\begin{arxiv} \camera{Information on the individual traces is provided in \cref{tab:trace-info} in the \cref{appsec:benchmarks}.} \end{arxiv} 
\begin{asplos} \camera{Information on the individual traces is provided in our technical report \cite{arxiv}.} \end{asplos}

\input{tables/trace_stats}

\input{figures/fig_time_comparison}

\Paragraph{Setup.}
Each trace was processed for computing each of the $\Maz$, $\SHB$ and $\HB$ partial orders using both tree clocks and the standard vector clocks.
This allows us to directly measure the speedup conferred by tree clocks in computing the respective partial order, which is the goal of this paper.

As the computation of these partial orders is usually the first component of any analysis, in general, we evaluated the impact of the conferred speedup in an overall analysis as follows.
For each pair of conflicting events $\Event_1, \Event_2$, we computed whether these events are concurrent wrt the corresponding partial order (e.g., whether $\Event_1 \unordhb{\Trace} \Event_2$).
This test is performed in dynamic race detection (in the cases of $\HB$ and $\SHB$) where such pairs constitute data races, as well in stateless model checking (in the case of $\Maz$) where the model checker identifies such event pairs and attempts to reverse their order on its way to exhaustively enumerate all Mazurkiewicz traces of the concurrent program.
For a fair comparison, in the case of $\HB$ we used common epoch optimizations~\cite{Flanagan09} to speed up the analysis for both tree clocks and vector clocks (recall \cref{rem:epochs}).
For consistency, every measurement was repeated 3 times and the average time was reported.

\Paragraph{Running times.}
For each partial order,
\cref{tab:speedups} shows the average speedup over all benchmarks, both with and without the analysis component.
We see that tree clocks are very effective in reducing the running time of the computation of all 3 partial orders, with the most significant impact being on $\SHB$ where the average speedup is 2.53 times.
For the cases of $\Maz$ and $\SHB$, this speedup also lead to a significant speedup in the overall analysis time.
On the other hand, although $\HB$ with tree clocks is about 2 times faster than with vector clocks, this speedup has a smaller effect on the overall analysis time.
The reason behind this observation is straightforward:
$\SHB$ and $\Maz$ are much more computationally-heavy, as they are defined using all types of events; on the other hand, $\HB$ is defined only on synchronization events ($\Acquire$ and $\Release$)
and on average, only $\simeq9.5\%$ of the events are synchronization events on our benchmark traces.
Since our analysis considers all events, the $\HB$-computation component occupies a smaller fraction of the overall analysis time.
We remark, however, that for programs that are more synchronization-heavy, or for analyses that are more lightweight (e.g., when checking for data races on a specific variable as opposed to all variables), the speedup of tree clocks will be larger on the whole analysis.
Indeed, \cref{fig:ratio_sync} shows the obtained speedup on the total analysis time for $\HB$ as a function of synchronization events.
We observe a trend for the speedup to increase as the percentage of synchronization events increases in the trace. A further observation is that speedup is prominent when the number of threads are large.

\input{tables/tab_speedups}

\input{figures/fig_ratio_sync}

\cref{fig:time_comparison} gives a more detailed view of the tree clocks vs vector clocks times across all benchmarks.
We see that tree clocks almost always outperform vector clocks on all partial orders, and in some cases by large margins.
Interestingly, the speedup tends to be larger on more demanding benchmarks (i.e., on those that take more time).
In the very few cases tree clocks are slower, this is only by a small factor.
These are traces where the sub-linear updates of tree clocks only yield a small potential for improvement, which does not justify the overhead of maintaining the more complex tree data structure (as opposed to a vector). 
Nevertheless, overall tree clocks consistently deliver a generous speedup to each of $\Maz$, $\HB$ and $\SHB$.
Finally, we remark that all these speedups come directly from just replacing the underlying data structure, without any attempt to optimize the algorithm that computes the respective partial order, or its interaction with the data structure.

\input{figures/fig_vtwork}
\input{figures/fig_vtwork_histogram}

\Paragraph{Comparison with vt-work.}
We also investigate the total number of entries updated using each of the data structures.
Recall that the metric $\VTWork(\Trace)$ (\cref{sec:race_detection})
measures the minimum amount of updates
that any implementation of the vector time must perform
when computing the $\HB$ partial order.
We can likewise define the metrics $\TCWork(\Trace)$ and
$\VCWork(\Trace)$ corresponding to the number of entries updated
when processing each event when using respectively the data structures
tree clocks and vector clocks.
These metrics are visualized in \cref{fig:vtwork} for \camera{computing the $\HB$ partial order in} our benchmark suite.
The figure shows that the $\VCWork(\Trace)/ \VTWork(\Trace)$ ratio is often considerably large.
In contrast, the ratio $\TCWork(\Trace)/ \VTWork(\Trace)$ is typically significantly smaller.
The differences in running times between vector and tree clocks reflect the discrepancies between $\TCWork(\cdot)$ and $\VCWork(\cdot)$.
%In fact, the benchmarks which have the highest ratios $\TCWork(\Trace)/\VCWork(\Trace)$ 
%also show a corresponding high speed-up
%(cassandra, tradebeans, hsqldb and raxextended).
Next, recall the intuition behind the optimality of tree clocks (\cref{thm:vtoptimality}), namely that $\TCWork(\Trace)\leq 3\cdot  \VTWork(\Trace)$.
\cref{fig:vtwork} confirms this theoretical bound, 
as the $\TCWork(\Trace)/ \VTWork\allowbreak(\Trace)$ ratio stays nicely upper-bounded by $3$ while the $\VCWork(\Trace)/\allowbreak \VTWork(\Trace)$ ratio grows to \camera{nearly} $100$.
Interestingly, for some benchmarks we have $\TCWork(\Trace)\simeq 2.99\cdot  \VTWork(\Trace)$,
i.e., these benchmarks push tree clocks to their vt-work upper-bound. 
\camera{
Going one step further, \cref{fig:vtwork-histogram} shows the ratio $\VCWork(\Trace)/ \TCWork(\Trace)$ for each partial order in our dataset.
The results indicate that vector clocks perform a lot of unnecessary work compared to tree clocks,
and experimentally demonstrate the source of speedup on tree clocks.
Although \cref{fig:vtwork-histogram} indicates that the potential for speedup can be large (reaching \hcomment{$55 \times$}),
the actual speedup in our experiments (\cref{fig:time_comparison}) is smaller, as a single tree clock operation is more computationally heavy than a single vector clock operation.
%Nevertheless, the speedup is present and consistent across our dataset.
}

\Paragraph{\camera{Scalability.}}
\camera{
To get a better insight on the scalability of tree clocks, we performed a set of controlled experiments on custom benchmarks, by controlling the number of threads and the number of locks parameters while keeping the communication patterns constant. 
Each trace consists of $10$M events, while the number of threads varies between $10$ and $360$. 
The traces are generated in a way such that a randomly chosen thread performs two consecutive operations, $acq(l)$ followed by a $rel(l)$, on a randomly (when applicable) chosen lock $l$. 
We have considered four cases:
(a)~all threads communicate over a single common lock (single lock);
(b)~similar to (a), but there are 50 locks, and $20\%$ of the threads are $5$ times more likely to perform an operation compared to the rest of the threads (50 locks, skewed);
(c)~$\NumThreads-1$ client threads communicate with a single server thread via a dedicated lock per thread (star topology);
(d)~similar to (a), but every pair of threads communicates over a dedicated lock (pairwise communication).
%The random selection process is ensured to adhere to the underlying communication pattern. 
%For instance, the set of available locks to a given thread may be restricted by the underlying communication pattern and this is taken into account in the random selection process. 
%The generated traces are based on four communication patterns: (a) a standard clique topology with one lock (b) a star topology where every thread utilizes a unique lock to communicate with the central thread (c) a topology in which every pair of threads use a unique lock to communicate with each other (d) a skewed communication pattern in which $20\%$ of the threads are $5$ times more likely to perform an operation compared to the rest of the threads and there are in total $50$ locks that are utilized to communicate between threads.
}
\input{figures/fig_scalability}
\camera{
\cref{fig:scalability}  shows the obtained results.
Scenario (a) shows how tree clocks have a constant-factor speedup over vector clocks in this setting.
As we move to more locks in scenario (b), thread communication becomes more independent and the benefit of tree clocks may slightly diminish.
As a subset of the threads is more active than the rest, timestamps are frequently exchanged through them, making tree clocks faster in this setting as well.
Scenario (c) represents a case in which tree clocks thrive: while the time taken by vector clocks increases with the number of threads, it stays constant for tree clocks.
This is because the star topology implies that, on average, every tree clock join and copy operation only affects a constant number of tree clock entries, despite the fact that every thread is aware of the state of every other thread.
Intuitively, the star communication topology \emph{naturally} affects the shape of the tree to (almost) a star, which leads to this effect.
Finally, scenario (d) represents the worst case for tree clocks as all pairs of threads can communicate with each other and the communication is conducted via a unique lock per thread pair. 
This pattern nullifies the benefit of tree clocks, while their increased complexity results in a general slowdown.
However, even in this worst-case scenario, the difference between tree clocks and vector clocks remains relatively small.
%The results indicate that in the cases of (a) and (d), tree clocks outperform vector clocks and the difference increases linearly with the number of threads. In the case of (b), vector clocks performance degrades linearly with increasing number of threads whereas tree clocks performance stays stable. This is due to that at every join operation vector clocks need to iterate over all the entries in their clocks whereas tree clocks need to update only a single entry since all threads communicate with only one thread (i.e, the central thread) and they use unique locks. The benchmark (c) represents the worst-case scenario for tree clocks as all pairs of threads can communicate with each other and the communication is conducted via a unique lock per thread pair. 
}

%% file: tables/trace_stats.tex
% \begin{table}
% \setlength\tabcolsep{3pt}
% \renewcommand{\arraystretch}{1.2}
% \small
% \centering
% \begin{tabular}{|r|r||r|r|}
% \hline
% % \hline
% Threads & 3 / 222 / 31 & Events & 51 / 2.1B / 227M \\
% Locks & 1 / 60.5k / 688 & Sync. Events (\%) & 0.0 / 44.4 / 9.5 \\
% Variables & 18 / 37.8M / 1.8M & R/W Events (\%) & 55.6 / 100 / 90.5  \\
% \hline
% \end{tabular}
% \caption{
% Trace Statistics
% }
% \label{tab:trace_stats}
% \end{table}

\begin{table}[h]
\setlength\tabcolsep{2pt}
\renewcommand{\arraystretch}{1.2}
\small
\centering
\begin{tabular}{|r|r:r:r||r|r:r:r|}
\hline
& \textbf{Min} & \textbf{Max} & \textbf{Mean} &
& \textbf{Min} & \textbf{Max} & \textbf{Mean} \\
\hline
\hline
Threads & 3 & 222 & 31 & Events & 51 & 2.1B & 227M \\
Locks & 1 & 60.5k & 688 & Sync. Events (\%) & 0.0 & 44.4 & 9.5 \\
Variables & 18 & 37.8M & 1.8M & R/W Events (\%) & 55.6 & 100 & 90.5  \\
\hline
\end{tabular}
\caption{
Trace Statistics
}
\label{tab:trace_stats}
\end{table}

%% file: figures/fig_time_comparison.tex
\begin{figure*}
\def\scatterscale{0.28}
\centering
\begin{subfigure}[b]{0.3\textwidth}
\includegraphics[scale=\scatterscale]{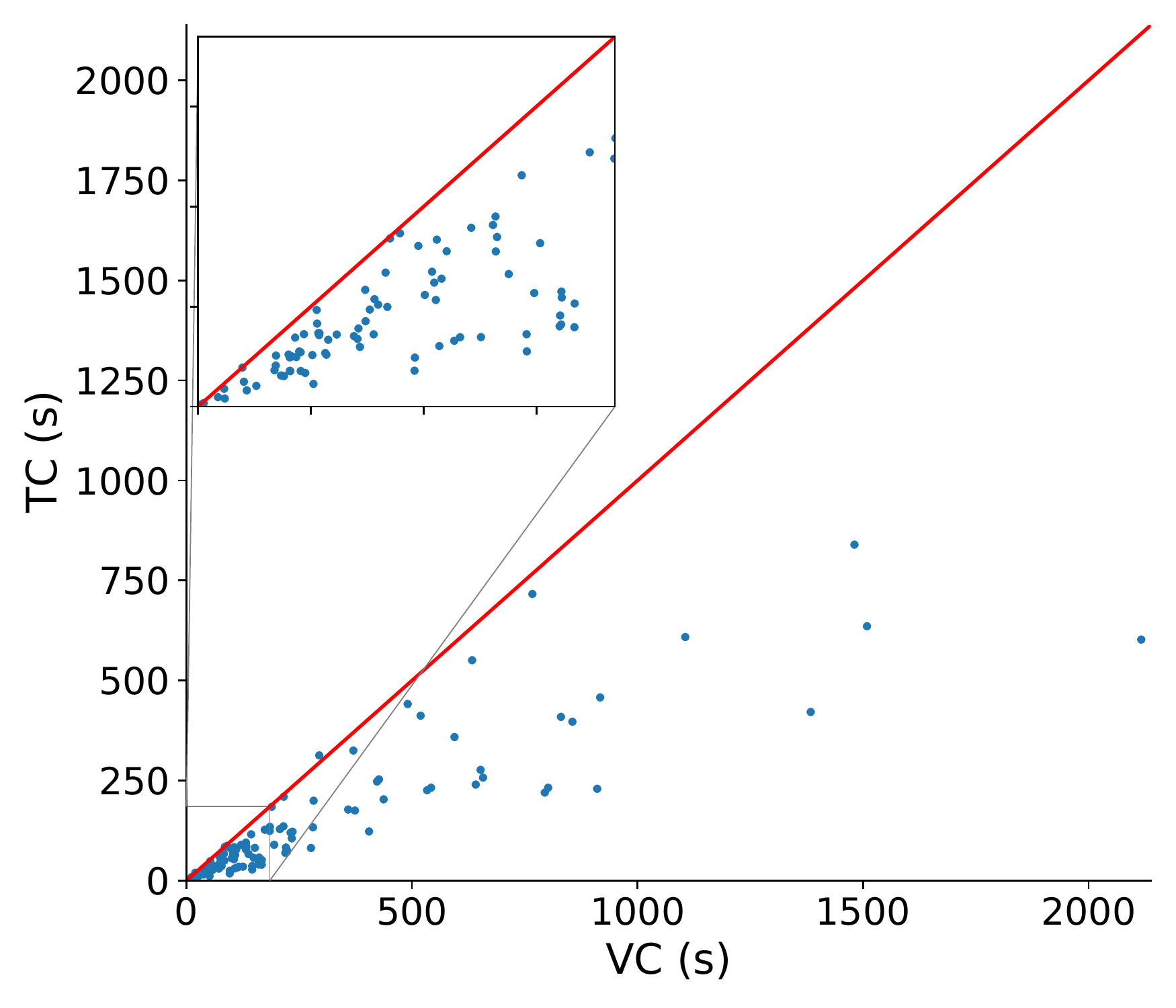}
\label{subfig:maz_nr}
\caption{$\Maz$}
\end{subfigure}
\begin{subfigure}[b]{0.3\textwidth}
\includegraphics[scale=\scatterscale]{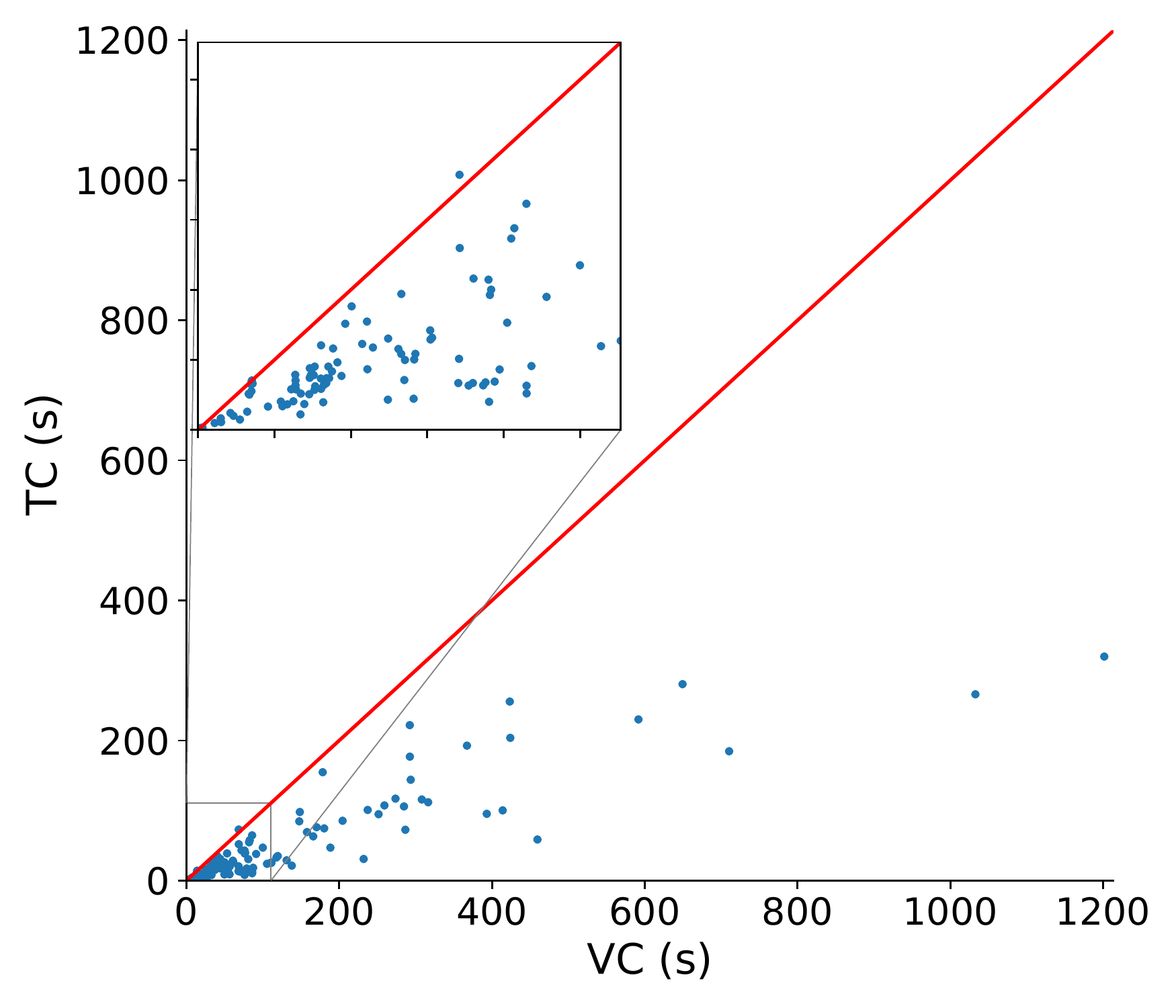}
\label{subfig:shb_nr}
\caption{$\SHB$}
\end{subfigure}
\begin{subfigure}[b]{0.3\textwidth}
\includegraphics[scale=\scatterscale]{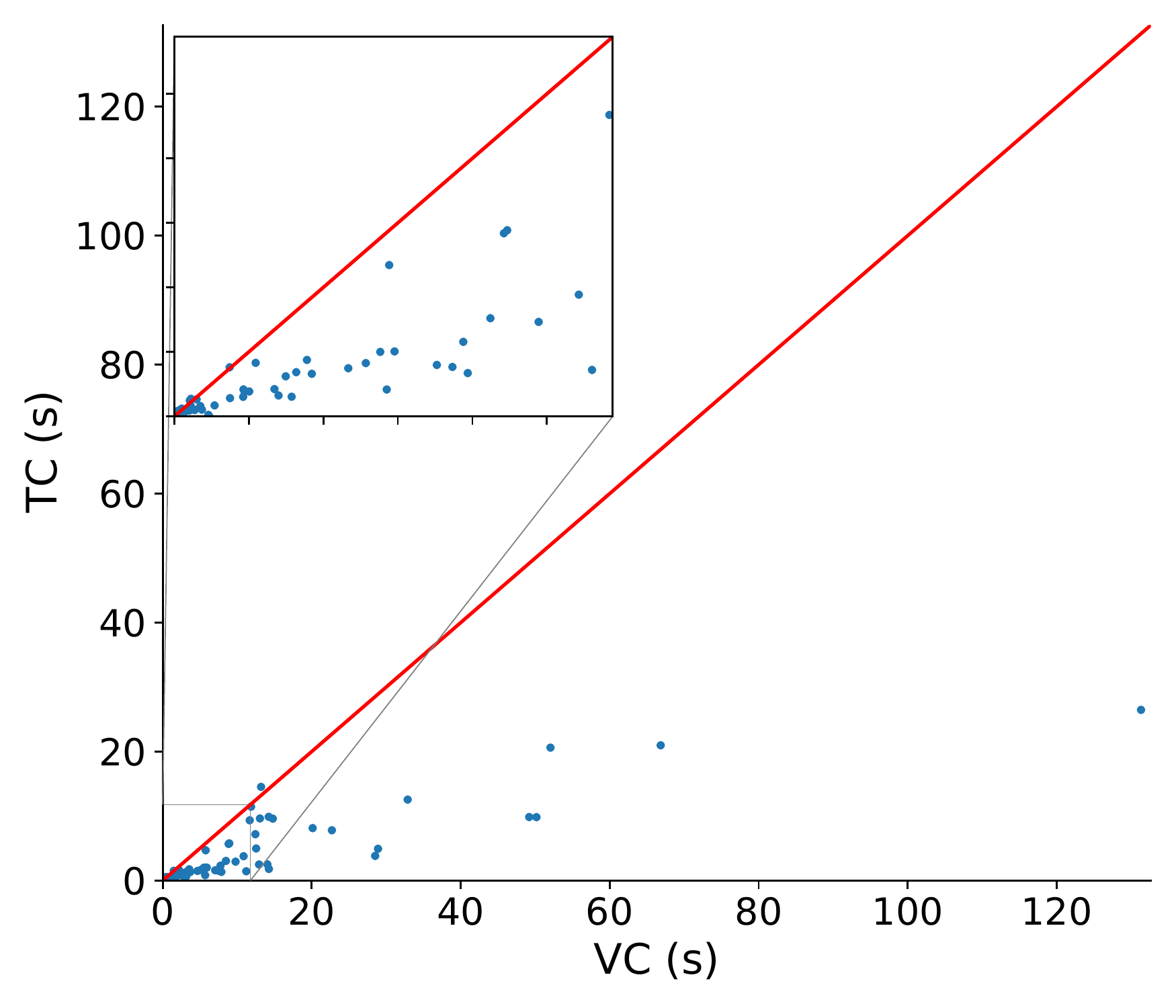}
\caption{$\HB$}
\label{subfig:hb_nr}
\end{subfigure}
\vspace{0.5cm}

\begin{subfigure}[b]{0.3\textwidth}
\includegraphics[scale=\scatterscale]{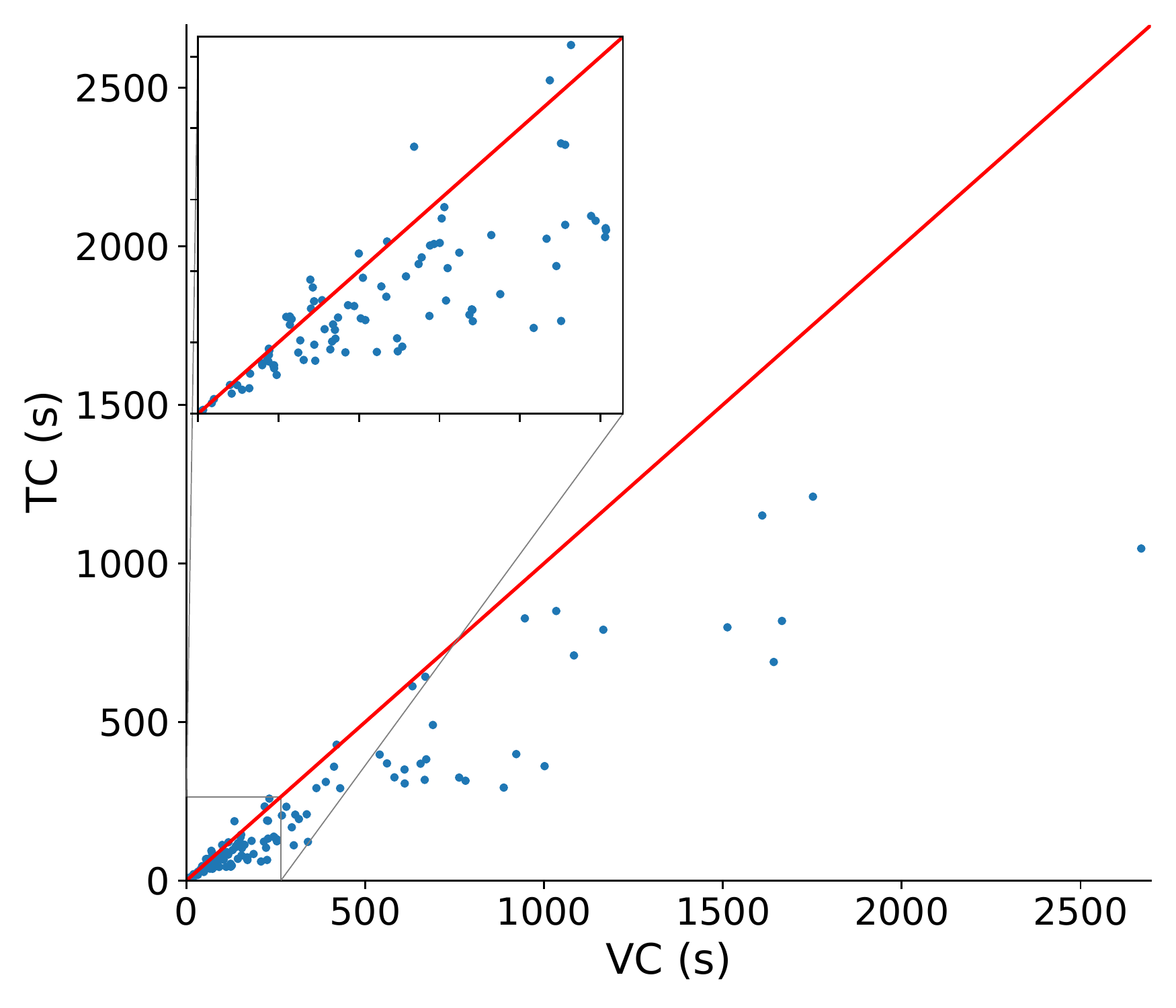}
\caption{$\Maz$+Analysis}
\label{subfig:maz}
\end{subfigure}
\begin{subfigure}[b]{0.3\textwidth}
\includegraphics[scale=\scatterscale]{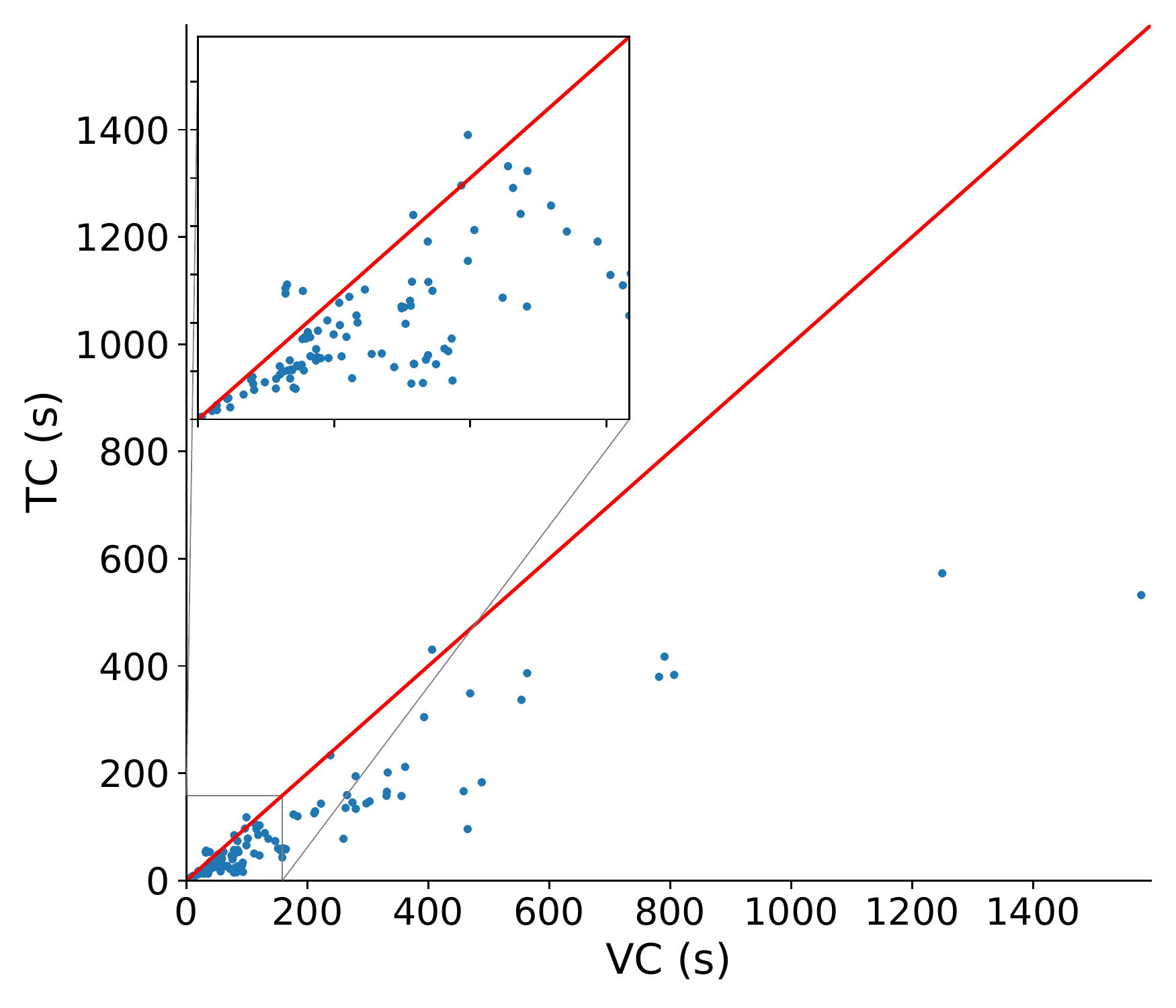}
\caption{$\SHB$+Analysis}
\label{subfig:shb}
\end{subfigure}
\begin{subfigure}[b]{0.3\textwidth}
\includegraphics[scale=\scatterscale]{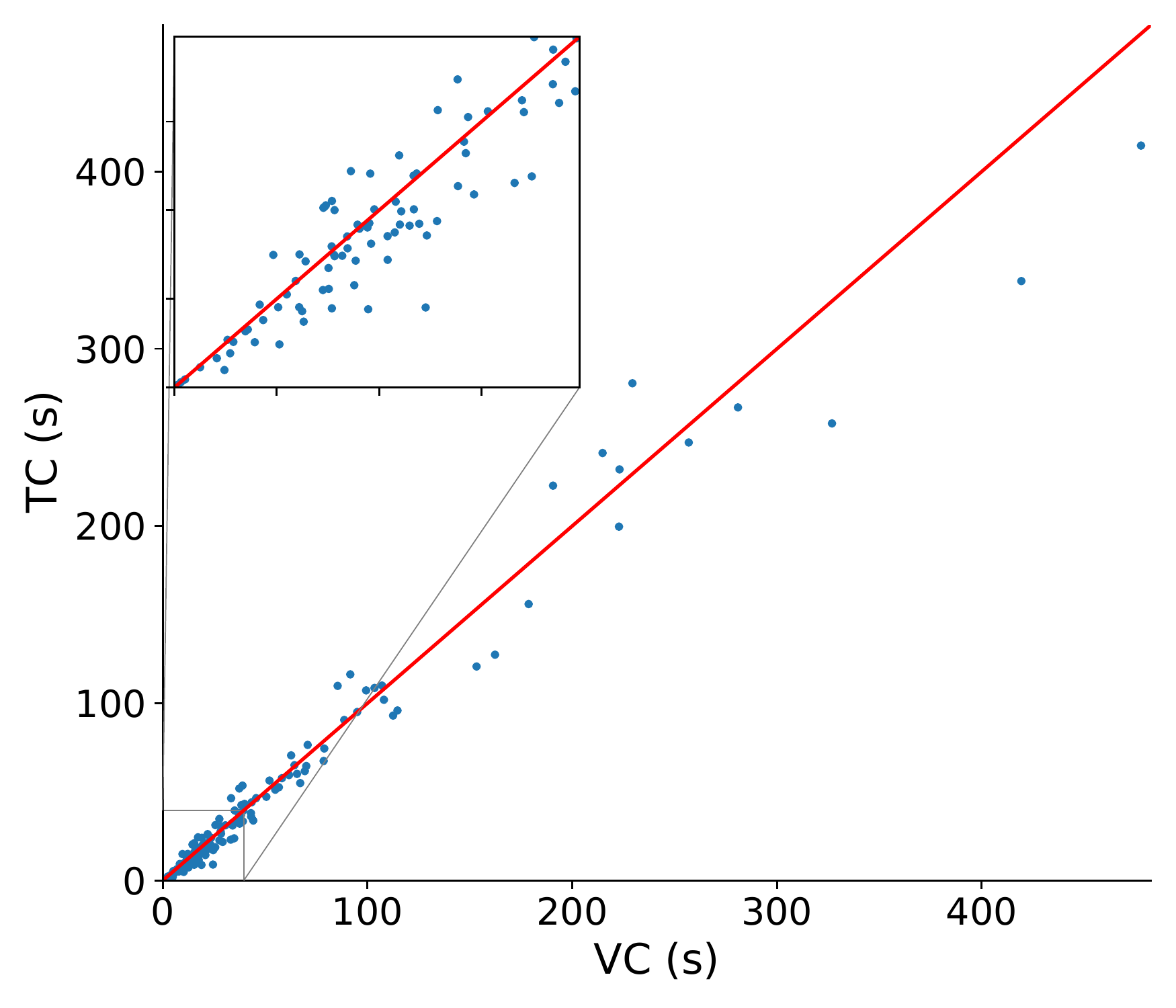}
\caption{$\HB$+Analysis}
\label{subfig:hb}
\end{subfigure}
\caption{
Times for processing each benchmark trace using tree clocks (TC) and vector clocks (VC).
The top row shows the time for computing the partial order, while the bottom row shows the time including the analysis component.
}
\label{fig:time_comparison}
\end{figure*}

%% file: tables/tab_speedups.tex
\begin{table}
\setlength\tabcolsep{3pt}
\renewcommand{\arraystretch}{1.2}
\small
\centering
\begin{tabular}{|r|c|c|c|}
\hline
& \textbf{$\Maz$} & \textbf{$\SHB$} & \textbf{$\HB$} \\
\hline
\hline
PO & 2.02 & 2.66 & 2.97 \\
PO + Analysis & 1.49 & 1.80 & 1.11 \\
\hline
\end{tabular}
\caption{
Average speedup for computing the partial order due to tree clocks.
}
\label{tab:speedups}
\end{table}

%% file: figures/fig_ratio_sync.tex
\begin{figure}
\centering
\includegraphics[scale=0.4]{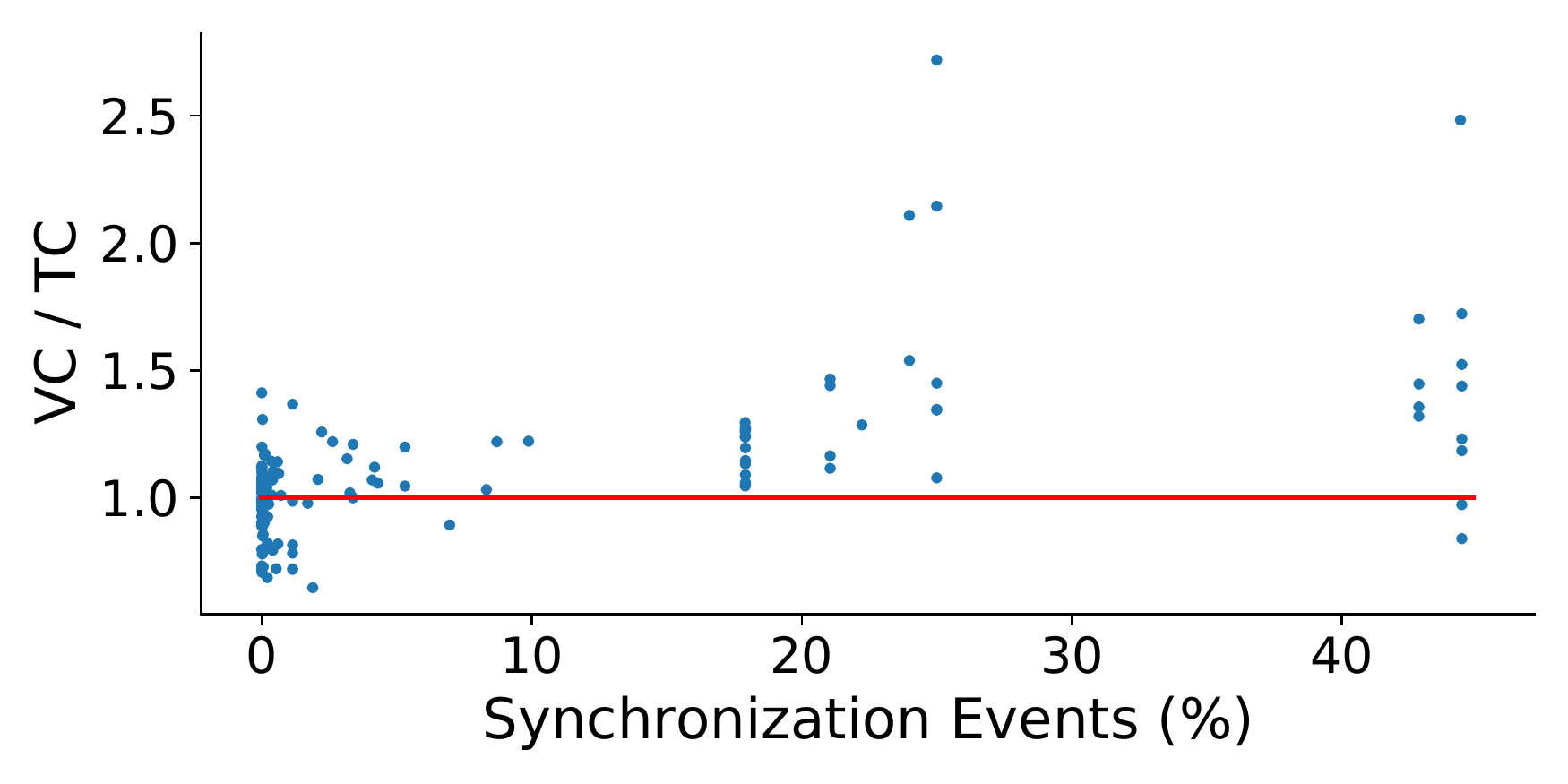}
\caption{
Speedup on $\HB$+analysis computation as a function of the percentage of synchronization events, for the traces where the total time is not too small ($\geq 100$ms).
}
\label{fig:ratio_sync}
\end{figure}

%% file: figures/fig_vtwork.tex
\begin{figure}
\centering
\includegraphics[scale=0.4]{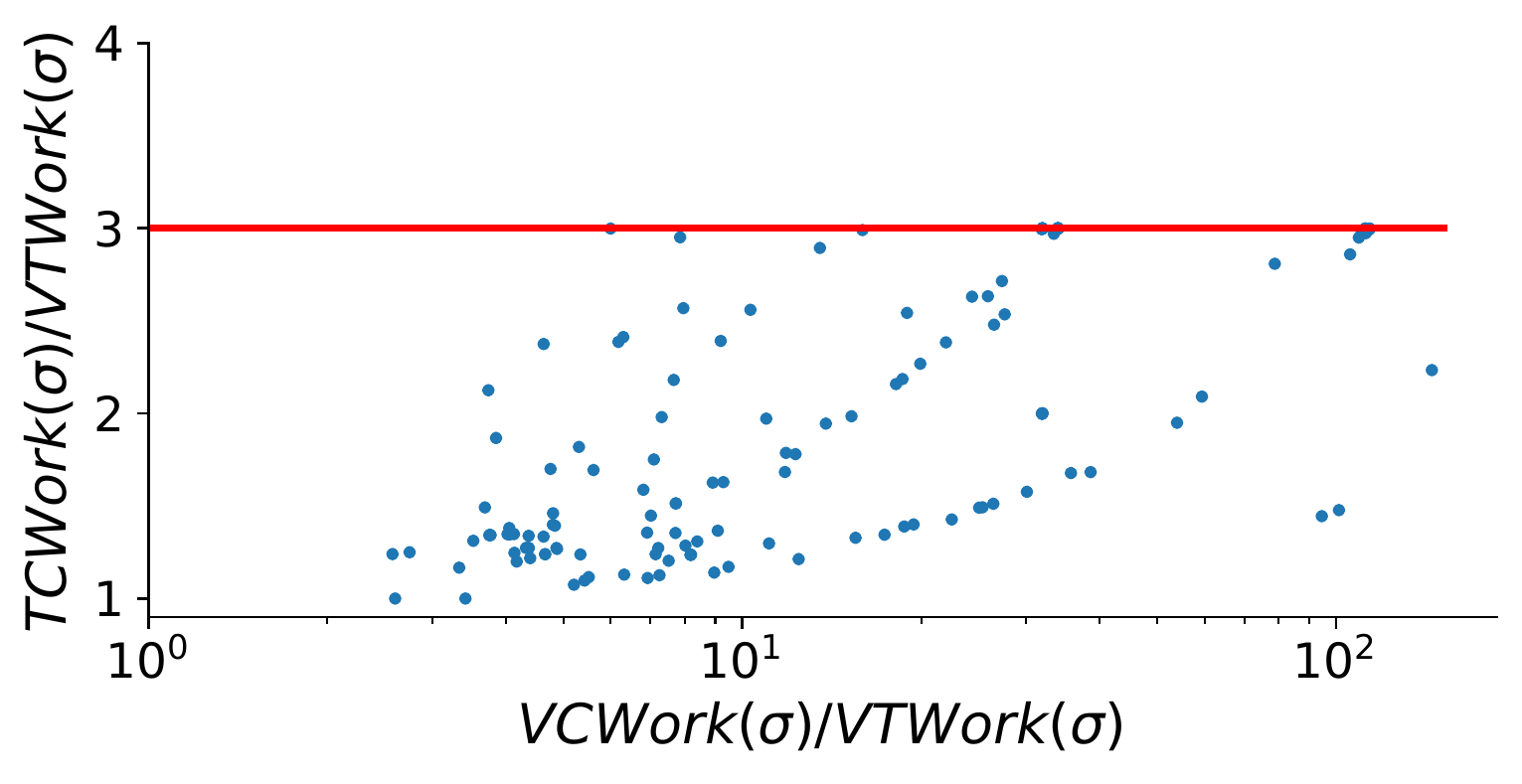}
\caption{
Comparison of the ratios $\TCWork(\Trace)/ \VTWork(\Trace)$ and $\VCWork(\Trace)/ \VTWork(\Trace)$ across all benchmarks.
}
\label{fig:vtwork}
\end{figure}

%% file: figures/fig_vtwork_histogram.tex
\begin{figure*}
\def\scatterscale{0.32}
\centering
\begin{subfigure}[b]{0.3\textwidth}
\includegraphics[scale=\scatterscale]{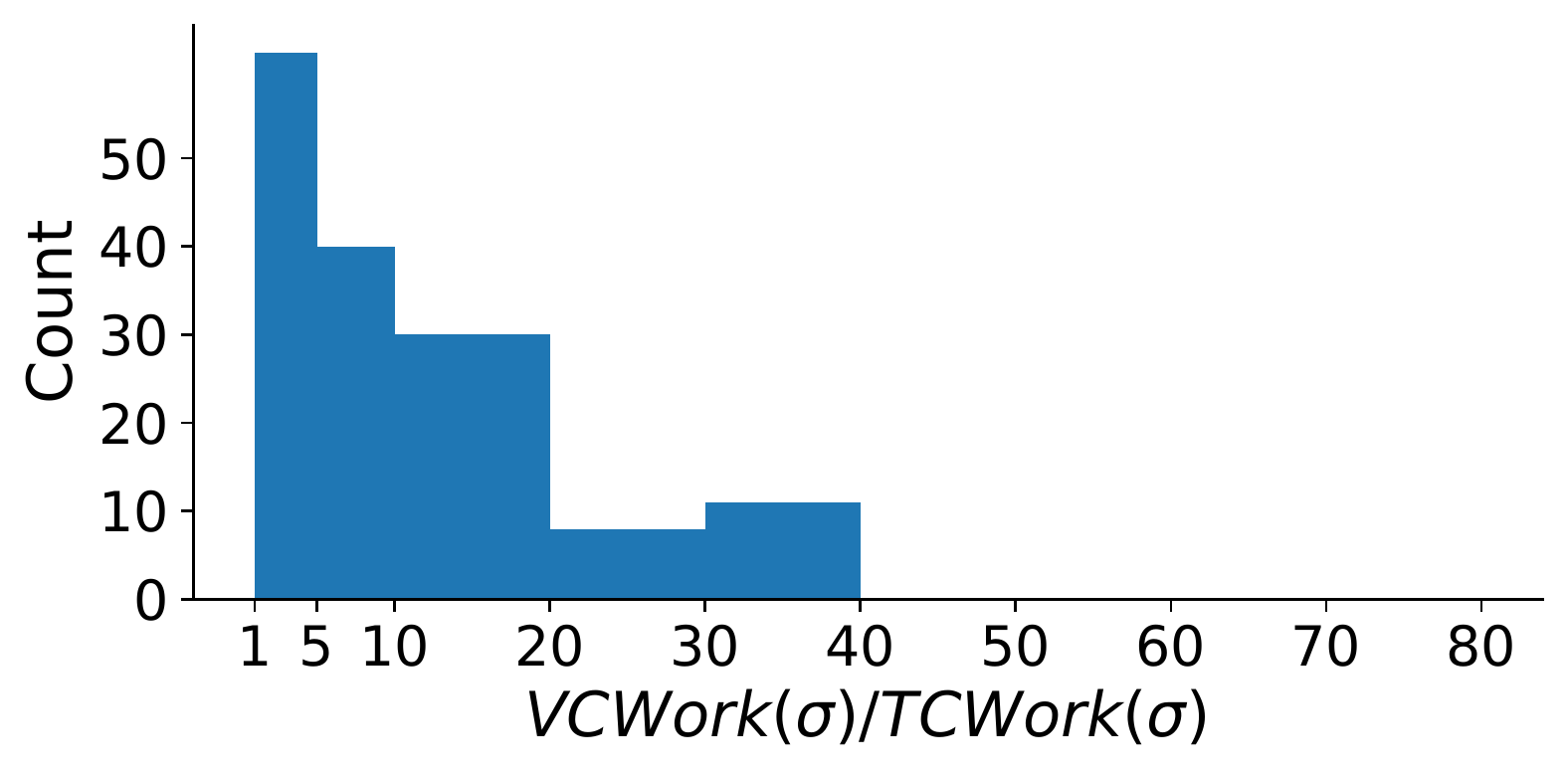}
\label{subfig:maz-vtwork}
\caption{$\Maz$}
\end{subfigure}
\begin{subfigure}[b]{0.3\textwidth}
\includegraphics[scale=\scatterscale]{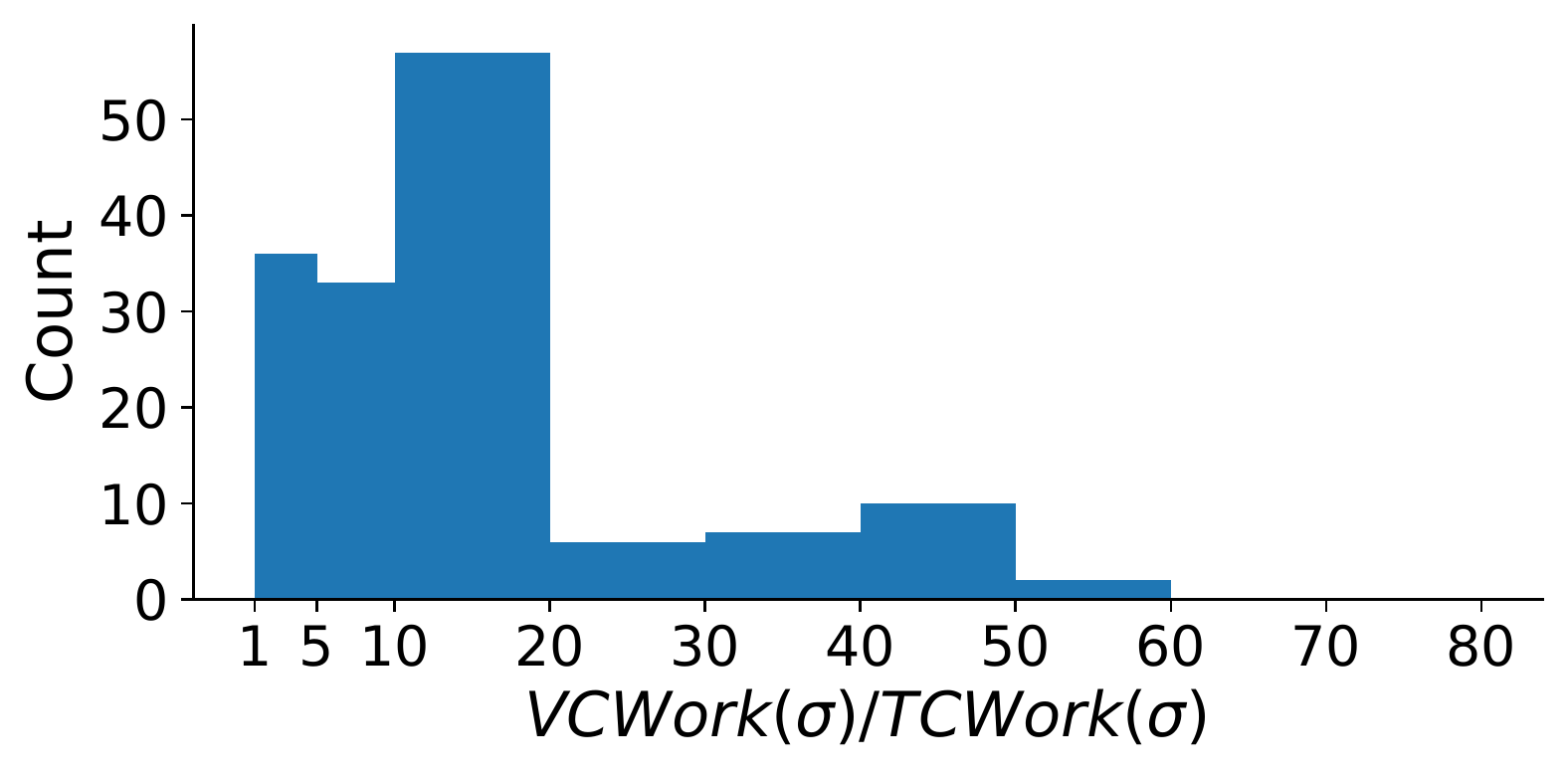}
\label{subfig:shb-vtwork}
\caption{$\SHB$}
\end{subfigure}
\begin{subfigure}[b]{0.3\textwidth}
\includegraphics[scale=\scatterscale]{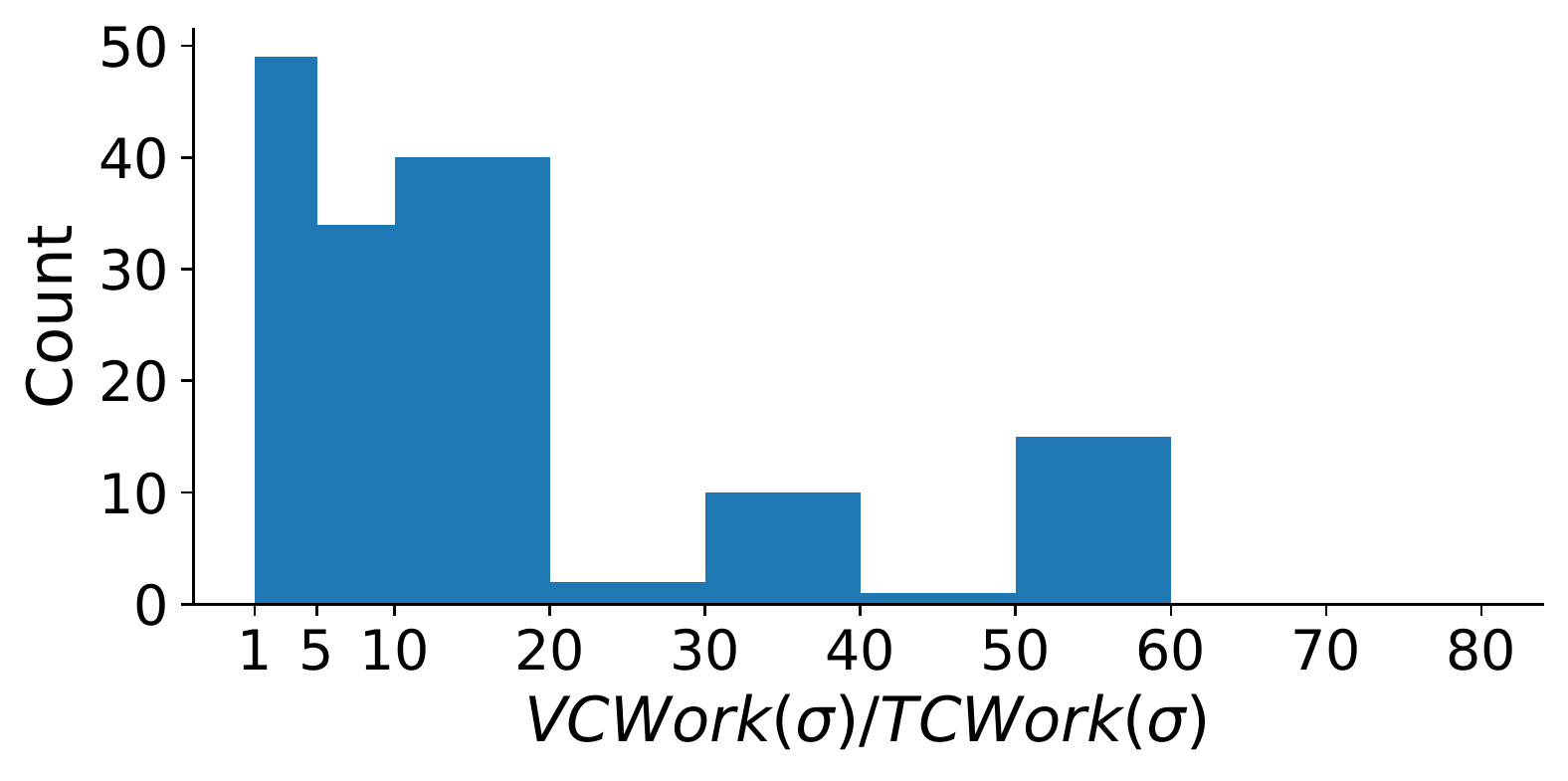}
\caption{$\HB$}
\label{subfig:hb-vtwork}
\end{subfigure}
\caption{
\camera{
Histogram of the ratios 
$\VCWork(\Trace)/ \TCWork(\Trace)$ across all benchmarks in our dataset.
}
}
\label{fig:vtwork-histogram}
\end{figure*}

%% file: figures/fig_scalability.tex
\begin{figure*}[!h]
\def\scatterscale{0.34}
\centering
\begin{subfigure}[b]{0.225\textwidth}
\includegraphics[scale=0.23]{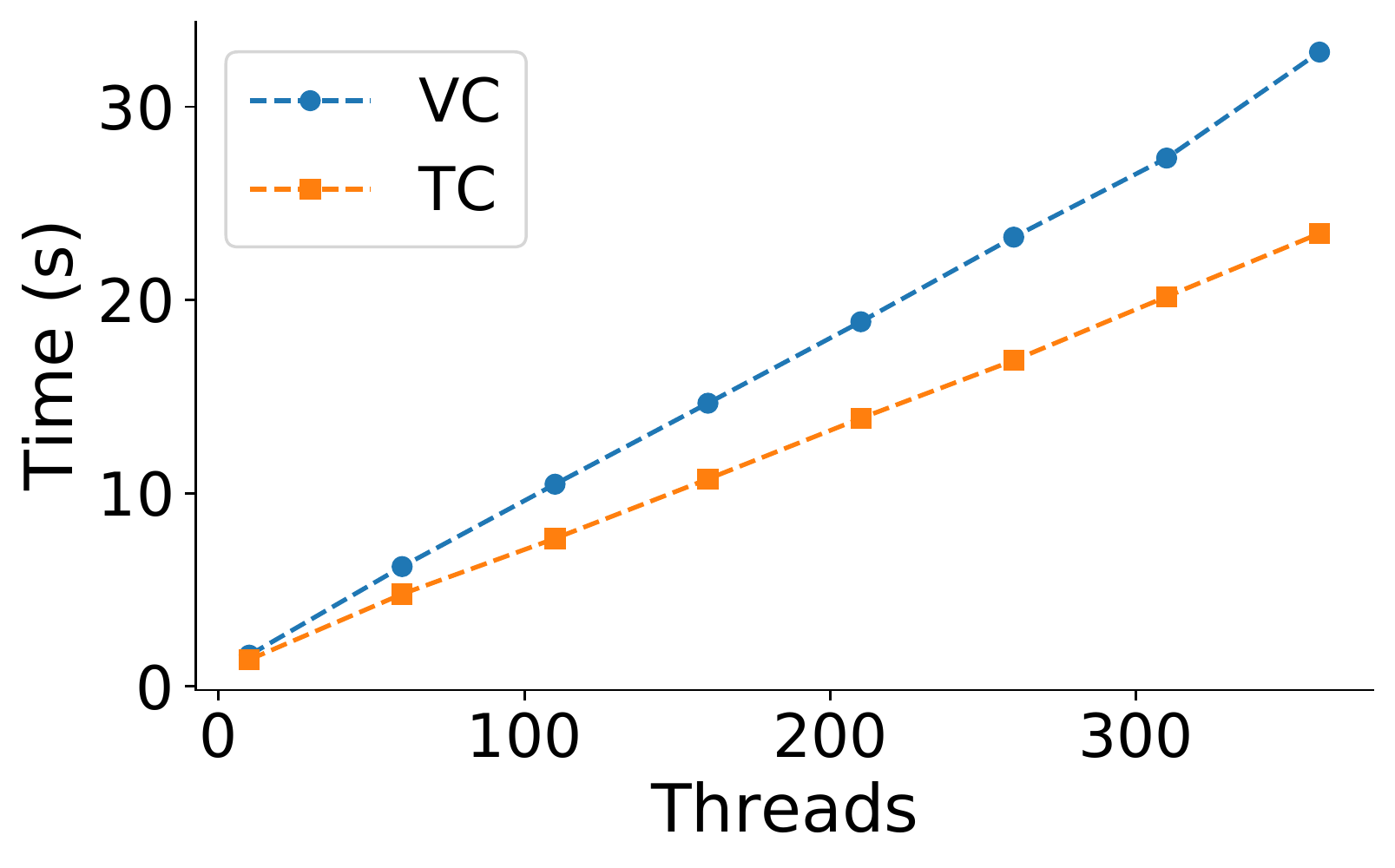}
\label{subfig:sca-one-lock-clique}
\caption{Single lock}
\end{subfigure}
\begin{subfigure}[b]{0.225\textwidth}
\includegraphics[scale=0.23]{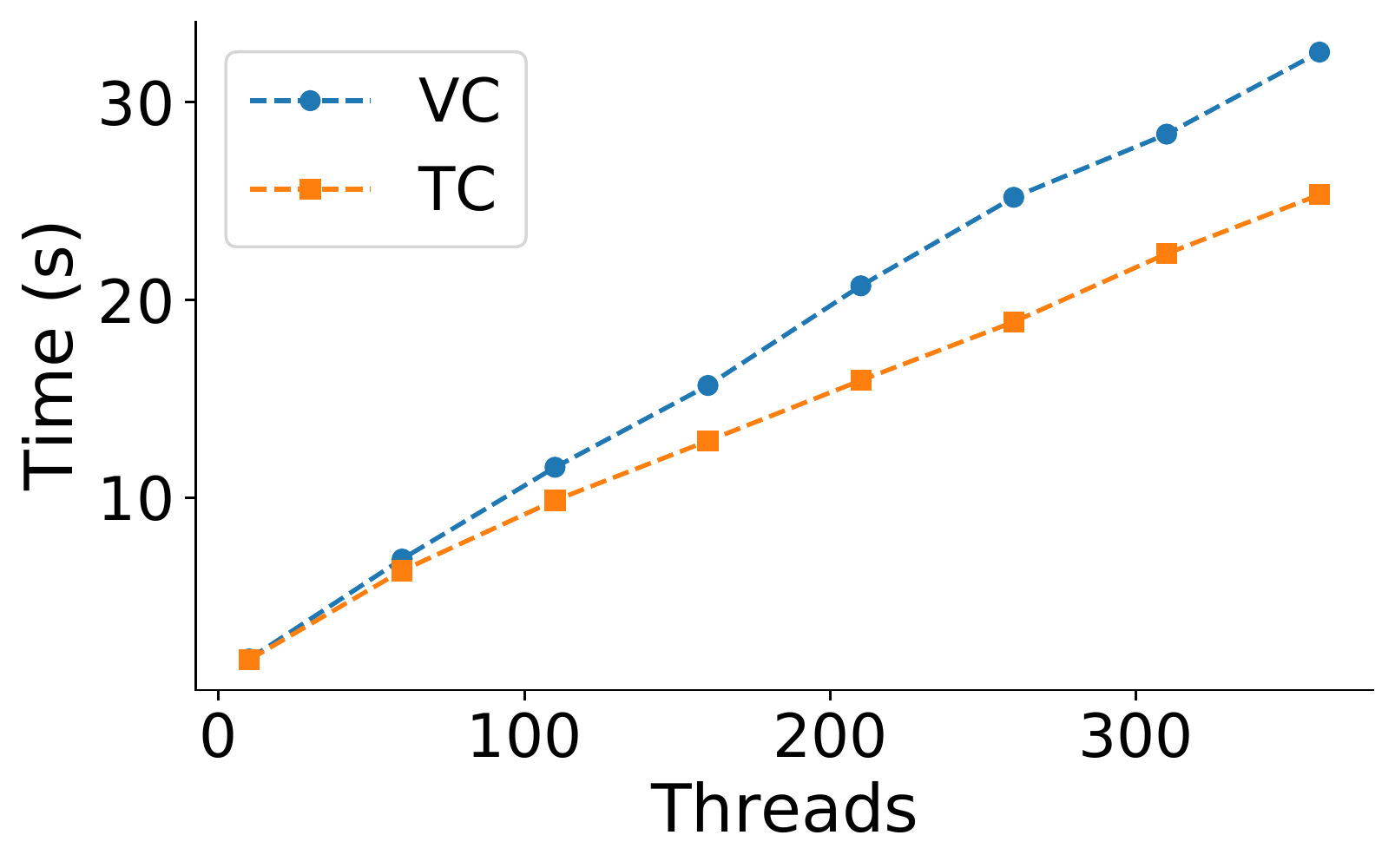}
\caption{Fifty locks}
\label{subfig:sca-skewed}
\end{subfigure}
\begin{subfigure}[b]{0.225\textwidth}
\includegraphics[scale=0.23]{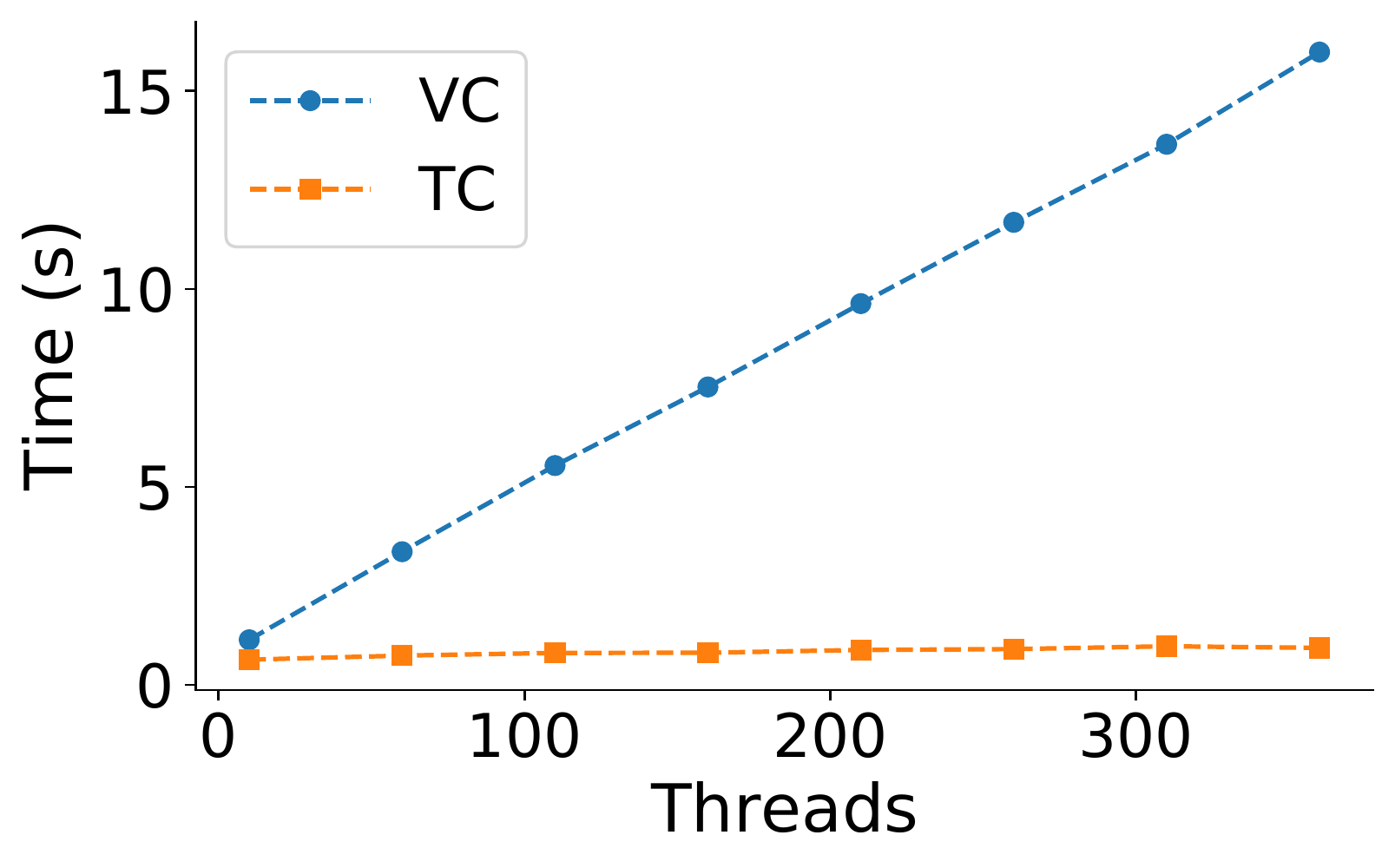}
\label{subfig:sca-star}
\caption{Star topology}
\end{subfigure}
\begin{subfigure}[b]{0.225\textwidth}
\includegraphics[scale=0.23]{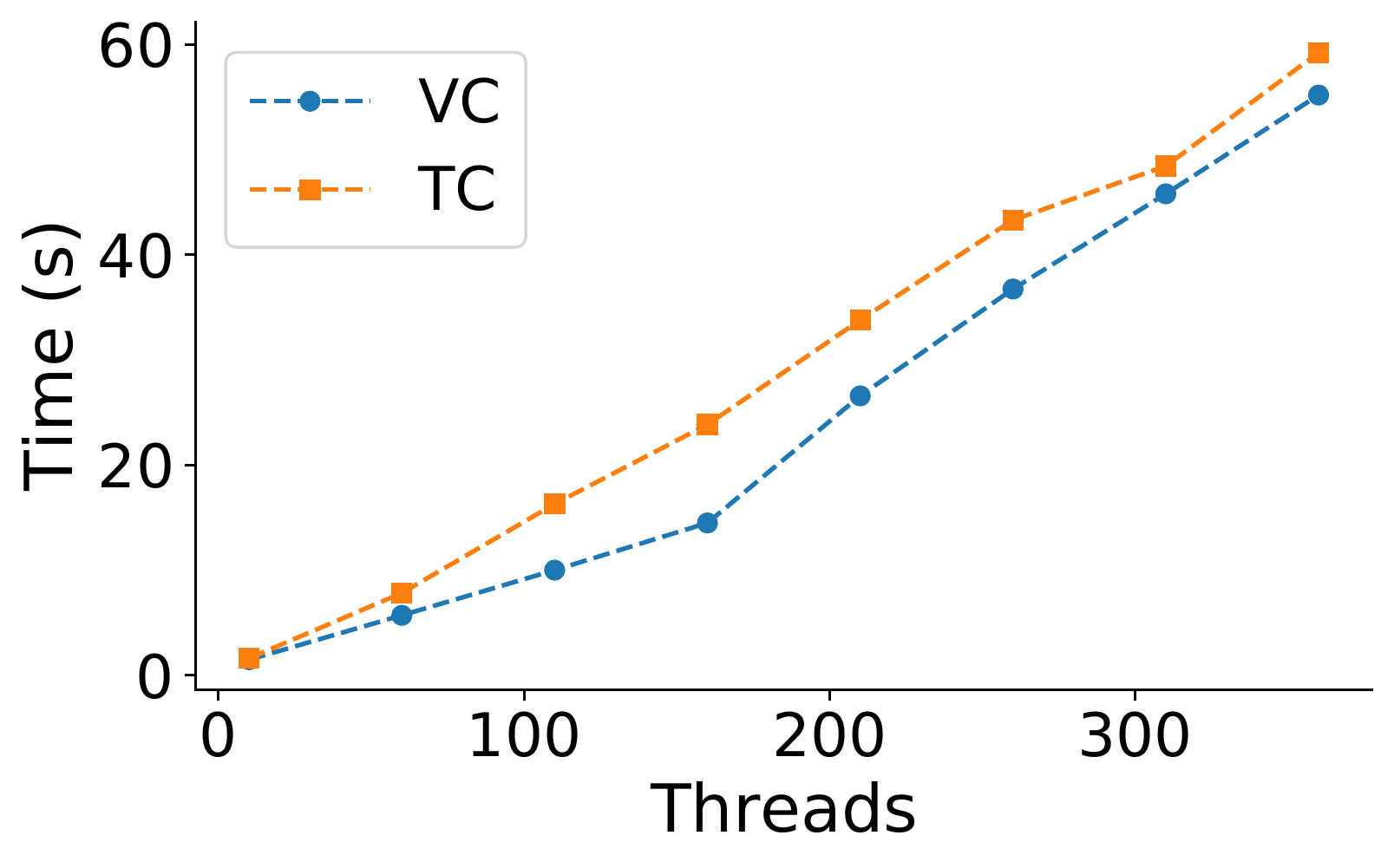}
\caption{Pairwise communication}
\label{subfig:sca-unique-lock}
\end{subfigure}
\caption{
\camera{
Comparison of tree clocks (TC) and vector clocks (VC) on $4$ different benchmarks with increasing number of threads.
}
}
\label{fig:scalability}
\end{figure*}

%% file: related_work.tex
%!TEX root= main.tex

\section{Related Work}\label{sec:related_work}

%Here we discuss related work and potential applications of tree clocks in other analyses.

%In this section we discuss potential applications to other partial orders,
%and related work in speeding up concurrent dynamic analyses, mostly in the context of race detection.

%Dynamic data race detection has been an active area of research
%for over four decades now, and many approaches have been proposed
%to speed-up vector clock based $\HB$-race detection in the past.
%Before we discuss these, we will first discuss the potential 
%and challenges of using vector clocks beyond $\HB$ (and $\SHB$).

\Paragraph{Other partial orders and tree clocks.}
As we have mentioned in the introduction,
besides $\HB$ and $\SHB$, many other partial orders perform dynamic analysis using vector clocks.
In such cases, tree clocks can replace vector clocks either partially or completely, sometimes requiring small extensions to the data structure as presented here.
In particular, we foresee interesting applications of tree clocks for the $\WCP$~\cite{Kini17}, $\DC$~\cite{Roemer18} and $\SDP$~\cite{Genc19} partial orders.

%To foster future research, we touch on these points here.
%
%One such example is the $\DC$ partial order~\cite{}.
%Although tree clocks can be directly applied here, a naive application might not be the most efficient.
%The reason is that the underlying algorithm maintains various queues which store copies of vector clocks.
%For a more efficient approach, one has to consider ways to alleviate the cost of deep copies of tree clocks.
%Another example is the $\WCP$ partial order~\cite{Kini17}.
%The challenge here is that $\WCP$ does not contain the thread order, and hence at a first glance, the monotonicity and transitivity properties of tree clocks fail.
%However, because $\WCP$ composes with $\HB$, these two properties only fail for the root of the tree clock, and apply as usual from the first level on. Thus, the join and monotone copy operations have to be adapted to account for this anomaly on the root.
%The $\SDP$ partial order has similar flavor to $\WCP$, but incurs fewer orderings~\cite{Genc19}.
%Here some tree clocks have to be further generalized to forest-like clocks that have multiple roots.
%Although such extensions are beyond the scope of this paper, they serve as fertile ground for future work.

\Paragraph{Speeding up dynamic analyses.}
Vector-clock based dynamic race detection 
is known to be slow~\cite{sadowski-tools-2014},
which many prior works have aimed to mitigate.
One of the most prominent performance bottlenecks is the
linear dependence of the size of vector timestamps on the number of threads.
Despite theoretical limits~\cite{CharronBost1991},
prior research exploits special structures in 
traces~\cite{cheng1998detecting,Feng1997,surendran2016dynamic,Dimitrov2015,Agrawal2018} 
that enable succinct vector time representations.
The Goldilocks~\cite{Elmas07} algorithm infers $\HB$-orderings
using locksets instead of vector timestamps
but incurs severe slowdown~\cite{Flanagan09}.
 % in a streaming setting.
%  but has nevertheless
% been used to speed up race detection for compressed traces~\cite{Kini2018compressed}.
The \fasttrack~\cite{Flanagan09} optimization uses
epochs for maintaining succinct access histories and our work
complements this optimization --- tree clocks offer optimizations
for other clocks (thread and lock clocks).
Other optimizations in clock representations are catered towards dynamic thread creation~\cite{Raychev2013,Wang2006,Raman2012}.     
Another major source of slowdown is program instrumentation
and expensive metadata synchronization.
Several approaches have attempted to minimize this slowdown,
including hardware assistance~\cite{RADISH2012,HARD2007},
hybrid race detection~\cite{OCallahan03,Yu05},
%based on the lockset principle~\cite{Savage97},
static analysis~\cite{bigfoot2017,redcard2013},
and sophisticated
ownership protocols~\cite{Bond2013,Wood2017,Roemer20}.

\section{Conclusion}
We have introduced tree clocks, a new data structure for maintaining logical times in concurrent executions. 
In contrast to standard vector clocks, tree clocks can dynamically capture communication patterns in their structure and perform join and copy operations in sublinear time, thereby avoiding the traditional overhead of these operations when possible.
Moreover, we have shown that tree clocks are vector-time optimal for computing the $\HB$ partial order, performing at most a constant factor work compared to what is absolutely necessary, in contrast to vector clocks.
Finally, our experiments show that tree clocks effectively reduce the running time for computing the $\Maz$, $\SHB$ and $\HB$ partial orders significantly, and thus offer a promising alternative over vector clocks.

\camera{
%We remark that our experimental evaluation in Section~\ref{sec:experiments} is based on single-threaded analyses on pre-recorded trace files.
Interesting future work includes incorporating tree clocks in an online analysis such as ThreadSanitizer~\cite{threadsanitizer}.
%A valuable extension of this evaluation is comparing the performance of tree clocks and vector clocks in a strictly online, program-instrumenta\-tion-based multi-threaded setting. 
%One of the main issues in this setting is the added complexity of maintaining analyses atomicity. 
Any use of additional synchronization to maintain analysis atomicity in this online setting is identical and of the same granularity to both vector clocks and tree clocks. 
However, the faster joins performed by tree clocks may lead to less congestion compared to vector clocks,
especially  for partial orders such as $\SHB$ and $\Maz$ where synchronization occurs on all events (i.e., synchronization, as well as access events).
We leave this evaluation for future work. 
Finally, since tree clocks are a drop-in replacement of vector clocks, 
most of the existing techniques that minimize the slowdown due to metadata synchronization (\cref{sec:related_work}) are directly applicable to tree clocks.
%the as already discussed in \cref{sec:related_work}, there exists a considerable amount of work which aims to minimize the slowdown incurred due to metadata synchronization in analyses based on vector clocks. 
%Since tree clocks are a drop-in replacement of vector clocks, this has the benefit that the existing optimizations for vector clocks also directly apply to tree clocks.
}

\begin{acks}
We thank anonymous reviewers for their constructive feedback on an earlier draft of this manuscript.
Umang Mathur was partially supported by the Simons Institute for the Theory of Computing. Mahesh Viswanathan is partially
supported by grants NSF SHF 1901069 and NSF CCF 2007428.
\end{acks}

%% file: arxiv/arxiv_appendix.tex
\input{proofs}

\input{hb_example}
\input{arxiv/threads_table}

%% file: proofs.tex
\section{Proofs}\label{appsec:proofs}

\smallskip
\lemcopymonotonicity*
\begin{proof}
Consider a trace $\Trace$, a release event $\Release(\ell)$ and let $\Acquire(\ell)$ be the matching acquire event.
When $\Acquire(\ell)$ is processed the algorithm performs $\CTC_{\Thread}\gets \CTC_{\Thread}\CJoin\CTC_{\ell}$, and thus $\CTC_{\ell}\CSmaller \CTC_{\Thread}$ after this operation.
By lock semantics, there exists no release event $\Release'(\ell)$ such that
$\Acquire(\ell)<^{\Trace}\Release'(\ell)<^{\Trace}\Release(\ell)$, and hence $\CTC_{\ell}$ is not modified until $\Release(\ell)$ is processed.
Since vector clock entries are never decremented, when $\Release(\ell)$ is processed we have $\CTC_{\ell} \CSmaller \CTC_{\Thread}$, as desired.
\end{proof}

\lemtcmonotonicity*
\begin{proof}
First, note that after initialization $u$ has no children, hence each statement is trivially true.
Now assume that both statements hold when the algorithm processes an event $\Event$, and we show that they both hold after the algorithm has processed $\Event$.
We distinguish cases based on the type of $\Event$.

\noindent $\Event=\langle \Thread, \Acquire(\ell) \rangle$.
The algorithm performs the operation $\CTC_{\Thread}.\FunctionCTJoin(\CTL_{\ell})$, hence the only tree clock modified is $\CTC_{\Thread}$, and thus it suffices to examine the cases that $\CTC_{\Thread}$ is $\CTC$ and $\CTC_{\Thread}$ is $\CTC'$.
\begin{compactenum}
\item $\CTC_{\Thread}$ is $\CTC$.
First consider the case that $u=\CTC_{\Thread}.\GTree.\Root$.
Observe that $u.\clock > \CTC'.\FunctionCTGet(u.\tid)$, and thus \cref{item:monotonicity1} holds trivially.
For \cref{item:monotonicity2}, we distinguish cases based on whether $v.\clock$ has progressed by the $\FunctionCTJoin$ operation.
If yes, then we have $v.\pclock = u.\clock$, and the statement holds trivially for the same reason as in \cref{item:monotonicity1}.
Otherwise, we have that for every descendant $w$ of $v$, the clock $w.\clock$ has not progressed by the $\FunctionCTJoin$ operation, hence the statement holds by the induction hypothesis on $\CTC_{\Thread}$.

Now consider the case that $u\neq \CTC_{\Thread}.\GTree.\Root$.
If $u.\clock$ has not progressed by the $\FunctionCTJoin$ operation then each statement holds by the induction hypothesis on $\CTC_{\Thread}$.
Otherwise, using the induction hypothesis one can show that for every descendant $w$ of $u$,
there exists a node $w_{\ell}$ of $\CTL_{\ell}$ that is a descendant of a node $u_{\ell}$ such that $w_{\ell}.\tid = w.\tid$ and $u_{\ell}.\tid=u.\tid$.
Then, each statement holds by the induction hypothesis on $\CTL_{\ell}$.

\item $\CTC_{\Thread}$ is $\CTC'$.
For \cref{item:monotonicity1}, if $u.\clock \leq \CTC'.\FunctionCTGet(u.\tid)$ holds before the $\FunctionCTJoin$ operation, then the statement holds by the induction hypothesis, since $\FunctionCTJoin$ does not decrease the clocks of $\CTC_{\Thread}$.
Otherwise, the statement follows by the induction hypothesis on $\CTL_{\ell}$.
The analysis  for \cref{item:monotonicity2} is similar.

The desired result follows.
\end{compactenum}

\noindent $\Event=\langle \Thread, \Release(\ell) \rangle$.
The algorithm performs the operation $\CTL_{\ell}.\linebreak\FunctionCTMonotoneCopy(\CTC_{\Thread})$.
The analysis is similar to the previous case, this time also using \cref{lem:copy_monotonicity} to argue that no time stored in $\CTL_{\ell}$ decreases.
\end{proof}

\smallskip
\lemhbcor*
\begin{proof}
The lemma follows directly from \cref{lem:tc_monotonicity}.
In each case, if the corresponding operation (i.e., $\FunctionCTJoin$ for event $\langle \Thread, \Acquire(\ell) \rangle$ and $\FunctionCTMonotoneCopy$ for $\langle \Thread, \Release(\ell) \rangle$),
if the clock of a node $w$ of the tree clock that performs the operation does not progress, then
we are guaranteed that $w.\clock$ is not smaller than the time of the thread $w.\tid$ in the tree clock that is passed as an argument to the operation.
\end{proof}

\Paragraph{First remote acquires.}
Consider a trace $\Trace$ and a lock-release event $e = \ev{\Thread, \Release(\lk)}$ of $\Trace$,
such that there exists a later acquire event $\Event' = \ev{\Thread', \Acquire(\lk)}$ ($\Event \stricttrord{\Trace} \Event'$).
The \emph{first remote acquire} of $\Event$ is the first event $\Event'$ with the above property.
For example, in \cref{subfig:hb_example_trace}, $\Event_7$ is the first remote acquire of $\Event_2$.
While constructing the $\HB$ partial order, the algorithm makes $\HB$ orderings from lock-release events to their first remote acquires $\Release(\lk) \stricthb{\Trace} \Acquire(\lk)$.
The following lemma captures the property that the edges of tree clocks are essentially the inverses of such orderings.

\smallskip
\begin{restatable}{lemma}{lemuniqueness}\label{lem:uniqueness}
Consider the execution of \cref{algo:hb} on a trace $\Trace$.
For every tree clock $\CTC_i$ and node $u$ of $\CTC_i.\GTree$ other than the root,
the following assertions hold.
\vspace{-0.1cm}
\begin{compactenum}
\item $u$ points to a lock-release event $\Release(\ell)$.
\item $\Release(\ell)$ has a first remote acquire $\Acquire(\ell)$ and $(v.\tid, u.\pclock)$ points to $\Acquire(\ell)$, where $v$ is the parent of $u$ in $\CTC_i.\GTree$. 
%{\color{red} Should there be v.clk as well?}
\end{compactenum}
\end{restatable}
\begin{proof}
The lemma follows by a straightforward induction on $\Trace$.
\end{proof}

\cref{lem:uniqueness} allows us to prove the vt-optimality of tree clocks.

\smallskip
\thmvtoptimality*
\begin{proof}
Consider a critical section of a thread $\Thread$ on lock $\ell$, marked by two events $\Acquire(\ell)$, $\Release(\ell)$.
We define the following vector times.
\begin{compactenum}
\item $\VectorTime_{\Thread}^{1}$ and $\VectorTime_{\Thread}^{2}$ are the vector times of $\CTC_{\Thread}$ right before and right after $\Acquire(\ell)$ is processed, respectively.
\item $\VectorTime_{\ell}^{1}$ is the vector time of $\CTC_{\ell}$ right before $\Acquire(\ell)$ is processed.
\item $\VectorTime_{\Thread}^{3}$ is the vector time of $\CTC_{\Thread}$ right before $\Release(\ell)$ is processed.
\item $\VectorTime_{\ell}^{3}$ and $\VectorTime_{\ell}^{4}$ are the vector times of $\CTC_{\ell}$ right before and right after $\Release(\ell)$ is processed, respectively,
\end{compactenum}
First, note that 
(i)~$\VectorTime_{\Thread}^{1}\CSmaller\VectorTime_{\Thread}^{3}$, and
(ii)~due to lock semantics, we have $\VectorTime_{\ell}^3=\VectorTime_{\ell}^{1}$.
Let $W=W_{J}+W_{C}$, where
\begin{align*}
W_J=&|\{\Thread'\colon \VectorTime_{\Thread}^{2}(\Thread')\neq \VectorTime_{\Thread}^{1}(\Thread')\}| \quad \text{ and}\\
W_C=&|\{\Thread'\colon \VectorTime_{\ell}^{4}(\Thread')\neq \VectorTime_{\ell}^{3}(\Thread')\}|
\end{align*}
i.e., $W_J$ and $W_C$ are the vt-work for handling $\Acquire(\ell)$ and $\Release(\ell)$, respectively.
Let $\Time_{J}$ be the time spent in $\TreeClock_{\Thread}.\FunctionCTJoin$ due to $\Acquire(\ell)$.
Similarly, let $\Time_{C}$ be the time spent in $\TreeClock_{\ell}.\FunctionCTMonotoneCopy$ due to $\Release(\ell)$.
We will argue that $\Time_{J}=O(W)$ and $\Time_{C}=O(W_C)$, and thus $\Time_{J}+\Time_{C}=O(W)$.
Note that this proves the lemma, simply by summing over all critical sections of $\Trace$.

We start with $\Time_J$.
Observe that the time spent in this operation is proportional to the number of times the loop in \cref{line:routineupdatednodesforjoin_loop} is executed, i.e., the number of nodes $v'$ that the loop iterates over.
Consider the if statement in \cref{line:routineupdatednodesforjoin_if}.
If $\FunctionCTGet(v'.\tid)<v'.\clock$, then we have $\VectorTime_{\Thread}^{2}(v'.\tid)>\VectorTime_{\Thread}^{1}(v'.\tid)$, and thus this iteration is accounted for in $W_J$.
On the other hand, if $\FunctionCTGet(v'.\tid)>v'.\clock$, then we have $\VectorTime_{\Thread}^{1}(v'.\tid)>\VectorTime_{\ell}^{1})(v'.\tid)$.
Due to (i) and (ii) above, we have $\VectorTime_{\ell}^{4}(v'.\tid)>\VectorTime_{\ell}^{3}(v'.\tid)$, and thus this iteration is accounted for in $W_C$.
Finally, consider the case that $\FunctionCTGet(v'.\tid)=v'.\clock$, and let $v$ be the node of $\CTC_{\Thread}$ such that $v.\tid=v'.\tid$.
There can be at most one such $v$ that is not the root of $\TreeClock_{\Thread}$.
For every other such $v$, let $u=\Parent(v)$.
Note that $v'$ is not the root of $\CTC_{\ell}$, and let $u'=\Parent(v')$.
Let $\Release(\ell)$ be the lock-release event that $v$ and $v'$ point to.
By \cref{lem:uniqueness}, $\Release(\ell)$ has a first remote acquire $\Acquire(\ell)$ such that
(i)~$u.\tid=u'.\tid=\Thread'$, where $\Thread'$ is the thread of $\Acquire(\ell)$, and
(ii)~$v.\pclock$ is the local clock of $\Acquire(\ell)$.
Since $\RoutineUpdatedNodesForCopy$ examines $v'$, we must have $u'.\clock > u.\clock$.
In turn, we have $u.\clock \geq v.\pclock$, and thus $u'.\pclock>v.\pclock$.
Hence, due to \cref{line:routineupdatednodesforjoinifbreak}, $u'$ can have at most one child $v'$ with $v'.\clock=\FunctionCTGet(v'.\tid)$.
Thus, we can account for the time of this case in $W_J$.
Hence, $\Time_J=O(W)$, as desired.

%yo
We now turn our attention to $\Time_C$.
Similarly to the previous case, the time spent in this operation is proportional to the number of times the loop in \cref{line:routineupdatednodesforcopy_loop} is executed.
Consider the if statement in \cref{line:routineupdatednodesforcopy_if}.
If $\FunctionCTGet(v'.\tid)<v'.\clock$, then we have $\VectorTime_{\ell}^{4}(v'.\tid)>\VectorTime_{\ell}^{3}(v'.\tid)$, and thus this iteration is accounted for in $W_C$.
Note that as the copy is monotone (\cref{lem:copy_monotonicity}),
we can't have $\FunctionCTGet(v'.\tid)>v'.\clock$.
Finally, the reasoning for the case where $\FunctionCTGet(v'.\tid)=v'.\clock$ is similar to the analysis of $\Time_J$,
using \cref{line:routineupdatednodesforcopyifbreak} instead of \cref{line:routineupdatednodesforjoinifbreak}.
Hence, $\Time_C=O(W_C)$, as desired.

The desired result follows.
\end{proof}

%% file: hb_example.tex
\section{Example of Tree Clocks in $\HB$}\label{sec:app_hb_example}

\input{figures/fig_hb_example}

\cref{fig:hb_example} shows an example run of \cref{algo:hb} on a trace $\Trace$, showing how tree clocks grow during the execution.
The figure shows the tree clocks $\CTC_{\Thread}$ of the threads; the tree clocks of locks $\CTC_{\ell}$ are just copies of the former after processing a lock-release event (shown in parentheses in the figure).
\cref{fig:hb_example2} presents a closer look of the $\FunctionCTJoin$ and $\FunctionCTMonotoneCopy$ operations for the last events of $\Trace$.

\input{figures/fig_hb_example2}

%% file: figures/fig_hb_example.tex
%!TEX root = ../main.tex

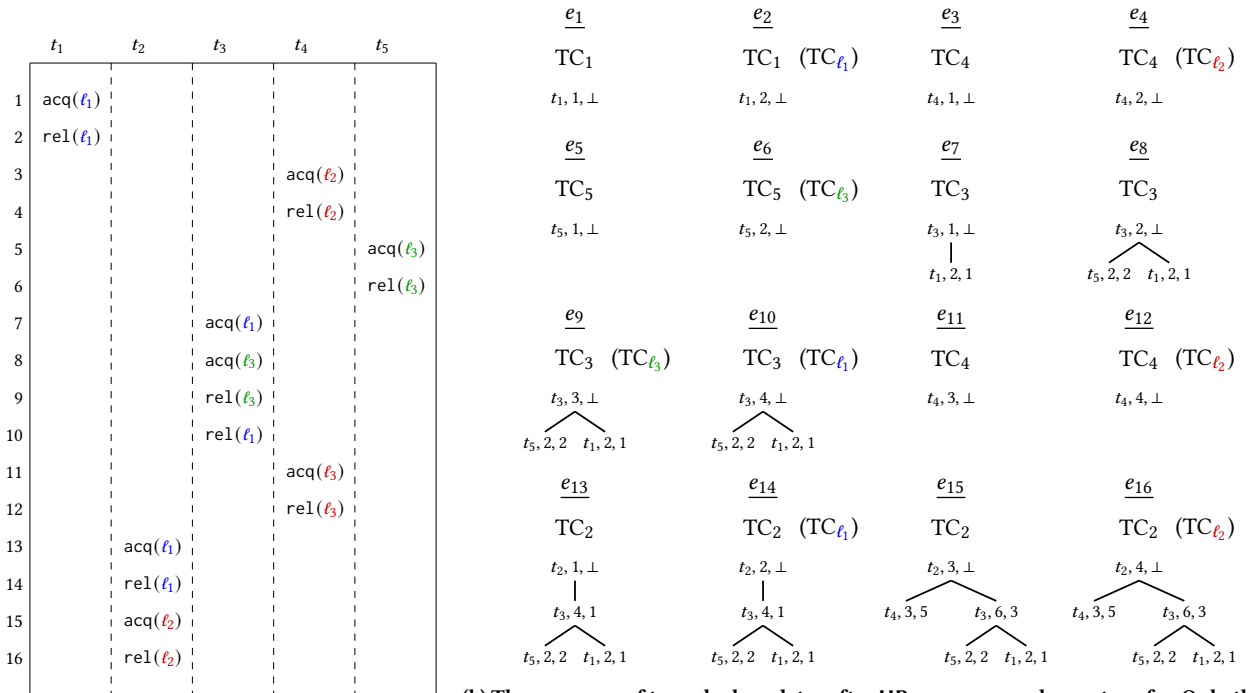
\begin{figure*}
\small
\newcommand{\xdisposition}{2.5}
\newcommand{\ydisposition}{2.5}
\newcommand{\halfyOuter}{0.3}
\newcommand{\halfxOuter}{0.5}
\newcommand{\halfyInner}{0.2}
\newcommand{\halfxInner}{0.4}
\newcommand{\grayDark}{gray!50}
\newcommand{\grayLight}{gray!20}
\newcommand{\evYDist}{15}
\newcommand{\scaletree}{0.85}
\begin{subfigure}[b]{0.33\textwidth}
\centering

\begin{comment}
\def\rownumber{}
\begin{tabular}[b]{@{\makebox[1.2em][r]{\rownumber\space}} | c | c | c | c | c |}
% \normalsize{$\mathbf{\LocalTrace_1}$} & \normalsize{$\mathbf{\LocalTrace_2}$} & \normalsize{$\mathbf{\LocalTrace_3}$} & \normalsize{$\mathbf{\LocalTrace_4}$} & \normalsize{$\mathbf{\LocalTrace_5}$}
\normalsize{$\Thread_1$} & \normalsize{$\Thread_2$} & \normalsize{$\Thread_3$} & \normalsize{$\Thread_4$} & \normalsize{$\Thread_5$}
\gdef\rownumber{\stepcounter{magicrownumbers}\arabic{magicrownumbers}} \\
\hline
$\Acquire(\textcolor{blue}{\ell_1})$ & & & &  \\
$\Release(\textcolor{blue}{\ell_1})$ & & & &  \\
&&&$\Acquire(\textcolor{\darkred}{\ell_2})$ & \\
&&&$\Release(\textcolor{\darkred}{\ell_2})$ & \\
&&&&$\Acquire(\textcolor{\darkgreen}{\ell_3})$ \\
&&&&$\Release(\textcolor{\darkgreen}{\ell_3})$ \\
&&$\Acquire(\textcolor{blue}{\ell_1})$ & &  \\
&&$\Acquire(\textcolor{\darkgreen}{\ell_3})$ & &  \\
&&$\Release(\textcolor{\darkgreen}{\ell_3})$ & &  \\
&&$\Release(\textcolor{blue}{\ell_1})$ & &  \\
&& $\Acquire(\textcolor{\darkred}{\ell_2})$ & & \\
&&$\Release(\textcolor{\darkred}{\ell_2})$ & & \\
&$\Acquire(\textcolor{blue}{\ell_1})$ & & & \\
&$\Release(\textcolor{blue}{\ell_1})$ & & &  \\
&$\Acquire(\textcolor{\darkred}{\ell_2})$ & & & \\
&$\Release(\textcolor{\darkred}{\ell_2})$ & & &  \\
\hline
\end{tabular}
\end{comment}
\executioncustom{5}{
\figevcustom{1}{\Acquire(\textcolor{blue}{\ell_1})}
\figevcustom{1}{\Release(\textcolor{blue}{\ell_1})}
\figevcustom{4}{\Acquire(\textcolor{\darkred}{\ell_2})}
\figevcustom{4}{\Release(\textcolor{\darkred}{\ell_2})}
\figevcustom{5}{\Acquire(\textcolor{\darkgreen}{\ell_3})}
\figevcustom{5}{\Release(\textcolor{\darkgreen}{\ell_3})}
\figevcustom{3}{\Acquire(\textcolor{blue}{\ell_1})}
\figevcustom{3}{\Acquire(\textcolor{\darkgreen}{\ell_3})}
\figevcustom{3}{\Release(\textcolor{\darkgreen}{\ell_3})}
\figevcustom{3}{\Release(\textcolor{blue}{\ell_1})}
\figevcustom{4}{\Acquire(\textcolor{\darkred}{\ell_3})}
\figevcustom{4}{\Release(\textcolor{\darkred}{\ell_3})}
\figevcustom{2}{\Acquire(\textcolor{blue}{\ell_1})}
\figevcustom{2}{\Release(\textcolor{blue}{\ell_1})}
\figevcustom{2}{\Acquire(\textcolor{\darkred}{\ell_2})}
\figevcustom{2}{\Release(\textcolor{\darkred}{\ell_2})}
}
\caption{
An example trace $\Trace$.
% Only synchronization events are important for this example.
}
\label{subfig:hb_example_trace}
\end{subfigure}
\quad
\begin{subfigure}[b]{0.63\textwidth}

\centering
\begin{tikzpicture}[thick,sibling distance=1em, node distance=15,
every tree node/.style={thick, rectangle, rounded corners, minimum height=0.8mm, minimum width=5mm, inner sep=2pt},
boundingBox/.style={very thick, draw=none},
sibling distance=3pt,
level distance=20pt,
]
\tikzset{edge from parent/.append style={ thick}}
\pgfdeclarelayer{background}
\pgfdeclarelayer{foreground}
\pgfsetlayers{background,main,foreground}

\begin{scope}[shift={(0*\xdisposition,0*\ydisposition)}]

\begin{scope}[scale=\scaletree]
\Tree [
. \node[] (11) {$\treethr{1},1,\bot$};
]
\end{scope}
\node[above of=11] (CT11) {\normalsize $\TreeClock_1$};
\node[above of=CT11, node distance=\evYDist] {\normalsize \underline{$e_1$}};

\node[boundingBox,  fit=(11) (CT11)] {};

\end{scope}

\begin{scope}[shift={(1*\xdisposition,0*\ydisposition)}]
\begin{scope}[scale=\scaletree]
\Tree [
. \node[] (12) {$\treethr{1},2,\bot$};
]
\end{scope}
\node[above of=12] (CT12) {\normalsize $\TreeClock_1$};
\node[right of = CT12, node distance=25] (CTL11) {\normalsize ($\TreeClock_{\textcolor{blue}{\ell_1}}$)};
\node[above of=CT12, node distance=\evYDist] {\normalsize \underline{$e_2$}};

\node[boundingBox,  fit=(12) (CT12) (CTL11)] {};

\end{scope}

\begin{scope}[shift={(2*\xdisposition,0*\ydisposition)}]
\begin{scope}[scale=\scaletree]
\Tree [
. \node[] (41) {$\treethr{4},1,\bot$};
]
\end{scope}
\node[above of=41] (CT41) {\normalsize $\TreeClock_4$};
\node[boundingBox,  fit=(41) (CT41)] {};

\node[above of=CT41, node distance=\evYDist] {\normalsize \underline{$e_3$}};

\end{scope}

\begin{scope}[shift={(3*\xdisposition,0*\ydisposition)}]
\begin{scope}[scale=\scaletree]
\Tree [
. \node[] (42) {$\treethr{4},2,\bot$};
]
\end{scope}
\node[above of=42] (CT42) {\normalsize $\TreeClock_4$};
\node[right of = CT42, node distance=25] (CTL21) {\normalsize ($\TreeClock_{\textcolor{\darkred}{\ell_2}}$)};
\node[above of=CT42, node distance=\evYDist] {\normalsize \underline{$e_4$}};

\node[boundingBox,  fit=(42) (CT42) (CTL21)] {};

\end{scope}

\begin{scope}[shift={(0*\xdisposition,-0.7*\ydisposition)}]
\begin{scope}[scale=\scaletree]
\Tree [
. \node[] (51) {$\treethr{5},1,\bot$};
]
\end{scope}
\node[above of=51] (CT51) {\normalsize $\TreeClock_5$};
\node[above of=CT51, node distance=\evYDist] {\normalsize \underline{$e_5$}};

\end{scope}

\begin{scope}[shift={(1*\xdisposition,-0.7*\ydisposition)}]
\begin{scope}[scale=\scaletree]
\Tree [
. \node[] (52) {$\treethr{5},2,\bot$};
]
\end{scope}
\node[above of=52] (CT52) {\normalsize $\TreeClock_5$};
\node[right of = CT52, node distance=25] (CTL31) {\normalsize ($\TreeClock_{\textcolor{\darkgreen}{\ell_3}}$)};
\node[above of=CT52, node distance=\evYDist] {\normalsize \underline{$e_6$}};

\end{scope}
\begin{scope}[shift={(2*\xdisposition,-0.7*\ydisposition)}]
\begin{scope}[scale=\scaletree]

\Tree [
. \node[] (31) {$\treethr{3},1,\bot$}; 
	[. \node[] {$\treethr{1},2,1$}; ]
]
\end{scope}
\node[above of=31] (CT31) {\normalsize $\TreeClock_3$};
\node[above of=CT31, node distance=\evYDist] {\normalsize \underline{$e_7$}};

\end{scope}

\begin{scope}[shift={(3*\xdisposition,-0.7*\ydisposition)}]
\begin{scope}[scale=\scaletree]
\Tree [
. \node[] (32) {$\treethr{3},2,\bot$};  
	[. \node[] {$\treethr{5},2,2$};] 
		[. \node[] {$\treethr{1},2,1$};
	]
]
\end{scope}
\node[above of=32] (CT32) {\normalsize $\TreeClock_3$};
\node[above of=CT32, node distance=\evYDist] {\normalsize \underline{$e_8$}};

\end{scope}

\begin{scope}[shift={(0*\xdisposition,-1.6*\ydisposition)}]
\begin{scope}[scale=\scaletree]
\Tree [
. \node[] (33) {$\treethr{3},3,\bot$};  
	[. \node[] {$\treethr{5},2,2$};] 
	[. \node[] {$\treethr{1},2,1$};]
]
\end{scope}
\node[above of=33] (CT33) {\normalsize $\TreeClock_3$};
\node[right of = CT33, node distance=25] (CTL32) {\normalsize ($\TreeClock_{\textcolor{\darkgreen}{\ell_3}}$)};
\node[above of=CT33, node distance=\evYDist] {\normalsize \underline{$e_9$}};

\end{scope}

\begin{scope}[shift={(1*\xdisposition,-1.6*\ydisposition)}]
\begin{scope}[scale=\scaletree]
\Tree [
. \node[] (34) {$\treethr{3},4,\bot$};  
	[. \node[] {$\treethr{5},2,2$};] 
	[. \node[] {$\treethr{1},2,1$};]
]
\end{scope}
\node[above of=34] (CT34) {\normalsize $\TreeClock_3$};
\node[right of = CT34, node distance=25] (CTL12) {\normalsize ($\TreeClock_{\textcolor{blue}{\ell_1}}$)};
\node[above of=CT34, node distance=\evYDist] {\normalsize \underline{$e_{10}$}};

\end{scope}

\begin{scope}[shift={(2*\xdisposition,-1.6*\ydisposition)}]
\begin{scope}[scale=\scaletree]
\Tree [
. \node[] (35) {$\treethr{4},3,\bot$}; 
%	[. \node[] {$\treethr{4},2,3$};]  
%	[. \node[] {$\treethr{5},2,2$};] 
%	[. \node[] {$\treethr{1},2,1$};]
]
\end{scope}
\node[above of=35] (CT43) {\normalsize $\TreeClock_4$};
\node[above of=CT43, node distance=\evYDist] {\normalsize \underline{$e_{11}$}};

\end{scope}

\begin{scope}[shift={(3*\xdisposition,-1.6*\ydisposition)}]
\begin{scope}[scale=\scaletree]
\Tree [
. \node[] (36) {$\treethr{4},4,\bot$}; 
%	[. \node[] {$\treethr{4},2,5$};]  
%	[. \node[] {$\treethr{5},2,2$};] 
%	[. \node[] {$\treethr{1},2,1$};]
]
\end{scope}
\node[above of=36] (CT44) {\normalsize $\TreeClock_4$};
\node[right of = CT44, node distance=25] (CTL22) {\normalsize ($\TreeClock_{\textcolor{\darkred}{\ell_2}}$)};
\node[above of=CT44, node distance=\evYDist] {\normalsize \underline{$e_{12}$}};

\end{scope}

\begin{scope}[shift={(0*\xdisposition,-2.5*\ydisposition)}]
\begin{scope}[scale=\scaletree]
\Tree [
. \node[] (21) {$\treethr{2},1,\bot$}; 
	[. \node[] {$\treethr{3},4,1$}; 
		[ .\node[] {$\treethr{5},2,2$};  ]  
		[ .\node[] {$\treethr{1},2,1$};  ]
	]
]
\end{scope}
\node[above of=21] (CT21) {\normalsize $\TreeClock_2$};
\node[above of=CT21, node distance=\evYDist] {\normalsize \underline{$e_{13}$}};

\end{scope}

\begin{scope}[shift={(1*\xdisposition,-2.5*\ydisposition)}]
\begin{scope}[scale=\scaletree]
\Tree [
. \node[] (22) {$\treethr{2},2,\bot$}; 
	[. \node[] {$\treethr{3},4,1$}; 
		[ .\node[] {$\treethr{5},2,2$};  ]  
		[ .\node[] {$\treethr{1},2,1$};  ]
	]
]
\end{scope}
\node[above of=22] (CT22) {\normalsize $\TreeClock_2$};
\node[right of = CT22, node distance=25] (CTL13) {\normalsize ($\TreeClock_{\textcolor{blue}{\ell_1}}$)};
\node[above of=CT22, node distance=\evYDist] {\normalsize \underline{$e_{14}$}};

\end{scope}

\begin{scope}[shift={(2*\xdisposition,-2.5*\ydisposition)}]
\begin{scope}[scale=\scaletree]
\Tree [
. \node[] (23) {$\treethr{2},3,\bot$}; 
	[. \node[] {$\treethr{4},3,5$}; ]
	[.	\node[] {$\treethr{3},6,3$}; 
		[ .\node[] {$\treethr{5},2,2$};  ]  
		[ .\node[] {$\treethr{1},2,1$};  ]
	]
]
\end{scope}
\node[above of=23] (CT23) {\normalsize $\TreeClock_2$};
\node[above of=CT23, node distance=\evYDist] {\normalsize \underline{$e_{15}$}};

\end{scope}

\begin{scope}[shift={(3*\xdisposition,-2.5*\ydisposition)}]
\begin{scope}[scale=\scaletree]
\Tree [
. \node[] (24) {$\treethr{2},4,\bot$}; 
	[. \node[] {$\treethr{4},3,5$}; ]
	[.	\node[] {$\treethr{3},6,3$}; 
		[ .\node[] {$\treethr{5},2,2$};  ]  
		[ .\node[] {$\treethr{1},2,1$};  ]
	]
]
\end{scope}
\node[above of=24] (CT24) {\normalsize $\TreeClock_2$};
\node[right of = CT24, node distance=25] (CTL23) {\normalsize ($\TreeClock_{\textcolor{\darkred}{\ell_2}}$)};
\node[above of=CT24, node distance=\evYDist] {\normalsize \underline{$e_{16}$}};

\end{scope}

\end{tikzpicture}
\caption{
The sequence of tree-clock updates after $\HB$ processes each event $\Event_i$ of $\Trace$.
Only the tree clock $\TreeClock_j$ of thread $j$ that performs $\Event_i$ is shown.
When the thread performs a lock-release the corresponding tree clock of the lock is mentioned in parenthesis.
}
\label{subfig:hb_example_cts}
\end{subfigure}
\caption{
Example run of $\HB$ and the updates on the corresponding clock-trees.
}
\label{fig:hb_example}
\end{figure*}

%% file: figures/fig_hb_example2.tex
\begin{figure*}
\small
\newcommand{\xdisposition}{3}
\newcommand{\ydisposition}{2.5}
\newcommand{\halfyOuter}{0.3}
\newcommand{\halfxOuter}{0.5}
\newcommand{\halfyInner}{0.2}
\newcommand{\halfxInner}{0.4}
\newcommand{\grayDark}{gray!50}
\newcommand{\grayLight}{gray!20}
\newcommand{\evYDist}{15}
\begin{subfigure}[b]{0.475\textwidth}
\centering
\begin{tikzpicture}[thick,sibling distance=1em, node distance=15,
every tree node/.style={thick, rectangle, rounded corners, minimum height=0.8mm, minimum width=5mm, inner sep=2pt},
boundingBox/.style={very thick, draw=none},
sibling distance=3pt,
level distance=20pt,
]
\tikzset{edge from parent/.append style={ thick}}
\pgfdeclarelayer{background}
\pgfdeclarelayer{foreground}
\pgfsetlayers{background,main,foreground}

\begin{scope}[shift={(0*\xdisposition,0*\ydisposition)}]

\Tree [
. \node[] (36) {$3,6,\bot$}; [. \node[] (a) {$4,2,5$};]  [. \node[] (b) {$5,2,2$};] [. \node[] (c) {$1,2,1$};]
]
\node[above of=36] (CTL22) {\normalsize $\TreeClock_{\textcolor{\darkred}{\ell_2}}$};

\begin{pgfonlayer}{background}
\draw[smooth, draw=none, rounded corners, fill=\grayLight] ($ (36) + (0, \halfyOuter) $) to ($ (36) + (-\halfxOuter, \halfyOuter) $) to ($ (a) + (-\halfxOuter, \halfyOuter) $) to ($ (a) + (-\halfxOuter, -\halfyOuter) $) to ($ (b) + (\halfxOuter, -\halfyOuter) $) to ($ (36) + (\halfxOuter, \halfyOuter) $) to ($ (36) + (0, \halfyOuter) $);

\draw[smooth, draw=none, rounded corners, fill=\grayDark] ($ (36) + (0, \halfyInner) $) to ($ (36) + (-\halfxInner, \halfyInner) $) to ($ (a) + (-\halfxInner, \halfyInner) $) to ($ (a) + (-\halfxInner, -\halfyInner) $) to ($ (a) + (\halfxInner, -\halfyInner) $) to ($ (a) + (\halfxInner, \halfyInner) $) to ($ (36) + (\halfxInner, -\halfyInner) $) to ($ (36) + (\halfxInner, \halfyInner) $) to ($ (36) + (0, \halfyInner) $);

\end{pgfonlayer}

\end{scope}

\begin{scope}[shift={(1*\xdisposition,0*\ydisposition)}]

\Tree [
. \node[] (23) {$2,3,\bot$}; [. \node[fill=\grayDark] {$3,4,1$};  [ .\node[] {$5,2,2$};  ]  [ .\node[] {$1,2,1$};  ]]
]

\node[above of=23] (CT23) {\normalsize $\TreeClock_2$};

\end{scope}

\end{tikzpicture}
\caption{
Details of the join operation due to $\Event_{15}$.
$\TreeClock_{\textcolor{\darkred}{\ell_2}}$ shows the tree clock of lock $\ell_2$.
$\TreeClock_2$ shows the tree clock of thread $\Thread_2$ right before performing the operation $\TreeClock_2.\FunctionCTJoin(\TreeClock_{\textcolor{\darkred}{\ell_2}})$.
Note that $\TreeClock_2$ is not yet aware of thread $\Thread_4$.
The result is $\TreeClock_2$ shown in \cref{subfig:hb_example_cts}.
}
\label{subfig:hb_example_joi1}
\end{subfigure}
\quad
\begin{subfigure}[b]{0.475\textwidth}
\centering
\begin{tikzpicture}[thick,sibling distance=1em, node distance=15,
every tree node/.style={thick, rectangle, rounded corners, minimum height=0.8mm, minimum width=5mm, inner sep=2pt},
boundingBox/.style={very thick, draw=none},
sibling distance=3pt,
level distance=20pt,
]
\tikzset{edge from parent/.append style={ thick}}
\pgfdeclarelayer{background}
\pgfdeclarelayer{foreground}
\pgfsetlayers{background,main,foreground}

\begin{scope}[shift={(0*\xdisposition,0*\ydisposition)}]

\Tree [
. \node[] (24) {$2,4,\bot$}; [. \node[] (w) {$3,6,3$}; [ .\node[] (x) {$4,2,5$};  ]  [ .\node[] (y) {$5,2,2$};  ]  [ .\node[] (z) {$1,2,1$};  ]]
]

\node[above of=24] (CT24) {\normalsize $\TreeClock_2$};

\begin{pgfonlayer}{background}

\draw[smooth, draw=none, rounded corners, fill=\grayLight] ($ (24) + (0, \halfyOuter) $) to ($ (24) + (-\halfxOuter, \halfyOuter) $) to  ($ (w) + (-\halfxOuter, -\halfyOuter) $) to ($ (w) + (\halfxOuter, -\halfyOuter) $) to ($ (24) + (\halfxOuter, \halfyOuter) $) to ($ (24) + (0, \halfyOuter) $);

\draw[smooth, draw=none, rounded corners, fill=\grayDark] ($ (24) + (0, \halfyInner) $) to ($ (24) + (-\halfxInner, \halfyInner) $) to  ($ (w) + (-\halfxInner, -\halfyInner) $) to ($ (w) + (\halfxInner, -\halfyInner) $) to ($ (24) + (\halfxInner, \halfyInner) $) to ($ (24) + (0, \halfyInner) $);

\end{pgfonlayer}

\end{scope}

\begin{scope}[shift={(1*\xdisposition,0*\ydisposition)}]

\Tree [
. \node[fill=\grayDark] (36) {$3,6,\bot$}; [. \node[] (a) {$4,2,5$};]  [. \node[] (b) {$5,2,2$};] [. \node[] (c) {$1,2,1$};]
]
\node[above of=36] (CTL22) {\normalsize $\TreeClock_{\textcolor{\darkred}{\ell_2}}$};

\end{scope}

\end{tikzpicture}
\caption{
Details of processing $\Event_{16}$.
$\TreeClock_2$ shows the tree clock of thread $\Thread_2$.
$\TreeClock_{\textcolor{\darkred}{\ell_2}}$ shows the tree clock of lock $\ell_2$ right before performing $\TreeClock_{\textcolor{\darkred}{\ell_2}}.\FunctionCTMonotoneCopy(\TreeClock_2)$.
Note that $\TreeClock_{\textcolor{\darkred}{\ell_2}}$ is not yet aware of thread $\Thread_2$.
The result is $\TreeClock_{\textcolor{\darkred}{\ell_2}}$ shown in \cref{subfig:hb_example_cts}.
}
\label{subfig:hb_example_join2}
\end{subfigure}
\caption{
A closer look at the last two steps of the example of \cref{fig:hb_example} in the processing of events $\Event_{15}$ and $\Event_{16}$.
In each figure, light gray marks the nodes of the left tree clock whose vector time is compared to the time of the tree clock on the right.
Similarly, dark gray marks the nodes of the left (resp., right) tree clock that are updating (resp., being updated).
}
\label{fig:hb_example2}
\end{figure*}
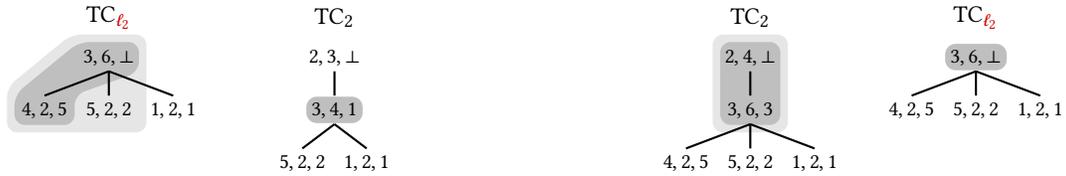

%% file: arxiv/threads_table.tex
\clearpage
\section{Benchmarks}\label{appsec:benchmarks}

\renewcommand{\captionsize}{\small}

\bottomcaption{
Information on Benchmark Traces. We denote by $\mathcal{N}$, $\mathcal{T}$, $\mathcal{M}$ and $\mathcal{L}$ the total number of events, number of threads, number of memory locations and number of locks, respectively. 
} \label{tab:trace-info}

\tablefirsthead{\toprule \textbf{Benchmark} & $\mathcal{N}$ & $\mathcal{T}$ & $\mathcal{M}$
& $\mathcal{L}$ \\ \midrule}

\tablehead{ \toprule \textbf{Benchmark} & $\mathcal{N}$ & $\mathcal{T}$ & $\mathcal{M}$
& $\mathcal{L}$ \\ \midrule}

\tabletail{\hline}
\tablelasttail{ \hline}
\scriptsize

\renewcommand{\arraystretch}{0.85}
\begin{xtabular}[h*]{rcccc}
CoMD-omp-task-1 & 175.1M & 56 & 66.1K & 114\\
CoMD-omp-task-2 & 175.1M & 56 & 66.1K & 114\\
CoMD-omp-task-deps-1 & 174.1M & 16 & 63.0K & 34\\
CoMD-omp-task-deps-2 & 175.1M & 56 & 66.1K & 114\\
CoMD-omp-taskloop-1 & 251.5M & 16 & 4.0M & 35\\
CoMD-omp-taskloop-2 & 251.5M & 56 & 4.0M & 115\\
CoMD-openmp-1 & 174.1M & 16 & 63.0K & 34\\
CoMD-openmp-2 & 175.1M & 56 & 66.1K & 114\\
DRACC-OMP-009-Counter-wrong-critical-yes & 135.0M & 16 & 971 & 36\\
DRACC-OMP-010-Counter-wrong-critical-Intra-yes & 135.0M & 16 & 971 & 36\\
DRACC-OMP-011-Counter-wrong-critical-Inter-yes & 135.0M & 16 & 853 & 21\\
DRACC-OMP-012-Counter-wrong-critical-simd-yes & 104.9M & 16 & 1.5K & 36\\
DRACC-OMP-013-Counter-wrong-critical-simd-Intra-yes & 104.9M & 16 & 1.5K & 36\\
DRACC-OMP-014-Counter-wrong-critical-simd-Inter-yes & 104.9M & 16 & 1.4K & 21\\
DRACC-OMP-015-Counter-wrong-lock-yes & 135.0M & 16 & 971 & 36\\
DRACC-OMP-016-Counter-wrong-lock-Intra-yes & 135.0M & 16 & 971 & 36\\
DRACC-OMP-017-Counter-wrong-lock-Inter-yes-1 & 135.0M & 16 & 853 & 21\\
DRACC-OMP-017-Counter-wrong-lock-Inter-yes-2 & 27.0M & 16 & 853 & 21\\
DRACC-OMP-018-Counter-wrong-lock-simd-yes & 104.9M & 16 & 1.5K & 36\\
DRACC-OMP-019-Counter-wrong-lock-simd-Intra-yes & 104.9M & 16 & 1.5K & 36\\
DRACC-OMP-020-Counter-wrong-lock-simd-Inter-yes & 104.9M & 16 & 1.4K & 21\\
DataRaceBench-DRB062-matrixvector2-orig-no-1 & 183.9M & 16 & 7.0K & 33\\
DataRaceBench-DRB062-matrixvector2-orig-no-2 & 193.2M & 56 & 9.0K & 113\\
DataRaceBench-DRB105-taskwait-orig-no-1 & 134.0M & 16 & 3.2K & 33\\
DataRaceBench-DRB105-taskwait-orig-no-2 & 134.0M & 56 & 9.7K & 113\\
DataRaceBench-DRB106-taskwaitmissing-orig-yes-1 & 134.0M & 16 & 3.6K & 33\\
DataRaceBench-DRB106-taskwaitmissing-orig-yes-2 & 134.0M & 56 & 11.0K & 113\\
DataRaceBench-DRB110-ordered-orig-no-1 & 120.0M & 16 & 775 & 36\\
DataRaceBench-DRB110-ordered-orig-no-2 & 120.0M & 56 & 2.3K & 116\\
DataRaceBench-DRB122-taskundeferred-orig-no-1 & 112.0M & 16 & 956 & 33\\
DataRaceBench-DRB122-taskundeferred-orig-no-2 & 112.0M & 56 & 3.0K & 113\\
DataRaceBench-DRB123-taskundeferred-orig-yes-1 & 112.0M & 16 & 1.2K & 33\\
DataRaceBench-DRB123-taskundeferred-orig-yes-2 & 112.0M & 56 & 3.7K & 113\\
DataRaceBench-DRB144-critical-missingreduction-orig-gpu-yes & 140.0M & 16 & 966 & 35\\
DataRaceBench-DRB148-critical1-orig-gpu-yes & 135.0M & 16 & 971 & 36\\
DataRaceBench-DRB150-missinglock1-orig-gpu-yes & 112.0M & 16 & 968 & 35\\
DataRaceBench-DRB152-missinglock2-orig-gpu-no & 112.0M & 16 & 968 & 35\\
DataRaceBench-DRB154-missinglock3-orig-gpu-no.c & 112.0M & 16 & 851 & 20\\
DataRaceBench-DRB155-missingordered-orig-gpu-no-1 & 50.0M & 16 & 2.0M & 36\\
DataRaceBench-DRB155-missingordered-orig-gpu-no-2 & 50.0M & 56 & 2.0M & 116\\
DataRaceBench-DRB176-fib-taskdep-no-1 & 1.6B & 17 & 12.4K & 33\\
DataRaceBench-DRB176-fib-taskdep-no-2 & 1.6B & 57 & 41.3K & 113\\
DataRaceBench-DRB176-fib-taskdep-no-3 & 618.3M & 17 & 11.7K & 33\\
DataRaceBench-DRB176-fib-taskdep-no-4 & 618.3M & 57 & 38.3K & 113\\
DataRaceBench-DRB176-fib-taskdep-no-5 & 90.2M & 17 & 9.8K & 33\\
DataRaceBench-DRB176-fib-taskdep-no-6 & 90.3M & 57 & 30.9K & 113\\
DataRaceBench-DRB177-fib-taskdep-yes-1 & 1.6B & 17 & 10.5K & 33\\
DataRaceBench-DRB177-fib-taskdep-yes-2 & 382.1M & 17 & 9.3K & 33\\
DataRaceBench-DRB177-fib-taskdep-yes-3 & 236.2M & 17 & 9.0K & 33\\
DataRaceBench-DRB177-fib-taskdep-yes-4 & 1.0B & 17 & 10.0K & 33\\
DataRaceBench-DRB177-fib-taskdep-yes-5 & 618.3M & 17 & 9.8K & 33\\
DataRaceBench-DRB177-fib-taskdep-yes-6 & 618.3M & 57 & 30.6K & 113\\
DataRaceBench-DRB177-fib-taskdep-yes-7 & 90.2M & 17 & 8.7K & 33\\
DataRaceBench-DRB177-fib-taskdep-yes-8 & 90.3M & 57 & 25.5K & 113\\
OmpSCR-v2.0-c-LUreduction-1 & 136.4M & 16 & 181.6K & 34\\
OmpSCR-v2.0-c-LUreduction-2 & 136.9M & 56 & 183.6K & 114\\
OmpSCR-v2.0-c-Mandelbrot-1 & 115.7M & 16 & 3.0K & 34\\
OmpSCR-v2.0-c-Mandelbrot-2 & 115.7M & 56 & 5.1K & 114\\
OmpSCR-v2.0-c-MolecularDynamic-1 & 204.3M & 16 & 7.2K & 34\\
OmpSCR-v2.0-c-MolecularDynamic-2 & 204.4M & 56 & 9.4K & 114\\
OmpSCR-v2.0-c-Pi-1 & 150.0M & 16 & 976 & 34\\
OmpSCR-v2.0-c-Pi-2 & 150.0M & 56 & 3.0K & 114\\
OmpSCR-v2.0-c-QuickSort-1 & 134.3M & 16 & 101.6K & 35\\
OmpSCR-v2.0-c-QuickSort-2 & 134.3M & 56 & 103.6K & 115\\
OmpSCR-v2.0-c-fft-2 & 2.1B & 57 & 29.4M & 115\\
OmpSCR-v2.0-c-fft-3 & 496.0M & 17 & 7.3M & 35\\
OmpSCR-v2.0-c-fft-4 & 496.0M & 57 & 7.3M & 115\\
OmpSCR-v2.0-c-fft6-1 & 146.0M & 16 & 4.3M & 49\\
OmpSCR-v2.0-c-fft6-2 & 146.0M & 56 & 5.3M & 132\\
OmpSCR-v2.0-c-testPath-1 & 30.2M & 16 & 2.8M & 35\\
OmpSCR-v2.0-c-testPath-2 & 37.5M & 56 & 2.8M & 115\\
OmpSCR-v2.0-c-LoopsWithDependencies-c-loopA.badSolution-1 & 112.6M & 16 & 161.0K & 34\\
OmpSCR-v2.0-c-LoopsWithDependencies-c-loopA.badSolution-2 & 394.0M & 56 & 563.0K & 114\\
OmpSCR-v2.0-c-LoopsWithDependencies-c-loopA.solution1-1 & 192.6M & 16 & 321.0K & 34\\
OmpSCR-v2.0-c-LoopsWithDependencies-c-loopA.solution1-2 & 674.2M & 56 & 1.1M & 114\\
OmpSCR-v2.0-c-LoopsWithDependencies-c-loopA.solution2-1 & 337.2M & 56 & 562.7K & 114\\
OmpSCR-v2.0-c-LoopsWithDependencies-c-loopA.solution2-2 & 96.4M & 16 & 160.9K & 34\\
OmpSCR-v2.0-c-LoopsWithDependencies-c-loopA.solution3-1 & 337.3M & 56 & 562.7K & 114\\
OmpSCR-v2.0-c-LoopsWithDependencies-c-loopA.solution3-2 & 96.4M & 16 & 160.9K & 34\\
OmpSCR-v2.0-c-LoopsWithDependencies-c-loopB.badSolution1-1 & 112.6M & 16 & 161.0K & 34\\
OmpSCR-v2.0-c-LoopsWithDependencies-c-loopB.badSolution1-2 & 394.0M & 56 & 563.0K & 114\\
OmpSCR-v2.0-cpp-sortOpenMP-cpp-qsomp1-1 & 107.0M & 16 & 101.2K & 35\\
OmpSCR-v2.0-cpp-sortOpenMP-cpp-qsomp1-2 & 106.8M & 56 & 103.5K & 115\\
OmpSCR-v2.0-cpp-sortOpenMP-cpp-qsomp2-1 & 107.1M & 56 & 104.7K & 115\\
OmpSCR-v2.0-cpp-sortOpenMP-cpp-qsomp2-2 & 107.5M & 16 & 101.5K & 35\\
OmpSCR-v2.0-cpp-sortOpenMP-cpp-qsomp3-1 & 115.4M & 56 & 10.0M & 115\\
OmpSCR-v2.0-cpp-sortOpenMP-cpp-qsomp3-2 & 141.7M & 16 & 10.0M & 35\\
OmpSCR-v2.0-cpp-sortOpenMP-cpp-qsomp4-1 & 114.2M & 56 & 104.7K & 115\\
OmpSCR-v2.0-cpp-sortOpenMP-cpp-qsomp4-2 & 164.0M & 16 & 101.5K & 35\\
OmpSCR-v2.0-cpp-sortOpenMP-cpp-qsomp6 & 106.7M & 56 & 104.8K & 115\\
OmpSCR-v2.0-cpp-sortOpenMP-cpp-qsomp7-1 & 295.5M & 56 & 82.9K & 114\\
OmpSCR-v2.0-cpp-sortOpenMP-cpp-qsomp7-2 & 88.9M & 16 & 31.6K & 34\\
HPCCG-1 & 228.1M & 16 & 935.9K & 34\\
HPCCG-2 & 229.5M & 56 & 938.4K & 114\\
graph500-1 & 171.3M & 16 & 228.8K & 34\\
graph500-2 & 172.5M & 56 & 233.4K & 114\\
SimpleMOC & 170.2M & 16 & 4.3M & 5.0K\\
NPBS-DC.S-1 & 11.7M & 16 & 3.6M & 36\\
NPBS-DC.S-2 & 11.7M & 56 & 3.7M & 115\\
NPBS-IS.W-1 & 152.9M & 16 & 2.2M & 34\\
NPBS-IS.W-2 & 300.1M & 56 & 2.3M & 114\\
Kripke-1 & 117.5M & 16 & 203.1K & 35\\
Kripke-2 & 119.2M & 56 & 205.8K & 116\\
Lulesh-1 & 35.3M & 17 & 159.5K & 33\\
Lulesh-2 & 52.1M & 56 & 164.0K & 114\\
Lulesh-3 & 543.4M & 17 & 1.3M & 34\\
Lulesh-4 & 569.5M & 57 & 1.2M & 114\\
QuickSilver & 132.6M & 56 & 914.1K & 116\\
RSBench-1 & 1.3B & 16 & 2.0M & 34\\
RSBench-2 & 1.3B & 56 & 2.0M & 114\\
XSBench-1 & 96.6M & 16 & 25.7M & 33\\
XSBench-2 & 96.6M & 56 & 25.7M & 114\\
amg2013-1 & 169.9M & 18 & 2.8M & 74\\
amg2013-2 & 189.6M & 58 & 4.5M & 154\\
miniFE-1 & 206.7M & 58 & 1.2M & 154\\
miniFE-2 & 207.7M & 18 & 1.2M & 74\\
account & 134 & 5 & 41 & 3\\
airlinetickets & 140 & 5 & 44 & 0\\
array & 51 & 4 & 30 & 2\\
batik & 157.9M & 7 & 4.9M & 1.9K\\
boundedbuffer & 332 & 3 & 63 & 2\\
bubblesort & 4.6K & 13 & 167 & 2\\
bufwriter & 22.2K & 7 & 471 & 1\\
clean & 1.3K & 10 & 26 & 2\\
critical & 59 & 5 & 30 & 0\\
cryptorsa & 58.5M & 9 & 1.7M & 8.0K\\
derby & 1.4M & 5 & 185.6K & 1.1K\\
ftpserver & 49.6K & 12 & 5.5K & 301\\
jigsaw & 3.1M & 12 & 103.5K & 275\\
lang & 6.3K & 8 & 1.5K & 0\\
linkedlist & 1.0M & 13 & 3.1K & 1.0K\\
lufact & 134.1M & 5 & 252.1K & 1\\
luindex & 397.8M & 3 & 2.5M & 65\\
lusearch & 217.5M & 8 & 5.2M & 118\\
mergesort & 3.0K & 6 & 621 & 3\\
moldyn & 200.3K & 4 & 1.2K & 2\\
pingpong & 151 & 7 & 51 & 0\\
producerconsumer & 658 & 9 & 67 & 3\\
raytracer & 15.8K & 4 & 3.9K & 8\\
readerswriters & 11.3K & 6 & 18 & 1\\
sor & 606.9M & 5 & 1.0M & 2\\
sunflow & 11.7M & 17 & 1.3M & 9\\
tsp & 307.3M & 10 & 181.1K & 2\\
twostage & 193 & 13 & 21 & 2\\
wronglock & 246 & 23 & 26 & 2\\
xalan & 122.5M & 7 & 4.4M & 2.5K\\
biojava & 221.0M & 4 & 121.0K & 78\\
cassandra & 259.1M & 173 & 9.9M & 60.5K\\
graphchi & 215.8M & 20 & 24.9M & 60\\
hsqldb & 18.8M & 44 & 945.0K & 401\\
tradebeans & 39.1M & 222 & 2.8M & 6.1K\\
tradesoap & 39.1M & 221 & 2.8M & 6.1K\\
zxing & 546.4M & 15 & 37.8M & 1.5K\\
\end{xtabular}

%% file: main.bbl
\begin{thebibliography}{10}

\bibitem{coral2}
{CORAL-2 Benchmarks}.
\newblock \url{https://asc.llnl.gov/coral-2-benchmarks}.
\newblock Accessed: 2021-08-01.

\bibitem{coral}
{CORAL Benchmarks}.
\newblock \url{https://asc.llnl.gov/coral-benchmarks}.
\newblock Accessed: 2021-08-01.

\bibitem{ecp}
{ECP Proxy Applications}.
\newblock \url{https://proxyapps.exascaleproject.org}.
\newblock Accessed: 2021-08-01.

\bibitem{mantevo}
{Mantevo Project}.
\newblock \url{https://mantevo.org}.
\newblock Accessed: 2021-08-01.

\bibitem{Agrawal2018}
Kunal Agrawal, Joseph Devietti, Jeremy~T. Fineman, I-Ting~Angelina Lee, Robert
  Utterback, and Changming Xu.
\newblock {Race Detection and Reachability in Nearly Series-Parallel DAGs}.
\newblock In {\em Proceedings of the Twenty-Ninth Annual ACM-SIAM Symposium on
  Discrete Algorithms}, SODA '18, page 156–171, USA, 2018. Society for
  Industrial and Applied Mathematics.
\newblock \href {https://doi.org/10.1137/1.9781611975031.11}
  {\path{doi:10.1137/1.9781611975031.11}}.

\bibitem{nasbenchmark}
D.~H. Bailey, E.~Barszcz, J.~T. Barton, D.~S. Browning, R.~L. Carter, L.~Dagum,
  R.~A. Fatoohi, P.~O. Frederickson, T.~A. Lasinski, R.~S. Schreiber, H.~D.
  Simon, V.~Venkatakrishnan, and S.~K. Weeratunga.
\newblock {The NAS Parallel Benchmarks—Summary and Preliminary Results}.
\newblock In {\em Proceedings of the 1991 ACM/IEEE Conference on
  Supercomputing}, Supercomputing '91, page 158–165, New York, NY, USA, 1991.
  Association for Computing Machinery.
\newblock \href {https://doi.org/10.1145/125826.125925}
  {\path{doi:10.1145/125826.125925}}.

\bibitem{Bertoni1989}
A.~Bertoni, G.~Mauri, and N.~Sabadini.
\newblock Membership problems for regular and context-free trace languages.
\newblock {\em Information and Computation}, 82(2):135--150, 1989.
\newblock \href {https://doi.org/10.1016/0890-5401(89)90051-5}
  {\path{doi:10.1016/0890-5401(89)90051-5}}.

\bibitem{Biswas14}
Swarnendu Biswas, Jipeng Huang, Aritra Sengupta, and Michael~D. Bond.
\newblock {DoubleChecker: Efficient Sound and Precise Atomicity Checking}.
\newblock In {\em Proceedings of the 35th ACM SIGPLAN Conference on Programming
  Language Design and Implementation}, PLDI '14, pages 28--39, New York, NY,
  USA, 2014. ACM.
\newblock \href {https://doi.org/10.1145/2594291.2594323}
  {\path{doi:10.1145/2594291.2594323}}.

\bibitem{Blackburn06}
Stephen~M. Blackburn, Robin Garner, Chris Hoffmann, Asjad~M. Khang, Kathryn~S.
  McKinley, Rotem Bentzur, Amer Diwan, Daniel Feinberg, Daniel Frampton,
  Samuel~Z. Guyer, Martin Hirzel, Antony Hosking, Maria Jump, Han Lee,
  J.~Eliot~B. Moss, Aashish Phansalkar, Darko Stefanovi\'{c}, Thomas VanDrunen,
  Daniel von Dincklage, and Ben Wiedermann.
\newblock {The DaCapo Benchmarks: Java Benchmarking Development and Analysis}.
\newblock In {\em OOPSLA}, 2006.
\newblock \href {https://doi.org/10.1145/1167515.1167488}
  {\path{doi:10.1145/1167515.1167488}}.

\bibitem{boehmbenign2011}
Hans-J. Boehm.
\newblock {How to Miscompile Programs with “Benign” Data Races}.
\newblock In {\em Proceedings of the 3rd USENIX Conference on Hot Topic in
  Parallelism}, HotPar’11, page~3, USA, 2011. USENIX Association.
\newblock URL: \url{http://dl.acm.org/citation.cfm?id=2001252.2001255}.

\bibitem{Bond2013}
Michael~D. Bond, Milind Kulkarni, Man Cao, Minjia Zhang, Meisam Fathi~Salmi,
  Swarnendu Biswas, Aritra Sengupta, and Jipeng Huang.
\newblock {OCTET: Capturing and Controlling Cross-Thread Dependences
  Efficiently}.
\newblock In {\em Proceedings of the 2013 ACM SIGPLAN International Conference
  on Object Oriented Programming Systems Languages \& Applications}, OOPSLA
  '13, page 693–712, New York, NY, USA, 2013. Association for Computing
  Machinery.
\newblock \href {https://doi.org/10.1145/2509136.2509519}
  {\path{doi:10.1145/2509136.2509519}}.

\bibitem{CharronBost1991}
Bernadette Charron-Bost.
\newblock Concerning the size of logical clocks in distributed systems.
\newblock {\em Information Processing Letters}, 39(1):11 -- 16, 1991.
\newblock \href {https://doi.org/10.1016/0020-0190(91)90055-M}
  {\path{doi:10.1016/0020-0190(91)90055-M}}.

\bibitem{cheng1998detecting}
Guang-Ien Cheng, Mingdong Feng, Charles~E. Leiserson, Keith~H. Randall, and
  Andrew~F. Stark.
\newblock {Detecting Data Races in Cilk Programs That Use Locks}.
\newblock In {\em Proceedings of the Tenth Annual ACM Symposium on Parallel
  Algorithms and Architectures}, SPAA '98, pages 298--309, New York, NY, USA,
  1998. ACM.
\newblock \href {https://doi.org/10.1145/277651.277696}
  {\path{doi:10.1145/277651.277696}}.

\bibitem{RADISH2012}
Joseph Devietti, Benjamin~P. Wood, Karin Strauss, Luis Ceze, Dan Grossman, and
  Shaz Qadeer.
\newblock {RADISH: Always-on Sound and Complete Race Detection in Software and
  Hardware}.
\newblock In {\em Proceedings of the 39th Annual International Symposium on
  Computer Architecture}, ISCA '12, page 201–212, USA, 2012. IEEE Computer
  Society.
\newblock \href {https://doi.org/10.1109/ISCA.2012.6237018}
  {\path{doi:10.1109/ISCA.2012.6237018}}.

\bibitem{Dimitrov2015}
Dimitar Dimitrov, Martin Vechev, and Vivek Sarkar.
\newblock {Race Detection in Two Dimensions}.
\newblock In {\em Proceedings of the 27th ACM Symposium on Parallelism in
  Algorithms and Architectures}, SPAA '15, page 101–110, New York, NY, USA,
  2015. Association for Computing Machinery.
\newblock \href {https://doi.org/10.1145/2755573.2755601}
  {\path{doi:10.1145/2755573.2755601}}.

\bibitem{doESE05}
Hyunsook Do, Sebastian~G. Elbaum, and Gregg Rothermel.
\newblock {Supporting Controlled Experimentation with Testing Techniques: An
  Infrastructure and its Potential Impact.}
\newblock {\em Empirical Software Engineering: An International Journal},
  10(4):405--435, 2005.
\newblock \href {https://doi.org/10.1007/s10664-005-3861-2}
  {\path{doi:10.1007/s10664-005-3861-2}}.

\bibitem{dorta2005openmp}
A.J. Dorta, C.~Rodriguez, and F.~de~Sande.
\newblock {The OpenMP source code repository}.
\newblock In {\em 13th Euromicro Conference on Parallel, Distributed and
  Network-Based Processing}, pages 244--250, 2005.
\newblock \href {https://doi.org/10.1109/EMPDP.2005.41}
  {\path{doi:10.1109/EMPDP.2005.41}}.

\bibitem{Elmas07}
Tayfun Elmas, Shaz Qadeer, and Serdar Tasiran.
\newblock {Goldilocks: A Race and Transaction-Aware Java Runtime}.
\newblock In {\em Proceedings of the 28th ACM SIGPLAN Conference on Programming
  Language Design and Implementation}, PLDI '07, pages 245--255, New York, NY,
  USA, 2007. ACM.
\newblock \href {https://doi.org/10.1145/1250734.1250762}
  {\path{doi:10.1145/1250734.1250762}}.

\bibitem{Farchi03}
Eitan Farchi, Yarden Nir, and Shmuel Ur.
\newblock {Concurrent Bug Patterns and How to Test Them}.
\newblock In {\em Proceedings of the 17th International Symposium on Parallel
  and Distributed Processing}, IPDPS '03, pages 286.2--, Washington, DC, USA,
  2003. IEEE Computer Society.
\newblock URL: \url{http://dl.acm.org/citation.cfm?id=838237.838485}.

\bibitem{Feng1997}
Mingdong Feng and Charles~E. Leiserson.
\newblock {Efficient Detection of Determinacy Races in Cilk Programs}.
\newblock In {\em Proceedings of the Ninth Annual ACM Symposium on Parallel
  Algorithms and Architectures}, SPAA '97, page 1–11, New York, NY, USA,
  1997. Association for Computing Machinery.
\newblock \href {https://doi.org/10.1145/258492.258493}
  {\path{doi:10.1145/258492.258493}}.

\bibitem{Fidge91}
Colin Fidge.
\newblock {Logical Time in Distributed Computing Systems}.
\newblock {\em Computer}, 24(8):28--33, August 1991.
\newblock \href {https://doi.org/10.1109/2.84874} {\path{doi:10.1109/2.84874}}.

\bibitem{fidge1988timestamps}
Colin~J. Fidge.
\newblock Timestamps in message-passing systems that preserve the partial
  ordering.
\newblock In {\em Proc. 11th Australian Comput. Science Conf.}, pages 56--66,
  1988.

\bibitem{Flanagan09}
Cormac Flanagan and Stephen~N. Freund.
\newblock {FastTrack: Efficient and Precise Dynamic Race Detection}.
\newblock In {\em Proceedings of the 30th ACM SIGPLAN Conference on Programming
  Language Design and Implementation}, PLDI '09, pages 121--133, New York, NY,
  USA, 2009. ACM.
\newblock \href {https://doi.org/10.1145/1542476.1542490}
  {\path{doi:10.1145/1542476.1542490}}.

\bibitem{redcard2013}
Cormac Flanagan and Stephen~N. Freund.
\newblock {RedCard: Redundant Check Elimination for Dynamic Race Detectors}.
\newblock In Giuseppe Castagna, editor, {\em ECOOP 2013 -- Object-Oriented
  Programming}, pages 255--280, Berlin, Heidelberg, 2013. Springer Berlin
  Heidelberg.
\newblock \href {https://doi.org/10.1007/978-3-642-39038-8_11}
  {\path{doi:10.1007/978-3-642-39038-8_11}}.

\bibitem{Flanagan2008}
Cormac Flanagan, Stephen~N. Freund, and Jaeheon Yi.
\newblock {Velodrome: {A} Sound and Complete Dynamic Atomicity Checker for
  Multithreaded Programs}.
\newblock In {\em Proceedings of the 29th ACM SIGPLAN Conference on Programming
  Language Design and Implementation}, PLDI '08, pages 293--303, New York, NY,
  USA, 2008. ACM.
\newblock \href {https://doi.org/10.1145/1375581.1375618}
  {\path{doi:10.1145/1375581.1375618}}.

\bibitem{Flanagan2005}
Cormac Flanagan and Patrice Godefroid.
\newblock Dynamic partial-order reduction for model checking software.
\newblock In {\em Proceedings of the 32nd ACM SIGPLAN-SIGACT Symposium on
  Principles of Programming Languages}, POPL '05, page 110–121, New York, NY,
  USA, 2005. Association for Computing Machinery.
\newblock \href {https://doi.org/10.1145/1040305.1040315}
  {\path{doi:10.1145/1040305.1040315}}.

\bibitem{Genc19}
Kaan Gen\c{c}, Jake Roemer, Yufan Xu, and Michael~D. Bond.
\newblock Dependence-aware, unbounded sound predictive race detection.
\newblock {\em Proc. ACM Program. Lang.}, 3(OOPSLA), October 2019.
\newblock \href {https://doi.org/10.1145/3360605} {\path{doi:10.1145/3360605}}.

\bibitem{Godefroid1996}
Patrice Godefroid, J.~van Leeuwen, J.~Hartmanis, G.~Goos, and Pierre Wolper.
\newblock {\em Partial-Order Methods for the Verification of Concurrent
  Systems: An Approach to the State-Explosion Problem}.
\newblock Springer-Verlag, Berlin, Heidelberg, 1996.

\bibitem{Itzkovitz1999}
Ayal Itzkovitz, Assaf Schuster, and Oren Zeev-Ben-Mordehai.
\newblock {Toward Integration of Data Race Detection in DSM Systems}.
\newblock {\em J. Parallel Distrib. Comput.}, 59(2):180--203, November 1999.
\newblock \href {https://doi.org/10.1006/jpdc.1999.1574}
  {\path{doi:10.1006/jpdc.1999.1574}}.

\bibitem{Kini17}
Dileep Kini, Umang Mathur, and Mahesh Viswanathan.
\newblock {Dynamic Race Prediction in Linear Time}.
\newblock In {\em Proceedings of the 38th ACM SIGPLAN Conference on Programming
  Language Design and Implementation}, PLDI 2017, pages 157--170, New York, NY,
  USA, 2017. ACM.
\newblock \href {https://doi.org/10.1145/3062341.3062374}
  {\path{doi:10.1145/3062341.3062374}}.

\bibitem{Kulkarni2021}
Rucha Kulkarni, Umang Mathur, and Andreas Pavlogiannis.
\newblock {Dynamic Data-Race Detection Through the Fine-Grained Lens}.
\newblock In Serge Haddad and Daniele Varacca, editors, {\em 32nd International
  Conference on Concurrency Theory (CONCUR 2021)}, volume 203 of {\em Leibniz
  International Proceedings in Informatics (LIPIcs)}, pages 16:1--16:23,
  Dagstuhl, Germany, 2021. Schloss Dagstuhl -- Leibniz-Zentrum f{\"u}r
  Informatik.
\newblock \href {https://doi.org/10.4230/LIPIcs.CONCUR.2021.16}
  {\path{doi:10.4230/LIPIcs.CONCUR.2021.16}}.

\bibitem{Lamport78}
Leslie Lamport.
\newblock {Time, Clocks, and the Ordering of Events in a Distributed System}.
\newblock {\em Commun. ACM}, 21(7):558--565, July 1978.
\newblock \href {https://doi.org/10.1145/359545.359563}
  {\path{doi:10.1145/359545.359563}}.

\bibitem{liao2017dataracebench}
Chunhua Liao, Pei-Hung Lin, Joshua Asplund, Markus Schordan, and Ian Karlin.
\newblock {DataRaceBench: A Benchmark Suite for Systematic Evaluation of Data
  Race Detection Tools}.
\newblock In {\em {Proceedings of the International Conference for High
  Performance Computing, Networking, Storage and Analysis}}, SC '17, New York,
  NY, USA, 2017. Association for Computing Machinery.
\newblock \href {https://doi.org/10.1145/3126908.3126958}
  {\path{doi:10.1145/3126908.3126958}}.

\bibitem{lpsz08}
Shan Lu, Soyeon Park, Eunsoo Seo, and Yuanyuan Zhou.
\newblock {Learning from Mistakes: A Comprehensive Study on Real World
  Concurrency Bug Characteristics}.
\newblock In {\em Proceedings of the 13th International Conference on
  Architectural Support for Programming Languages and Operating Systems},
  ASPLOS XIII, pages 329--339, New York, NY, USA, 2008. ACM.
\newblock \href {https://doi.org/10.1145/1346281.1346323}
  {\path{doi:10.1145/1346281.1346323}}.

\bibitem{Mathur18}
Umang Mathur, Dileep Kini, and Mahesh Viswanathan.
\newblock {What Happens-after the First Race? Enhancing the Predictive Power of
  Happens-before Based Dynamic Race Detection}.
\newblock {\em Proc. ACM Program. Lang.}, 2(OOPSLA):145:1--145:29, October
  2018.
\newblock \href {https://doi.org/10.1145/3276515} {\path{doi:10.1145/3276515}}.

\bibitem{Mathur2020b}
Umang Mathur, Andreas Pavlogiannis, and Mahesh Viswanathan.
\newblock {The Complexity of Dynamic Data Race Prediction}.
\newblock In {\em Proceedings of the 35th Annual ACM/IEEE Symposium on Logic in
  Computer Science}, LICS ’20, page 713–727, New York, NY, USA, 2020.
  Association for Computing Machinery.
\newblock \href {https://doi.org/10.1145/3373718.3394783}
  {\path{doi:10.1145/3373718.3394783}}.

\bibitem{Mathur21}
Umang Mathur, Andreas Pavlogiannis, and Mahesh Viswanathan.
\newblock Optimal prediction of synchronization-preserving races.
\newblock {\em Proc. ACM Program. Lang.}, 5(POPL), January 2021.
\newblock \href {https://doi.org/10.1145/3434317} {\path{doi:10.1145/3434317}}.

\bibitem{Mathur2020}
Umang Mathur and Mahesh Viswanathan.
\newblock {Atomicity Checking in Linear Time Using Vector Clocks}.
\newblock In {\em Proceedings of the Twenty-Fifth International Conference on
  Architectural Support for Programming Languages and Operating Systems},
  ASPLOS ’20, page 183–199, New York, NY, USA, 2020. Association for
  Computing Machinery.
\newblock \href {https://doi.org/10.1145/3373376.3378475}
  {\path{doi:10.1145/3373376.3378475}}.

\bibitem{Mattern89}
Friedemann Mattern.
\newblock Virtual time and global states of distributed systems.
\newblock In M.~Cosnard et. al., editor, {\em Parallel and Distributed
  Algorithms: proceedings of the International Workshop on Parallel \&
  Distributed Algorithms}, pages 215--226. Elsevier Science Publishers B. V.,
  1989.

\bibitem{Mazurkiewicz87}
Antoni Mazurkiewicz.
\newblock {Trace Theory}.
\newblock In {\em Advances in Petri Nets 1986, Part II on Petri Nets:
  Applications and Relationships to Other Models of Concurrency}, pages
  279--324. Springer-Verlag New York, Inc., 1987.
\newblock \href {https://doi.org/10.1007/3-540-17906-2_30}
  {\path{doi:10.1007/3-540-17906-2_30}}.

\bibitem{Musuvathi08}
Madanlal Musuvathi, Shaz Qadeer, Thomas Ball, Gerard Basler,
  Piramanayagam~Arumuga Nainar, and Iulian Neamtiu.
\newblock {Finding and Reproducing Heisenbugs in Concurrent Programs}.
\newblock In {\em Proceedings of the 8th USENIX Conference on Operating Systems
  Design and Implementation}, OSDI'08, pages 267--280, Berkeley, CA, USA, 2008.
  USENIX Association.
\newblock URL: \url{http://dl.acm.org/citation.cfm?id=1855741.1855760}.

\bibitem{Netzer1990}
Robert~H.B. Netzer and Barton~P. Miller.
\newblock {On the Complexity of Event Ordering for Shared-Memory Parallel
  Program Executions}.
\newblock In {\em In Proceedings of the 1990 International Conference on
  Parallel Processing}, pages 93--97, 1990.

\bibitem{OCallahan03}
Robert O'Callahan and Jong-Deok Choi.
\newblock {Hybrid Dynamic Data Race Detection}.
\newblock {\em SIGPLAN Not.}, 38(10):167--178, June 2003.
\newblock \href {https://doi.org/10.1145/966049.781528}
  {\path{doi:10.1145/966049.781528}}.

\bibitem{Pavlogiannis2020}
Andreas Pavlogiannis.
\newblock {Fast, Sound, and Effectively Complete Dynamic Race Prediction}.
\newblock {\em Proc. ACM Program. Lang.}, 4(POPL), December 2019.
\newblock \href {https://doi.org/10.1145/3371085} {\path{doi:10.1145/3371085}}.

\bibitem{Pozniansky03}
Eli Pozniansky and Assaf Schuster.
\newblock {Efficient On-the-fly Data Race Detection in Multithreaded C++
  Programs}.
\newblock {\em SIGPLAN Not.}, 38(10):179--190, June 2003.
\newblock \href {https://doi.org/10.1145/966049.781529}
  {\path{doi:10.1145/966049.781529}}.

\bibitem{Raman2012}
Raghavan Raman, Jisheng Zhao, Vivek Sarkar, Martin Vechev, and Eran Yahav.
\newblock Efficient data race detection for async-finish parallelism.
\newblock {\em Formal Methods in System Design}, 41(3):321--347, Dec 2012.
\newblock \href {https://doi.org/10.1007/s10703-012-0143-7}
  {\path{doi:10.1007/s10703-012-0143-7}}.

\bibitem{Raychev2013}
Veselin Raychev, Martin Vechev, and Manu Sridharan.
\newblock {Effective Race Detection for Event-Driven Programs}.
\newblock In {\em Proceedings of the 2013 ACM SIGPLAN International Conference
  on Object Oriented Programming Systems Languages \& Applications}, OOPSLA
  '13, page 151–166, New York, NY, USA, 2013. Association for Computing
  Machinery.
\newblock \href {https://doi.org/10.1145/2509136.2509538}
  {\path{doi:10.1145/2509136.2509538}}.

\bibitem{bigfoot2017}
Dustin Rhodes, Cormac Flanagan, and Stephen~N. Freund.
\newblock {BigFoot: Static Check Placement for Dynamic Race Detection}.
\newblock In {\em Proceedings of the 38th ACM SIGPLAN Conference on Programming
  Language Design and Implementation}, PLDI 2017, page 141–156, New York, NY,
  USA, 2017. Association for Computing Machinery.
\newblock \href {https://doi.org/10.1145/3062341.3062350}
  {\path{doi:10.1145/3062341.3062350}}.

\bibitem{Roemer18}
Jake Roemer, Kaan Gen\c{c}, and Michael~D. Bond.
\newblock {High-coverage, Unbounded Sound Predictive Race Detection}.
\newblock In {\em Proceedings of the 39th ACM SIGPLAN Conference on Programming
  Language Design and Implementation}, PLDI 2018, pages 374--389, New York, NY,
  USA, 2018. ACM.
\newblock \href {https://doi.org/10.1145/3192366.3192385}
  {\path{doi:10.1145/3192366.3192385}}.

\bibitem{Roemer20}
Jake Roemer, Kaan Gen\c{c}, and Michael~D. Bond.
\newblock {SmartTrack: Efficient Predictive Race Detection}.
\newblock In {\em Proceedings of the 41st ACM SIGPLAN Conference on Programming
  Language Design and Implementation}, PLDI 2020, page 747–762, New York, NY,
  USA, 2020. Association for Computing Machinery.
\newblock \href {https://doi.org/10.1145/3385412.3385993}
  {\path{doi:10.1145/3385412.3385993}}.

\bibitem{rvpredict}
Grigore Rosu.
\newblock {{RV-Predict, Runtime Verification}}.
\newblock \url{https://runtimeverification.com/predict/}, 2021.
\newblock Accessed: 2021-08-01.

\bibitem{sadowski-tools-2014}
Caitlin Sadowski and Jaeheon Yi.
\newblock {How Developers Use Data Race Detection Tools}.
\newblock In {\em Proceedings of the 5th Workshop on Evaluation and Usability
  of Programming Languages and Tools}, PLATEAU '14, page 43–51, New York, NY,
  USA, 2014. Association for Computing Machinery.
\newblock \href {https://doi.org/10.1145/2688204.2688205}
  {\path{doi:10.1145/2688204.2688205}}.

\bibitem{Samak2014}
Malavika Samak and Murali~Krishna Ramanathan.
\newblock {Trace Driven Dynamic Deadlock Detection and Reproduction}.
\newblock In {\em Proceedings of the 19th ACM SIGPLAN Symposium on Principles
  and Practice of Parallel Programming}, PPoPP '14, page 29–42, New York, NY,
  USA, 2014. Association for Computing Machinery.
\newblock \href {https://doi.org/10.1145/2555243.2555262}
  {\path{doi:10.1145/2555243.2555262}}.

\bibitem{schmitz2019dataraceonaccelerator}
Adrian Schmitz, Joachim Protze, Lechen Yu, Simon Schwitanski, and Matthias~S.
  M{\"u}ller.
\newblock {DataRaceOnAccelerator -- A Micro-benchmark Suite for Evaluating
  Correctness Tools Targeting Accelerators}.
\newblock In Ulrich Schwardmann, Christian Boehme, Dora B.~Heras, Valeria
  Cardellini, Emmanuel Jeannot, Antonio Salis, Claudio Schifanella, Ravi~Reddy
  Manumachu, Dieter Schwamborn, Laura Ricci, Oh~Sangyoon, Thomas Gruber, Laura
  Antonelli, and Stephen~L. Scott, editors, {\em Euro-Par 2019: Parallel
  Processing Workshops}, pages 245--257, Cham, 2020. Springer International
  Publishing.
\newblock \href {https://doi.org/10.1007/978-3-030-48340-1_19}
  {\path{doi:10.1007/978-3-030-48340-1_19}}.

\bibitem{schwarz1994detecting}
Reinhard Schwarz and Friedemann Mattern.
\newblock {Detecting causal relationships in distributed computations: In
  search of the holy grail}.
\newblock {\em Distributed computing}, 7(3):149--174, 1994.
\newblock \href {https://doi.org/10.1007/BF02277859}
  {\path{doi:10.1007/BF02277859}}.

\bibitem{threadsanitizer}
Konstantin Serebryany and Timur Iskhodzhanov.
\newblock {ThreadSanitizer: Data Race Detection in Practice}.
\newblock In {\em Proceedings of the Workshop on Binary Instrumentation and
  Applications}, WBIA '09, page 62–71, New York, NY, USA, 2009. Association
  for Computing Machinery.
\newblock \href {https://doi.org/10.1145/1791194.1791203}
  {\path{doi:10.1145/1791194.1791203}}.

\bibitem{Shasha1988}
Dennis Shasha and Marc Snir.
\newblock {Efficient and Correct Execution of Parallel Programs That Share
  Memory}.
\newblock {\em ACM Trans. Program. Lang. Syst.}, 10(2):282–312, April 1988.
\newblock \href {https://doi.org/10.1145/42190.42277}
  {\path{doi:10.1145/42190.42277}}.

\bibitem{Shi10}
Yao Shi, Soyeon Park, Zuoning Yin, Shan Lu, Yuanyuan Zhou, Wenguang Chen, and
  Weimin Zheng.
\newblock {Do I Use the Wrong Definition?: DeFuse: Definition-use Invariants
  for Detecting Concurrency and Sequential Bugs}.
\newblock In {\em Proceedings of the ACM International Conference on Object
  Oriented Programming Systems Languages and Applications}, OOPSLA '10, pages
  160--174, New York, NY, USA, 2010. ACM.
\newblock \href {https://doi.org/10.1145/1869459.1869474}
  {\path{doi:10.1145/1869459.1869474}}.

\bibitem{Smaragdakis12}
Yannis Smaragdakis, Jacob Evans, Caitlin Sadowski, Jaeheon Yi, and Cormac
  Flanagan.
\newblock {Sound Predictive Race Detection in Polynomial Time}.
\newblock In {\em Proceedings of the 39th Annual ACM SIGPLAN-SIGACT Symposium
  on Principles of Programming Languages}, POPL '12, pages 387--400, New York,
  NY, USA, 2012. ACM.
\newblock \href {https://doi.org/10.1145/2103656.2103702}
  {\path{doi:10.1145/2103656.2103702}}.

\bibitem{Smith01}
L.~A. Smith, J.~M. Bull, and J.~Obdrz\'{a}lek.
\newblock {A Parallel Java Grande Benchmark Suite}.
\newblock In {\em Proceedings of the 2001 ACM/IEEE Conference on
  Supercomputing}, SC '01, pages 8--8, New York, NY, USA, 2001. ACM.
\newblock \href {https://doi.org/10.1145/582034.582042}
  {\path{doi:10.1145/582034.582042}}.

\bibitem{Sulzmann2018}
Martin Sulzmann and Kai Stadtm\"{u}ller.
\newblock Two-phase dynamic analysis of message-passing go programs based on
  vector clocks.
\newblock In {\em Proceedings of the 20th International Symposium on Principles
  and Practice of Declarative Programming}, PPDP '18, New York, NY, USA, 2018.
  Association for Computing Machinery.
\newblock \href {https://doi.org/10.1145/3236950.3236959}
  {\path{doi:10.1145/3236950.3236959}}.

\bibitem{surendran2016dynamic}
Rishi Surendran and Vivek Sarkar.
\newblock {Dynamic Determinacy Race Detection for Task Parallelism with
  Futures}.
\newblock In {\em International Conference on Runtime Verification}, pages
  368--385. Springer, 2016.
\newblock \href {https://doi.org/10.1007/978-3-319-46982-9_23}
  {\path{doi:10.1007/978-3-319-46982-9_23}}.

\bibitem{Tu2019}
Tengfei Tu, Xiaoyu Liu, Linhai Song, and Yiying Zhang.
\newblock {Understanding Real-World Concurrency Bugs in Go}.
\newblock In {\em Proceedings of the Twenty-Fourth International Conference on
  Architectural Support for Programming Languages and Operating Systems},
  ASPLOS '19, page 865–878, New York, NY, USA, 2019. Association for
  Computing Machinery.
\newblock \href {https://doi.org/10.1145/3297858.3304069}
  {\path{doi:10.1145/3297858.3304069}}.

\bibitem{Wang2006}
Xinli Wang, J.~Mayo, W.~Gao, and J.~Slusser.
\newblock {An Efficient Implementation of Vector Clocks in Dynamic Systems}.
\newblock In {\em PDPTA}, 2006.

\bibitem{Wood2017}
Benjamin~P. Wood, Man Cao, Michael~D. Bond, and Dan Grossman.
\newblock Instrumentation bias for dynamic data race detection.
\newblock {\em Proc. ACM Program. Lang.}, 1(OOPSLA), October 2017.
\newblock \href {https://doi.org/10.1145/3133893} {\path{doi:10.1145/3133893}}.

\bibitem{Yu2021}
Kunpeng Yu, Chenxu Wang, Yan Cai, Xiapu Luo, and Zijiang Yang.
\newblock {Detecting Concurrency Vulnerabilities Based on Partial Orders of
  Memory and Thread Events}.
\newblock In {\em Proceedings of the 29th ACM Joint Meeting on European
  Software Engineering Conference and Symposium on the Foundations of Software
  Engineering}, ESEC/FSE 2021, page 280–291, New York, NY, USA, 2021.
  Association for Computing Machinery.
\newblock \href {https://doi.org/10.1145/3468264.3468572}
  {\path{doi:10.1145/3468264.3468572}}.

\bibitem{Yu05}
Yuan Yu, Tom Rodeheffer, and Wei Chen.
\newblock {RaceTrack: Efficient Detection of Data Race Conditions via Adaptive
  Tracking}.
\newblock {\em SIGOPS Oper. Syst. Rev.}, 39(5):221--234, October 2005.
\newblock \href {https://doi.org/10.1145/1095809.1095832}
  {\path{doi:10.1145/1095809.1095832}}.

\bibitem{SoftwareErrors2009}
M.~Zhivich and R.~K. Cunningham.
\newblock {The Real Cost of Software Errors}.
\newblock {\em IEEE Security and Privacy}, 7(2):87–90, March 2009.
\newblock \href {https://doi.org/10.1109/MSP.2009.56}
  {\path{doi:10.1109/MSP.2009.56}}.

\bibitem{HARD2007}
P.~{Zhou}, R.~{Teodorescu}, and Y.~{Zhou}.
\newblock {HARD: Hardware-Assisted Lockset-based Race Detection}.
\newblock In {\em 2007 IEEE 13th International Symposium on High Performance
  Computer Architecture}, pages 121--132, 2007.
\newblock \href {https://doi.org/10.1109/HPCA.2007.346191}
  {\path{doi:10.1109/HPCA.2007.346191}}.

\end{thebibliography}
